\documentclass[12pt]{article}
\usepackage{amsmath, amssymb, amsthm}
\usepackage{graphicx}
\usepackage{enumerate}
\usepackage{url} 
\usepackage{setspace}
\usepackage{bm}
\usepackage{threeparttable}
\usepackage{graphicx}
\usepackage{enumerate}
\usepackage{bibentry}
\usepackage{booktabs}
\usepackage{tabularx}		
\usepackage{commath}			
\usepackage{babel}
\usepackage[titletoc,title]{appendix}
\usepackage[breaklinks,colorlinks=true,pagebackref=false,pdfauthor={},pdftitle={},citecolor=blue,linkcolor=blue,urlcolor=blue]{hyperref}
\usepackage{cleveref}
\usepackage{natbib}

\allowdisplaybreaks

\newcommand{\blind}{1}

\newtheorem{lem}{Lemma}

\newtheorem{assumption}{Assumption}
\newtheorem{prop}{Proposition}

\newtheorem{thm}{Theorem}

\newtheoremstyle{remark2}{1ex}{1ex}%
      {}
      {}
      {\bf}
      {.}
      {5pt}
      {\thmname{#1}\thmnumber{ #2}\thmnote{ \slshape{(#3)}}} 
\theoremstyle{remark2} 
\newtheorem{rem}{Remark}

\newcommand{\vtheta}{\bm \theta}
\newcommand{\E}{\operatorname{E}}
\newcommand{\p}{\operatorname{P}}
\newcommand{\LB}{\operatorname{LB}}
\newcommand{\Var}{\operatorname{Var}}
\newcommand{\Cov}{\operatorname{Cov}}
\newcommand{\VaR}{\operatorname{VaR}}

\newcommand{\ARMA}{\operatorname{ARMA}}
\newcommand{\GARCH}{\operatorname{GARCH}}

\newcommand{\sign}{\operatorname{sgn}}
\newcommand{\D}{\,\mathrm{d}}

\newcommand{\vxi}{\bm \xi}
\newcommand{\vlambda}{\bm \lambda}
\newcommand{\vx}{\bm x}
\newcommand{\vzeros}{\boldsymbol{0}}
\newcommand{\mTheta}{\bm \varTheta}
\newcommand{\mB}{\bm B}
\newcommand{\mI}{\bm I}
\newcommand{\mM}{\bm M}
\newcommand{\mZ}{\bm Z}
\newcommand{\s}{s}
\newcommand{\m}{m}

\begin{document}

\def\spacingset#1{\renewcommand{\baselinestretch}%
{#1}\small\normalsize} \spacingset{1}



\if1\blind
{
  \title{\bf Extremal Dependence-Based Specification Testing of Time Series}
  \author{Yannick Hoga\thanks{
	 The author is grateful to seminar participants at CREST and Erasmus University Rotterdam for valuable comments and suggestions, in particular Christian Francq, Jeroen Rombouts, Jean-Michel Zako\"{i}an and Chen Zhou. The author would also like to thank Christoph Hanck and Till Massing for their careful reading of an earlier version of this manuscript. This work was supported by the German Research Foundation (DFG) under Grant HO 6305/1-1.}\hspace{.2cm}\\
	Faculty of Economics and Business Administration\\University of Duisburg-Essen
}
  \maketitle
} \fi

\if0\blind
{
  \bigskip
  \bigskip
  \bigskip
  \begin{center}
    {\LARGE\bf Extremal Dependence-Based Specification Testing of Time Series}
\end{center}
  \medskip
} \fi

\onehalfspacing

\bigskip
\begin{abstract}
\noindent We propose a specification test for conditional location--scale models based on extremal dependence properties of the standardized residuals. We do so comparing the left-over serial extremal dependence---as measured by the pre-asymptotic tail copula---with that arising under serial independence at different lags. Our main theoretical results show that the proposed Portmanteau-type test statistics have nuisance parameter-free asymptotic limits. The test statistics are easy to compute, as they only depend on the standardized residuals, and critical values are likewise easily obtained from the limiting distributions. This contrasts with extant tests (based, e.g., on autocorrelations of squared residuals), where test statistics depend on the parameter estimator of the model and critical values may need to be bootstrapped. We show that our test performs well in simulations. An empirical application to S\&P~500 constituents illustrates that our test can uncover violations of residual serial independence that are not picked up by standard autocorrelation-based specification tests, yet are relevant when the model is used for, e.g., risk forecasting.
\end{abstract}

\noindent%
\vfill

\newpage

\section{Motivation}

\onehalfspacing

The dynamics of many economic and financial time series can be successfully captured by location--scale models, that allow for mean and variance changes in the conditional distribution. The benchmark models for incorporating mean changes are ARMA models, and GARCH-type processes have become the standard volatility models, since the seminal work of \citet{Eng82} and \citet{Bol86}. The main aim of this paper is to derive tests that check whether the fitted location--scale model adequately captures the time variation in the mean and variance of the conditional distribution. As is common for specification tests, a rejection points to a problem in the mean \textit{or} variance dynamics, but does not establish which of the two.\footnote{A rare exception is the work of \citet{FRZ06}.} While our tests may also pick up misspecification in the conditional mean, they are first and foremost designed to detect misspecified volatility dynamics. Thus, our main focus in the following will be on volatility models.

Testing for correctly specified volatility models is of high practical relevance. This is because volatility models are widely used to predict risk measures, such as the Value-at-Risk (VaR) and the Expected Shortfall (ES); see, e.g., \citet{MF00,Cea07,Hog18+}. Risk measure forecasts are key inputs in risk management procedures of financial institutions, due to regulatory requirements in the Basel framework of the \citet{BCBSBF19}. The Basel framework penalizes sustained underpredictions of risk by imposing higher capital requirements. On the other hand, if risk forecasts are too high, too much capital is put aside as a buffer against large losses. In both cases of over- and underprediction, some portion of the capital can no longer earn premiums, leading to foregone profits. Thus, accurate risk forecasts from correctly specified volatility models are important. Consequently, diagnostic tests should be routinely applied to the chosen volatility specification. 

To this end, several specification tests are available. These tests typically verify whether the model residuals are independent, identically distributed (i.i.d.). However, standard tests of the i.i.d.~property---such as \citet{BP70} or \citet{LB78} tests---cannot be applied `as usual', since the residuals are only estimated. \citet{LM94} were the first to propose a Portmanteau-type test for the autocorrelations of squared residuals that corrects for this fact in (conditionally Gaussian) location--scale models. \citet{BHK03b} and \citet{LL97} extend the applicability of the \citet{LM94} test to more general GARCH and $N(0,1)$--FARIMA--GARCH models, respectively. Similarly, \citet{FG12} consider weighted Portmanteau statistics applied to ARCH residuals. \citet{HZ07} develop spectral-based tests for ARCH($\infty$) series. The limiting distributions for all these proposed test statistics have only been derived for the quasi-maximum likelihood (QML) estimator. Moreover, practical application of these tests is complicated by the fact that the limiting distributions depend on nuisance parameters induced by parameter estimation. Valid tests thus require consistent estimators of the nuisance parameters (that have to be computed on a case-by-case basis) or involved bootstrap-based procedures to compute critical values. By construction, all these tests diagnose the complete serial dependence structure of the standardized residuals.

Here, instead of considering the complete dependence structure, we take a different tack by focusing on the left-over serial \textit{extremal} dependence in the standardized residuals. When diagnosing the complete dependence structure, the effect of any remaining residual extremal dependence may be `washed-out'. 
However, overlooked serial dependence in the extremes may be very harmful as it invalidates, e.g., the model's risk forecasts, such as VaR and ES forecasts \citep{Sea21+}. As pointed out above, misspecified risk forecasts are penalized under the Basel framework. Hence, a separate diagnostic for the extremes is desirable. In empirical work, \citet{DMC12} and \citet{DMZ13} use the pre-asymptotic (PA) extremogram as a diagnostic. For instance, to assess the GARCH fit for FTSE returns, \citet{DMC12} show that the PA-extremogram of the standardized residuals exhibits no signs of extremal dependence. However, such a model check lacks theoretical justification so far, because the central limit theory for extremal dependence measures has not been developed for model residuals. It is the main theoretical contribution of this paper to do so.

Specifically, we base our tests on the PA-tail copula \citep{SS06} of the (absolute values of the) standardized residuals. The tail copula---i.e., the limit of the PA-tail copula---has been used by, e.g., \citet{BJW15} to detect structural changes in tail dependence. The PA-tail copula is closely related to the PA-extremogram of \citet{DMC12}. However, for our purposes, working with the PA-tail copula gives rise to nuisance parameter-free limiting distributions, which would not be the case for the PA-extremogram.

It turns out that---unlike the above mentioned traditional specification tests---our tests are easy to apply: We only require the standardized residuals to compute the test statistics. In particular, it does not matter for our tests which parameter estimator (QML estimator, Self-Weighted QML estimator, Least Squares, etc.) was used in fitting the volatility model. This gives our tests wide practical appeal.

Testing for any left-over serial extremal dependence in the residuals directs the power of our specification tests to the tails. Thus, they are more sensitive to any remaining dependence in the tails. This may be important when, say, a GARCH model captures well the volatility dynamics in the body of the distribution, but not so much in the extremes (see also the simulations for such an example). It is exactly for these cases that our tests are designed. Thus, the diagnostic tests developed here are not meant to replace standard tests based (e.g.) on the autocorrelations of the squared residuals, but rather to supplement them.

We illustrate the good size and power of our tests in simulations. Specifically, we show that---as predicted by theory---in sufficiently large samples the size is indeed unaffected by parameter estimation effects. We also demonstrate that empirically relevant misspecifications can be detected more easily using our extremal dependence-based tests instead of autocorrela-tion-based tests.

The empirical application shows that for the S\&P~500 components an autocorrelation-based test can often not detect any significant departures from a fitted GARCH-type model. However, in many of these cases, applying our tests suggests some statistically significant serial extremal dependence left over in the residuals. This indicates that the fitted GARCH-type processes capture well the volatility changes in the center of the distribution, but not so much in the tails. This result is consistent with \citet{EM04}, who conclude from fitting their CAViaR models at different risk levels that `our findings suggest that the process governing the behavior of the tails might be different from that of the rest of the distribution.' Such a dynamic is particularly worrisome when using GARCH-type models to forecast risk measures, such as VaR and ES, far out in the tail.

The remainder of the paper proceeds as follows. In Section~\ref{Main results}, we present our model setup and propose our Portmanteau-type tests for serial extremal dependence. Section~\ref{Simulations} explores size and power of our tests in finite samples, while Section~\ref{Application} illustrates the benefits of our test for real data. Finally, Section~\ref{Conclusion} concludes. The supplementary Appendix contains all proofs and additional simulations.

\section{Main Results}\label{Main results}

\subsection{Location--Scale Model}\label{Location--Scale Model}

Denote by $Y_1,\ldots, Y_n$ the observations of interest to which some parametric model will be fitted. For instance, the $Y_1,\ldots, Y_n$ may denote log-returns on some speculative asset. We define the $\sigma$-field generated by $Y_t,Y_{t-1},\ldots$ and some exogenous (possibly multivariate) variables $\vx_{t},\vx_{t-1},\ldots$ by $\mathcal{F}_t=\sigma(Y_t,Y_{t-1},\ldots;\vx_{t},\vx_{t-1},\ldots)$. We assume the $Y_t$ to follow the parametric conditional location--scale model
\begin{equation}\label{eq:ls model}
	Y_t=\mu_{t}(\vtheta^\circ)+\sigma_t(\vtheta^\circ)\varepsilon_t,
\end{equation}
where $\vtheta^\circ\in\mTheta$ is the true parameter vector from the parameter space $\mTheta\subset\mathbb{R}^{m}$ ($m\in\mathbb{N}$), $\mu_{t}(\vtheta^\circ)$ and $\sigma_t(\vtheta^\circ)$ are the $\mathcal{F}_{t-1}$-measurable conditional mean and conditional standard deviation, and $\varepsilon_t$ is i.i.d.~with mean zero and unit variance (written: $\varepsilon_t\overset{\text{i.i.d.}}{\sim}(0,1)$), independent of $\mathcal{F}_{t-1}$.

In the following, we require a smoothness condition on the tail of the $|\varepsilon_t|$, whose distribution function (d.f.) we denote by $F(\cdot)$.

\begin{assumption}\label{ass:U distr 2}
The d.f.~$F(\cdot)$ of the $|\varepsilon_t|$ satisfies:
\begin{enumerate}
	\item[(i)] The upper endpoint of $F(\cdot)$ is infinite.
	\item[(ii)] There exists some constant $C_F>0$ such that $F(\cdot)$ is differentiable with density $f(\cdot)$ on $[C_F,\infty)$.
	\item[(iii)] $\lim_{x\to\infty}xf(x)/[1-F(x)]=\alpha\in(0,\infty)$.
\end{enumerate}
\end{assumption}

Assumption~\ref{ass:U distr 2} is known as one of the \textit{von Mises'} conditions \citep[see, e.g.,][Rem.~1.1.15]{HF06} and ensures a sufficiently well-behaved tail of the $|\varepsilon_t|$. This is needed to justify the replacement of $|\varepsilon_t|$ in our test statistic with the standardized residuals. In their autocorrelation-based diagnostic test, \citet{BHK03b} impose a similar regularity condition; see in particular their equation~(2.4). Assumption~\ref{ass:U distr 2} is only needed to ensure the validity of Lemma~\ref{lem:Lem F} in Appendix~\ref{Proofs of Propositions}. Thus, it may be replaced by any other assumption on the tail decay, as long as the conclusions of Lemma~\ref{lem:Lem F} remain valid. 

Next, we require a $\sqrt{n}$-consistent estimator of the unknown $\vtheta^{\circ}$:

\begin{assumption}\label{ass:estimator}
There exists an estimator $\widehat{\vtheta}$ satisfying $n^{1/2}|\widehat{\vtheta}-\vtheta^{\circ}|=O_{\p}(1)$, as $n\to\infty$.
\end{assumption}

Root-$n$ consistent estimators are available under various regularity conditions for ARMA--GARCH and GARCH-type models; see, e.g., \citet{BH04}, \citet{FZ04}, \citet{PWT08}, \citet{FT19}.

Once the parameters have been estimated, volatility $\widehat{\sigma}_t=\widehat{\sigma}_t(\widehat{\vtheta})$, the conditional mean $\widehat{\mu}_t=\widehat{\mu}_t(\widehat{\vtheta})$ and the standardized residuals $\widehat{\varepsilon}_t=\widehat{\varepsilon}_t(\widehat{\vtheta})=[Y_t-\widehat{\mu}_t(\widehat{\vtheta})]/\widehat{\sigma}_t(\widehat{\vtheta})$ can be computed. Note that $\mu_t(\vtheta)$ and $\sigma_t(\vtheta)$ may depend on the infinite past of $Y_t$ (and, possibly, $\vx_t$) and, hence, can only be approximated via $\widehat{\mu}_t(\vtheta)$ and $\widehat{\sigma}_t(\vtheta)$ using some initial values; see, e.g., the truncated recursion \eqref{eq:vola} in Appendix~\ref{APARCH and ARMA--GARCH Models}.

For our final assumption, we define the neighborhood of the true $\vtheta^\circ$ as
\[
	N_n(\eta)=\left\{\vtheta\ :\ n^{1/2}|\vtheta-\vtheta^{\circ}|\leq\eta\right\},\qquad\eta>0,
\]
where $|\cdot|$ is some norm on $\mTheta$. 

\begin{assumption}\label{ass:UA}
For any $\eta>0$, there exist r.v.s $m_t:=m_{n,t}(\eta)\geq0$ and $s_t:=s_{n,t}(\eta)\geq0$ with $\max_{t=\ell_{n},\ldots,n}\m_{t}=o_{\p}(1)$ and $\max_{t=\ell_{n},\ldots,n}\s_{t}=o_{\p}(1)$ for any $\ell_{n}\leq n$ with $\ell_n\to\infty$, such that with probability approaching 1, as $n\to\infty$,
\[
	|\varepsilon_t|(1-s_t)-m_t\leq |\widehat{\varepsilon}_t(\vtheta)| \leq |\varepsilon_t|(1+s_t)+m_t
\]
for all $t=\ell_{n},\ldots,n$ and all $\vtheta\in N_n(\eta)$.
\end{assumption}

Assumption~\ref{ass:UA} imposes a uniform approximability of the innovations $|\varepsilon_t|$ by the $|\widehat{\varepsilon}_t(\vtheta)|$ in a $n^{-1/2}$-neighborhood of the true parameter. However, Assumption~\ref{ass:UA} does not require the first $(\ell_{n}-1)$ innovations to be approximable. This is convenient as the first few residuals are often imprecise, due to initialization effects caused by using artificial initial values in the mean and variance recursions. The variables $m_t$ and $s_t$ bound the error in the $\widehat{\varepsilon}_t(\vtheta)$ associated with approximating the conditional mean and conditional variance, respectively. Of course, Assumption~\ref{ass:UA} is a high-level condition that must be verified for each specific model on a case-by-case basis. In Appendix~\ref{APARCH and ARMA--GARCH Models}, we verify Assumption~\ref{ass:UA} for two popular model classes, namely APARCH and ARMA--GARCH models.

\subsection{The Tail Copula and its Estimators}\label{The Tail Copula and its Estimators}

The \textit{survival copula} at lag $d$ of the $\{|\varepsilon_t|\}$ from Section~\ref{Location--Scale Model} is
\[
	\overline{C}^{(d)}(u,v)=\p\big\{|\varepsilon_t|>F^{\leftarrow}(1-u),\ |\varepsilon_{t-d}|>F^{\leftarrow}(1-v)\big\},\qquad 0\leq u,v\leq 1,
\]
where $F^{\leftarrow}(\cdot)$ denotes the left-continuous inverse of $F(\cdot)$. Following \citet{SS06}, the (upper) \textit{tail copula} is the directional derivative of the survival copula at the origin with direction $(x,y)$, i.e., 
\[
	\Lambda^{(d)}(x,y)=\lim_{s\to\infty}s\overline{C}^{(d)}(x/s, y/s),
\]
where the limit is assumed to exist. Thus, the tail copula describes the serial extremal dependence structure of the $|\varepsilon_t|$ at different lags $d$. For $x=y=1$, the tail copula simplifies to the \textit{tail dependence coefficient} of \citet{Sib60}. In this case, it also coincides with the most popular version of the \textit{extremogram} due to \citet{DM09}. 

Let $k=k(n)$ be an intermediate sequence satisfying $k\to\infty$ and $k/n\to0$ as $n\to\infty$. Then, replacing $s$ with $n/k\to\infty$ and also replacing population quantities with empirical counterparts, one may estimate the tail copula non-parametrically via
\[
	\widetilde{\Lambda}_{n}^{(d)}(x,y)=\frac{1}{k}\sum_{t=d+1}^{n}I_{\{|\varepsilon_t|>|\varepsilon|_{(\lfloor kx\rfloor+1)},\ |\varepsilon_{t-d|}>|\varepsilon|_{(\lfloor ky\rfloor+1)}\}},
\]
where $|\varepsilon|_{(1)}\geq\ldots\geq |\varepsilon|_{(n)}$ denote the order statistics, and $\lfloor\cdot\rfloor$ rounds to the nearest smaller integer.

In the case of interest here, where the $\varepsilon_t$ from \eqref{eq:ls model} are i.i.d., the tail copula can easily shown to be identically zero, i.e., $\Lambda^{(d)}(x,y)\equiv0$. However, even when there is strong (non-extremal) dependence between $|\varepsilon_t|$ and $|\varepsilon_{t-d}|$, we have $\Lambda^{(d)}(x,y)\equiv0$. For instance, this is the case when the dependence between the two variables is governed by a Gaussian copula \textit{for any} correlation parameter $\rho\in(-1,1)$ \citep{Hef00}. Hence, while the relationship may be strong in non-extremal regions, it is non-existent in the limit as measured by the tail copula. Likewise, \citet{Hil11a} shows that the tail copula is identically zero if $\varepsilon_t$ follows a \textit{serially dependent} stochastic volatility model. These are well-known to display weaker extremal dependence than GARCH-type models \citep{DM09,DMZ13}. Thus, the tail copula cannot adequately discriminate between independent and dependent variables, which would result in compromised power of a specification test based on the tail copula.

For this reason, we prefer to base our test on what we term the \textit{PA-tail copula}:
\begin{equation}\label{eq:PA-tail copula}
	\Lambda_{n}^{(d)}(x,y)=\frac{n}{k}\p\left\{|\varepsilon_t|>b\Big(\frac{n}{kx}\Big),\ |\varepsilon_{t-d}|>b\Big(\frac{n}{ky}\Big)\right\},
\end{equation}
where $b(x)=F^{\leftarrow}(1-1/x)$. Now, the PA-tail copula is $\Lambda_{n}^{(d)}(x,y)=(k/n)xy$ for serially independent $\varepsilon_t$, but $\Lambda_{n}^{(d)}(x,y)\neq (k/n)xy$ when $\varepsilon_t$ follows a stochastic volatility model \citep{Hil11a} or when $(|\varepsilon_t|, |\varepsilon_{t-d}|)^\prime$ possesses a Gaussian copula. Thus, the PA-tail copula captures finer differences in pre-asymptotic levels of extremal dependence than the tail copula, ensuring that our test has power against more subtle misspecifications. For $x=y=1$, the PA-tail copula is the perhaps most popular version of the \textit{PA-extremogram}, due to \citet{DM09}. Even though the PA-tail copula differs from the tail copula, it may also be estimated via $\widetilde{\Lambda}_{n}^{(d)}(x,y)$.

\subsection{A Portmanteau-Type Test for Residual Extremal Dependence}\label{A Portmanteau-Type Test for Residual Extremal Dependence}

In specification testing, one is interested in verifying whether observations $Y_1,\ldots,Y_n$ are generated by some parametric model \eqref{eq:ls model}. Formally, one would like to test
\[
	H_0\ :\ \{Y_t\}_{t=1,2,\ldots}\ \text{are generated by the parametric model \eqref{eq:ls model}.}
\]
Typically, this is done by fitting the specific model \eqref{eq:ls model} to the $Y_1,\ldots,Y_n$. Under $H_0$, the residuals should then be approximately i.i.d. This is verified by testing some implication of the i.i.d.~property. For instance, Ljung--Box-type tests focus on the implication that the autocorrelation function is identically zero. As pointed out in the Motivation, we test the implication of $H_0$ that the residuals display no serial \textit{extremal} dependence. To this end, note that the PA-tail copula in \eqref{eq:PA-tail copula} equals $(k/n)xy$ under serial independence. Hence, we test the implication that $\Lambda_n^{(d)}(x,y)=(k/n)xy$ for $d=1,2,\ldots$. 

However, for purposes of model checking, the estimator $\widetilde{\Lambda}_{n}^{(d)}(x,y)$ from the previous subsection is infeasible, as the $\varepsilon_t$ are unobserved. Hence, we rely on the feasible counterpart
\[
	\widehat{\Lambda}_{n}^{(d)}(x,y)=\frac{1}{k}\sum_{t=d+1}^{n}I_{\{|\widehat{\varepsilon}_t|>|\widehat{\varepsilon}|_{(\lfloor kx\rfloor+1)},\ |\widehat{\varepsilon}_{t-d}|>|\widehat{\varepsilon}|_{(\lfloor ky\rfloor+1)}\}}.
\]
With this estimator, we can detect deviations from $\Lambda_n^{(d)}(x,y)=(k/n)xy$ using the Portmanteau-type test statistic
\begin{equation*}
	\mathcal{P}_n^{(D)}(x,y) = \frac{n}{xy}\sum_{d=1}^{D}\Big[\widehat{\Lambda}_n^{(d)}(x,y)-(k/n)xy\Big]^2,
\end{equation*}
where $D\in\mathbb{N}$ is some fixed user-specified integer. Once again, the fact that we compare $\widehat{\Lambda}_n^{(d)}(x,y)$ with the null hypothetical value of the PA tail copula $(k/n)xy$ (and not with the null hypothetical value of the tail copula, i.e., $0$) gives our test power in situations, where there is extremal independence (as measured by the tail copula), but not independence (as measured by the PA tail copula). 

In choosing $D$, there is the usual trade-off. Using a small $D$ possibly leads to undetected misspecifications at higher lags, resulting in a loss of power. However, choosing $D$ too large, the estimates $\widehat{\Lambda}_{n}^{(D)}(x,y)$ may be based on too few observations, thus distorting size. We explore the choice of $D$ in detail in the simulations in Appendix~\ref{Simulation Results for Varying $D$}. There, we show that $D=5$ typically leads to a good balance between size and power.

Regarding the sequence $k$, we impose

\begin{assumption}\label{ass:k}
For $n\to\infty$, the sequence $k$ satisfies $k\to\infty$, $k/n\to0$, and $k/\sqrt{n}\to\infty$.
\end{assumption}

Consistent with the notion of extremal dependence, the requirement that $k/n\to0$ ensures that only a vanishing fraction of the residuals is used in estimation. On the other hand, sufficiently many extremes need to be used, as specified by $k/\sqrt{n}\to\infty$. In the related context of estimating the \textit{tail event correlation}
\[
	r_n^{(d)}=\frac{n}{k}\left[\p\big\{|\varepsilon_{t}|>b(n/k),\ |\varepsilon_{t-d}|>b(n/k)\big\}-\p\big\{|\varepsilon_{t}|>b(n/k)\big\}\p\big\{|\varepsilon_{t-d}|>b(n/k)\big\}\right],
\]
\citet[Sec.~5.2]{Hil11} even requires $k/n^{2/3}\to\infty$. Hence, assuming $k/\sqrt{n}\to\infty$ may be regarded as a rather mild requirement.

The following is our first main theoretical result:

\begin{thm}\label{thm:mainresult}
Suppose Assumptions~\ref{ass:U distr 2}--\ref{ass:k} are met under $H_0$. Then, for any $x>0$ and $y>0$, as $n\to\infty$,
\[
	\mathcal{P}_n^{(D)}(x,y)\overset{d}{\longrightarrow}\chi^2_{D},
\]
where $\chi^2_{D}$ denotes the $\chi^2$-distribution with $D$ degrees of freedom.
\end{thm}

We defer the proof of Theorem~\ref{thm:mainresult} to Appendix~\ref{Proof of Theorem 1}. The proof reveals that the same $\chi^2$-limit as in Theorem~\ref{thm:mainresult} is obtained, if the $\widehat{\varepsilon}_t$'s are replaced with the true $\varepsilon_t$'s in $\mathcal{P}_n^{(D)}(x,y)$. Thus, the effects of parameter estimation in $\widehat{\varepsilon}_t=\widehat{\varepsilon}_t(\widehat{\vtheta})$ vanish asymptotically. In the context of residual-based estimation of univariate extremal quantities (such as the tail index or large quantiles), such vanishing estimation effects have been observed before by, e.g., \citet{Cea07}, \citet{Hil15}, and \citet{Hog18+}. The explanation is that the convergence rate of the parameter estimator is faster than that of the tail quantity, leading to vanishing parameter estimation effects.

For $x=y=1$ the PA-tail copula is the most prominent version of \citeauthor{DM09}'s \citeyearpar{DM09} PA-extremogram. Thus, the choice $x=y=1$ is the leading case for applications of $\mathcal{P}_n^{(D)}(x,y)$. However, other choices of $x$ and $y$ may lead to different test results. Thus, the next section derives a test that combines the evidence across different values of $x$ and $y$. This may also lead to more powerful specification tests, as we will see in the simulations.

\subsection{A Functional Portmanteau-Type Test}

To obtain powerful functional versions of $\mathcal{P}_n^{(D)}(x,y)$, we need to integrate it over a suitable area in the $(x,y)^\prime$-space. Observe that $\widehat{\Lambda}_n^{(d)}(x,y)$ estimates $\Lambda_n^{(d)}(x,y)$, which converges to the tail copula $\Lambda^{(d)}(x,y)$ (if it exists) for $n\to\infty$. From Theorem~1~(ii) of \citet{SS06}, we know that the tail copula is homogeneous, i.e., it satisfies $\Lambda^{(d)}(tx,ty)=t\Lambda^{(d)}(x,y)$ for all $t>0$. This suggests that $\widehat{\Lambda}_n^{(d)}(tx,ty)\approx t\widehat{\Lambda}_n^{(d)}(x,y)$, which implies that integrating over a larger space than some sphere $S_{c,\ \norm{\cdot}}:=\{(x,y)^\prime\in[0,\infty)^2\ :\ \norm{(x,y)^\prime}=c\}$ does not provide us with more information against the null. Here, $c>0$ is a constant and $\norm{\cdot}$ some norm on $\mathbb{R}^2$. A typical choice in extreme value theory is the $\norm{\cdot}_1$-norm. Since---as pointed out above---putting $x=y=1$ is a natural choice and $\norm{(1,1)^\prime}_1=2$, we fix $c=2$ and, thus, integrate over the sphere $S_{2,\ \norm{\cdot}_1}$. Using the $\norm{\cdot}_1$-norm in defining the sphere also has the added advantage of giving a limiting distribution in Theorem~\ref{thm:mainresult2} in terms of simple Brownian bridges.

These considerations lead to the functional Portmanteau-type test statistic
\[
	\mathcal{F}_n^{(D)}=n\sum_{d=1}^{D}\int_{[\iota, 1-\iota]}\Big[\widehat{\Lambda}_n^{(d)}(2-2z,2z)-\frac{k}{n}(2-2z)2z\Big]^2\D z,
\]
where $\iota\in(0,1/2)$ is some typically small constant chosen by the practitioner. The parameter $\iota$ serves to bound the $z$-values in the integral away 0 and 1, where our non-parametric estimator $\widehat{\Lambda}_n^{(d)}(2-2z,2z)$ may be unreliable as very extreme quantile thresholds are considered, above which hardly any observations lie. We mention that there is a trade-off between size and power involved in choosing $\iota$. A large value of $\iota$ leads to a smaller area being integrated over, thus leading to compromised power, yet estimates tend to be more stable, thus improving size. However, unreported simulations show that the differences in size and power are not large. In our numerical experiments, we choose $\iota=0.1$ and compute the integral in $\mathcal{F}_n^{(D)}$ using standard numerical integration techniques, which allow to approximate integrals with any required accuracy.

\begin{thm}\label{thm:mainresult2}
Suppose Assumptions~\ref{ass:U distr 2}--\ref{ass:k} are met under $H_0$. Then, as $n\to\infty$,
\[
	\mathcal{F}_n^{(D)}\overset{d}{\longrightarrow}4\sum_{d=1}^{D}\int_{[\iota,1-\iota]}B_d^2(z)\D z,
\]
where $\{B_d(\cdot)\}_{d=1,\ldots,D}$ are mutually independent Brownian bridges.
\end{thm}

Theorem~\ref{thm:mainresult2} is proven in Appendix~\ref{Proof of Theorem 2}. While the limiting distribution in Theorem~\ref{thm:mainresult2} is non-standard, it is free of nuisance parameters. Thus, critical values can be computed to any desired degree of precision by simulating sufficiently often from the limit. Table~\ref{tab:cv} shows some selected critical values for $\iota=0.1$ and different numbers of lags $D$.\footnote{The critical values in Table~\ref{tab:cv} are based on 4,000,000 replications of a Brownian bridge simulated on 100,000 grid points.} 

\begin{rem}
In the spirit of \citet{AD52}, it is also possible to consider a weighted version of the functional test statistic \`{a} la
\begin{equation}\label{eq:lim WF}
	\mathcal{WF}_n^{(D)}=n\sum_{d=1}^{D}\int_{[\iota, 1-\iota]}\psi(z)\Big[\widehat{\Lambda}_n^{(d)}(2-2z,2z)-\frac{k}{n}(2-2z)2z\Big]^2\D z\overset{d}{\longrightarrow}4\sum_{d=1}^{D}\int_{[\iota,1-\iota]}\psi(z)B_d^2(z)\D z,
\end{equation}
where $\psi(z)\geq0$ is a weight function satisfying some regularity conditions (e.g., continuity). The weight function allows to put more weight on certain $z$-values; e.g., $\phi(z)=z(1-z)$ would accentuate $z$ that are further away from the boundaries. We refrain from exploring such possibilities. However, we mention that for $\psi(z)\equiv1/4$, $\iota=0$ and $D=1$, the limit in \eqref{eq:lim WF} equals the limit of the well-known Cram\'{e}r--von Mises goodness-of-fit test. For $\psi(z)=1/[z(1-z)]$, we obtain the limit of the Anderson--Darling test statistic.
\end{rem}

\begin{table}[!t]
	\centering
		\begin{tabular}{lrrrrrrrrrr}
			\toprule
			$\alpha$			& \multicolumn{10}{c}{$D$}\\[0.25ex]
\cline{2-11}\\[-2.25ex] 
										& 1 & 2 & 3 & 4 & 5 & 6 & 7 & 8 & 9 & 10 \\
			\midrule
			10\%          & 1.340 & 2.336 & 3.231 & 4.077 & 4.896 & 5.694 & 6.477 & 7.249 &  8.011 &  8.766 \\
			5\%           & 1.791 & 2.890 & 3.859 & 4.765 & 5.636 & 6.480 & 7.306 & 8.117 &  8.916 &  9.705 \\
			1\%           & 2.905 & 4.178 & 5.273 & 6.286 & 7.248 & 8.178 & 9.082 & 9.964 & 10.832 & 11.683 \\
			\bottomrule
		\end{tabular}
	\caption{Critical values $c_{\alpha}^{(D)}$ for significance level $\alpha\in(0,1)$ and $D$ for $\iota=0.1$ in the limiting distribution in Theorem~\ref{thm:mainresult2}.}
	\label{tab:cv}
\end{table}

\section{Simulations}\label{Simulations}

Here, we investigate the size and power of the specification tests based on $\mathcal{P}_n^{(D)}$ (with the standard choice $x=y=1$) and $\mathcal{F}_n^{(D)}$, and compare the results with those of classical Ljung--Box tests based on the autocorrelations of the squared residuals. We do so for different samples sizes $n\in\{500,\ 1000,\ 2000\}$ and different significance levels $\alpha\in\{1\%,\ 5\%,\ 10\%\}$. Furthermore, we pick $\iota=0.1$. Other choices of $\iota$ (e.g., $\iota=0.05$ or even $\iota=0.01$) do not materially change the results. All simulations use \textit{R version 4.0.3} (\citealp{R4.0.3}) and results are based on 10,000 replications.

The choice of $D$ is common to all tests and we pick $D=5$ in the following. This choice is supported by simulations in Appendix~\ref{Simulation Results for Varying $D$}, where we investigate the influence of $D$ on the tests. Here, we instead explore the effect of different $k$'s on our test statistics $\mathcal{P}_n^{(D)}$ and $\mathcal{F}_n^{(D)}$. In applications of extreme value methods the choice of $k$ is often a tricky issue, because the results may be very sensitive to the specific value of $k$ \citep{HF06}. Therefore, it is of interest to investigate the robustness of our test with respect to $k$. A very popular choice in applications of extreme value theory is to follow \citeauthor{DuM83}'s \citeyearpar{DuM83} rule by setting $k=\lfloor 0.1\cdot n\rfloor$ \citep{QFP01,MF00,MT11,CES14}. However, at least asymptotically, this is not a valid choice (here and elsewhere) since it fails the Assumption~\ref{ass:k} requirement that $k$ be a vanishing fraction of $n$ ($k/n\to0$). Nonetheless, its widespread use suggests some merit in finite samples. Thus, we consider $k=\lfloor\varrho n^{0.99}\rfloor$ for $\varrho\in[0.05,\, 0.15]$, where the choice $\varrho=0.11$ roughly corresponds to \citeauthor{DuM83}'s \citeyearpar{DuM83} rule. Note that $k=\lfloor\varrho n^{0.99}\rfloor$ satisfies both Assumption~\ref{ass:k} requirements.

For simplicity, we only consider model \eqref{eq:ls model} with zero conditional mean, $\mu_t(\vtheta^{\circ})\equiv0$. We do so, because unmodeled dynamics in the tails of the conditional distribution are most likely caused by misspecified volatility dynamics in practice. Throughout, we simulate from APARCH(1,1) models under the null of correct specifications, so that in particular our high-level Assumption~\ref{ass:UA} is satisfied; see Appendix~\ref{APARCH and ARMA--GARCH Models}.

 We compare our tests with classical Ljung--Box tests when volatility dynamics are misspecified (Section~\ref{Misspecified Volatility}) and for non-i.i.d.~$\varepsilon_t$ (Section~\ref{Misspecified Innovations}). In both cases, we estimate our models using the \texttt{rugarch} package (\citealp{rugarch}).

\subsection{Misspecified Volatility}\label{Misspecified Volatility}

We generate time series from an APARCH model with an exogenous covariate $x_t$ in the volatility equation. Specifically, we simulate from the APARCH--X(1,1) model with $\delta^{\circ}=1$ given by
\begin{align}
	Y_t             &= \sigma_t(\vtheta^{\circ}) \varepsilon_t,\qquad\varepsilon_t\overset{\text{i.i.d.}}{\sim}(0,1)\notag\\
	\sigma_t(\vtheta^{\circ}) &=\omega^{\circ} + \alpha_{+,1}^{\circ}(Y_{t-1})_{+}
																	+ \alpha_{-,1}^{\circ}(Y_{t-1})_{-}
																	+ \beta_{1}^{\circ}\sigma_{t-1}(\vtheta^{\circ})
																	+ \pi_{1}^{\circ}x_{t-1}.\label{eq:APARCH-X(1,1)}
\end{align}
Here, the $\varepsilon_t$ follow a (standardized) Student's $t$-distribution with $4.1$ degrees of freedom to ensure $\E|\varepsilon_t^4|<\infty$, which is required for $\sqrt{n}$-consistent QML estimation in Assumption~\ref{ass:estimator} \citep{FT19}. The exogenous $x_t$ are stationary with
\begin{align*}
	x_{t}&=\exp(z_t), \\
	z_t  &=0.9\cdot z_{t-1}+e_t,\qquad e_t\overset{\text{i.i.d.}}{\sim}N(0,1).
\end{align*}
For the parameters, we take $\vtheta^{\circ}:=\vtheta^{\circ}_{s}:=(\omega^{\circ}, \alpha_{+,1}^{\circ}, \alpha_{-,1}^{\circ}, \beta_{1}^{\circ},\pi_{1}^{\circ})^\prime=(0.046, 0.027, 0.092, 0.843, 0)^\prime$ (to compute size) and $\vtheta^{\circ}:=\vtheta^{\circ}_{p}:=(0.046, 0.027, 0.092, 0.843, 0.089)^\prime$ (to compute power).\footnote{The data-generating process is taken from \citet[Sec.~3.1]{FT19}, who obtained the parameters in \eqref{eq:APARCH-X(1,1)} from fitting an APARCH--X model to Boeing returns.} Irrespective of the choice of $\vtheta^{\circ}$, we estimate an APARCH(1,1) model (i.e., we estimate \eqref{eq:APARCH-X(1,1)} imposing $\pi_{1}^{\circ}=0$) via QML. Thus, when $\vtheta^{\circ}=\vtheta^{\circ}_{s}$ we fit a correctly specified model (yielding size), yet when $\vtheta^{\circ}=\vtheta^{\circ}_{p}$ the fitted model is misspecified (yielding power). In the latter case, the exogenous variable serves to introduce some unmodeled dynamics in the variance equation, with infrequent outbursts of $z_t$ inducing some dependence in the tails of the residuals.\footnote{The second-to-top panel of Figure~\ref{fig:SimTS} in Appendix~\ref{Simulation Results for Varying $D$} shows a representative trajectory of $z_t$.} Hence, our simulation setup is well-suited to $\mathcal{P}_n^{(D)}$ and $\mathcal{F}_n^{(D)}$ and, at the same time, relevant in practice as exogenous volatility shocks are often empirically plausible; see \citet{HK14} for references.

For $v=10$, we generate $\{Y_t\}_{t=-v+1,\ldots,n}$ from \eqref{eq:APARCH-X(1,1)}. Since volatility estimates $\widehat{\sigma}_{t}(\widehat{\vtheta})$ may be imprecise for the first few $t$ due to initialization effects in the variance equation \citep[see also][]{HY03}, we discard the first $v$ standardized residuals and only consider $\{\widehat{\varepsilon}_t=Y_t/\widehat{\sigma}_{t}(\widehat{\vtheta})\}_{t=1,\ldots,n}$ in our test statistics $\mathcal{P}_n^{(D)}$ and $\mathcal{F}_n^{(D)}$.

Since we consider $k=\lfloor\varrho n^{0.99}\rfloor$ with $\varrho\in[0.05,\, 0.15]$, the $k$'s are between $[23, 70]$ / $[46, 139]$ / $[92, 278]$ for $n=500$ / $n=1000$ / $n=2000$. These values may be regarded as being sufficiently small relative to the sample size for extreme value methods to be applicable.

For purposes of comparison, we also include the results of a classical Ljung--Box test (with test statistic denoted by $\LB_{n}^{(D)}$) based on the first $D$ autocorrelations of the squared residuals. We employ the corrected Ljung--Box test statistic of \citet{CF11}, which ensures that $\LB_n^{(D)}$ is asymptotically $\chi_D^2$-distributed. Note that this correction is specific to APARCH models estimated via a Gaussian QML estimator.

\begin{figure}[t!]
	\centering
		\includegraphics[width=\textwidth]{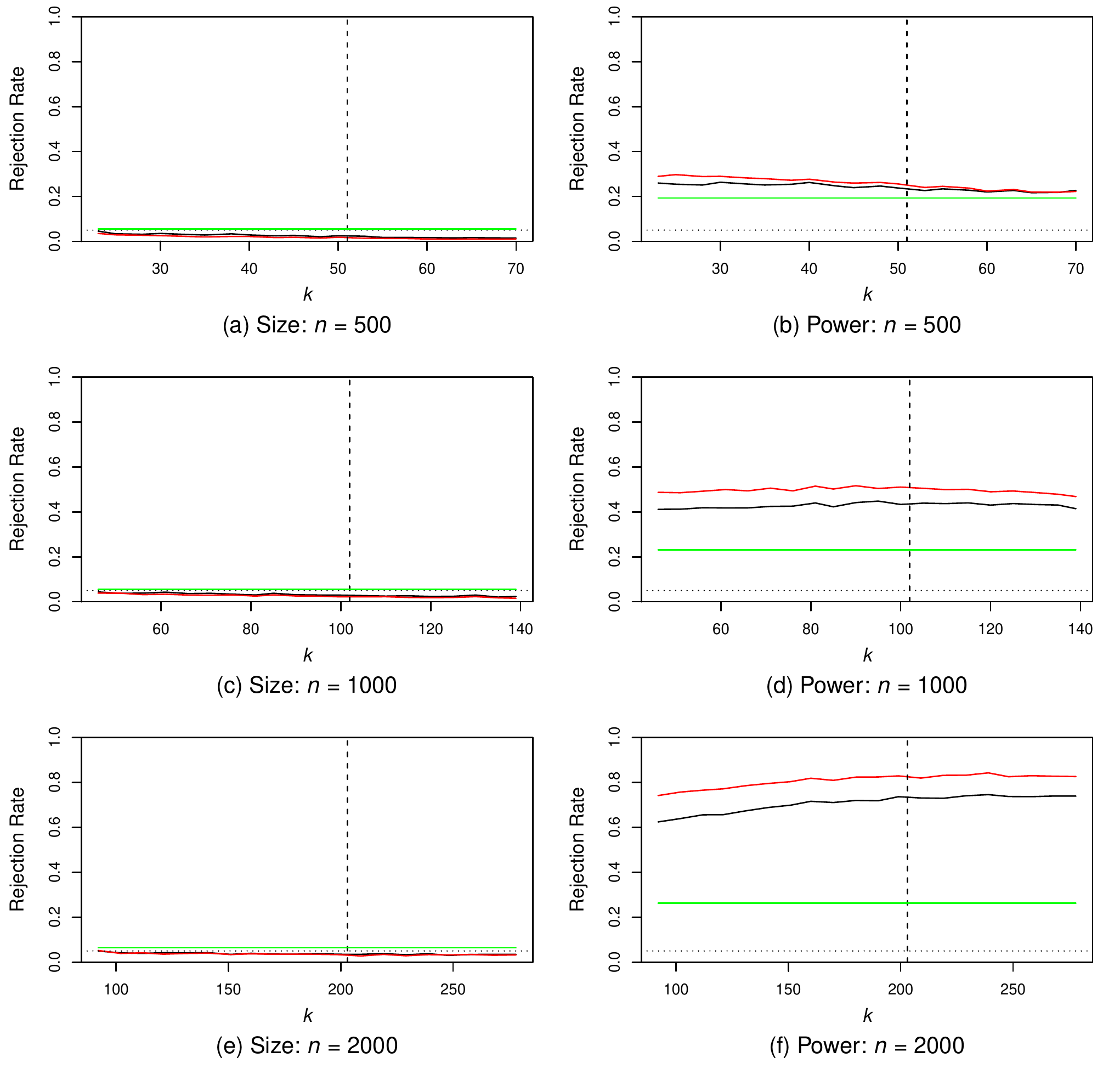}
	\caption{Rejection frequencies at the 5\%-level under the null and alternative for $\mathcal{P}_n^{(D)}$-test (black), $\mathcal{F}_n^{(D)}$-test (red) and $\LB_n^{(D)}$-test (green) for $D=5$ and $k=\lfloor\varrho n^{0.99}\rfloor$ with $\varrho\in[0.05,\, 0.15]$. Nominal level of 5\% indicated by the dotted horizontal line. Dashed vertical line indicates value of $k=\lfloor 0.11\cdot n^{0.99}\rfloor$. Results under the null (alternative) are for model \eqref{eq:APARCH-X(1,1)} with $\vtheta^{\circ}=\vtheta^{\circ}_{s}$ ($\vtheta^{\circ}=\vtheta^{\circ}_{p}$).}
	\label{fig:sev k}
\end{figure}

Figure~\ref{fig:sev k} displays the rejection frequencies of all tests at the 5\%-level. Both the $\mathcal{P}_n^{(D)}$-test (black) and the $\mathcal{F}_n^{(D)}$-test (red) have almost identical size, which---as the sample size increases---converges rapidly to the nominal level for all $k$. Larger differences in power between the two tests only emerge for large $n$. We also see that size and power are reasonably stable across different choices of $k$, suggesting that test results are quite robust to the particular value of $k$. This is encouraging given that extreme value methods can be very sensitive to the choice of the cutoff. Our approximation $k=\lfloor0.11\cdot n^{0.99}\rfloor$ to \citeauthor{DuM83}'s \citeyearpar{DuM83} rule---indicated by the dashed vertical lines in Figure~\ref{fig:sev k}---yields good results in terms of both size and power. Note that for $n=500$, where at first sight our heuristic does not lead to optimal power, size and power decrease with increasing $k$, such that the size-corrected power of our proposal for the choice of $k$ appears very close to optimal.

Finally, we compare our two tests with the $\LB_n^{(D)}$-test (green), which is of course independent of $k$. We find that---despite rejecting more often under the null---the $\LB_n^{(D)}$-test rejects less often under the alternative. Hence, the size-corrected power of our tests is even larger than the differences in the right-hand panels of Figure~\ref{fig:sev k} suggest.

\subsection{Misspecified Innovations}\label{Misspecified Innovations}

Standard location--scale models as in \eqref{eq:ls model} can only dynamically model the conditional mean and the conditional variance. However, higher-order features of the conditional distribution---such as skewness and kurtosis---may sometimes change as well for real data \citep{Han94,JR03a,BMT08}. Here, we assess how our tests perform relative to a standard Ljung--Box test in such a setting. To account for estimation effects, we use the version of the Ljung--Box test based on Theorem~8.2 of \citet{FZ10}. Since no confusion can arise, we denote the corresponding test statistic by $\LB_n^{(D)}$, which again has a standard $\chi_{D}^2$-limit.

To carry out the comparison, we consider the following GARCH(1,1) model
\begin{align}
	Y_t             &= \sigma_t(\vtheta^{\circ}) \varepsilon_t,\qquad\varepsilon_t\sim(0,1),\notag\\
	\sigma_t^2(\vtheta^{\circ}) &=\omega^{\circ} + \alpha^{\circ}Y_{t-1}^2																+ \beta^{\circ}\sigma_{t-1}^2(\vtheta^{\circ})\label{eq:GARCH(1,1)}
\end{align}
with $\vtheta^{\circ}:=(\omega^{\circ}, \alpha^{\circ}, \beta^{\circ})^\prime=(0.046, 0.127, 0.843)^\prime$. We assume that $\varepsilon_t$ is distributed according to \citeauthor{Han94}'s \citeyearpar{Han94} skewed $t$-distribution (written: $\varepsilon_t\sim st_{\lambda,\eta}$), i.e., it has density
\[
	f_{\lambda,\eta}(x)=bc\Bigg[1+\frac{1}{\eta-2}\Big(\frac{bx+a}{1+\sign(x+a/b)\lambda}\Big)^2\Bigg]^{-\frac{\eta+1}{2}},\qquad x\in(-\infty,\infty),\ \eta\in(2,\infty),\ \lambda\in(-1,1),
\]
where
\[
	a=4\lambda c\frac{\eta-2}{\eta-1},\qquad b^2=1+3\lambda^2 - a^2,\qquad c=\frac{\Gamma((\eta+1)/2)}{\sqrt{\pi(\eta-2)} \Gamma(\eta/2)}.
\]
For $\lambda=0$, $f_{\lambda,\eta}(\cdot)$ reduces to the standard $t$-distribution with degrees of freedom equal to $\eta$. \citet[Appendix~A]{JR03a} show that $f_{\lambda,\eta}(\cdot)$ has zero mean and unit variance. In their equations (2) and (3), they also derive formulas for the skewness and kurtosis as functions of $\lambda$ and $\eta$. 

To introduce time variation in the skewness and kurtosis, we let $\lambda=\lambda_t$ and $\eta=\eta_t$ vary over time. To do so, we use Model M4 of \citet[Table~1]{JR03a}. That is, we let
\begin{align*}
	\widetilde{\eta}_t 		&= a_1+b_1Y_{t-1}+c_1\widetilde{\eta}_{t-1},\\
	\widetilde{\lambda}_t &= a_2+b_2Y_{t-1}+c_2\widetilde{\lambda}_{t-1},
\end{align*}
and restrain $\eta_t$ and $\lambda_t$ to the intervals $(2,30)$ and $(-1, 1)$ by the logistic transformations $\eta_t=g_{(2, 30)}(\widetilde{\eta}_t)$ and $\lambda_t=g_{(-1, 1)}(\widetilde{\lambda}_t)$, where $g_{(L,U)}(x)=L+(U-L) / [1+\exp(-x)]$. Under the null, we set $(a_1,b_1,c_1)^\prime=(-3, 0, 0)^\prime$ and $(a_2,b_2,c_2)^\prime=(-1, 0, 0)^\prime$, leading to constant $\eta_t\equiv28.67...$ and $\lambda_t\equiv0.45...$. Thus, the $\varepsilon_t$ are serially independent and $H_0$ is satisfied. Under the alternative, we introduce time-variation in the skewness and kurtosis by using non-zero $b_i$ and $c_i$ ($i=1,2$). Specifically, we let $(a_1,b_1,c_1)^\prime=(-3, -6, 0.6)^\prime$ and $(a_2,b_2,c_2)^\prime=(-1, -2, 0.6)^\prime$. Thus, while $\varepsilon_t\sim(0,1)$, the $\varepsilon_t$ have time-varying skewness and kurtosis and, thus, are not i.i.d., as required under the null. 

To estimate the model parameters, we again use Gaussian QML. \citet[Theorem~2]{Esc09} shows that the QML estimator is asymptotically normal under both the null and the alternative.

\begin{figure}[t!]
	\centering
		\includegraphics[width=\textwidth]{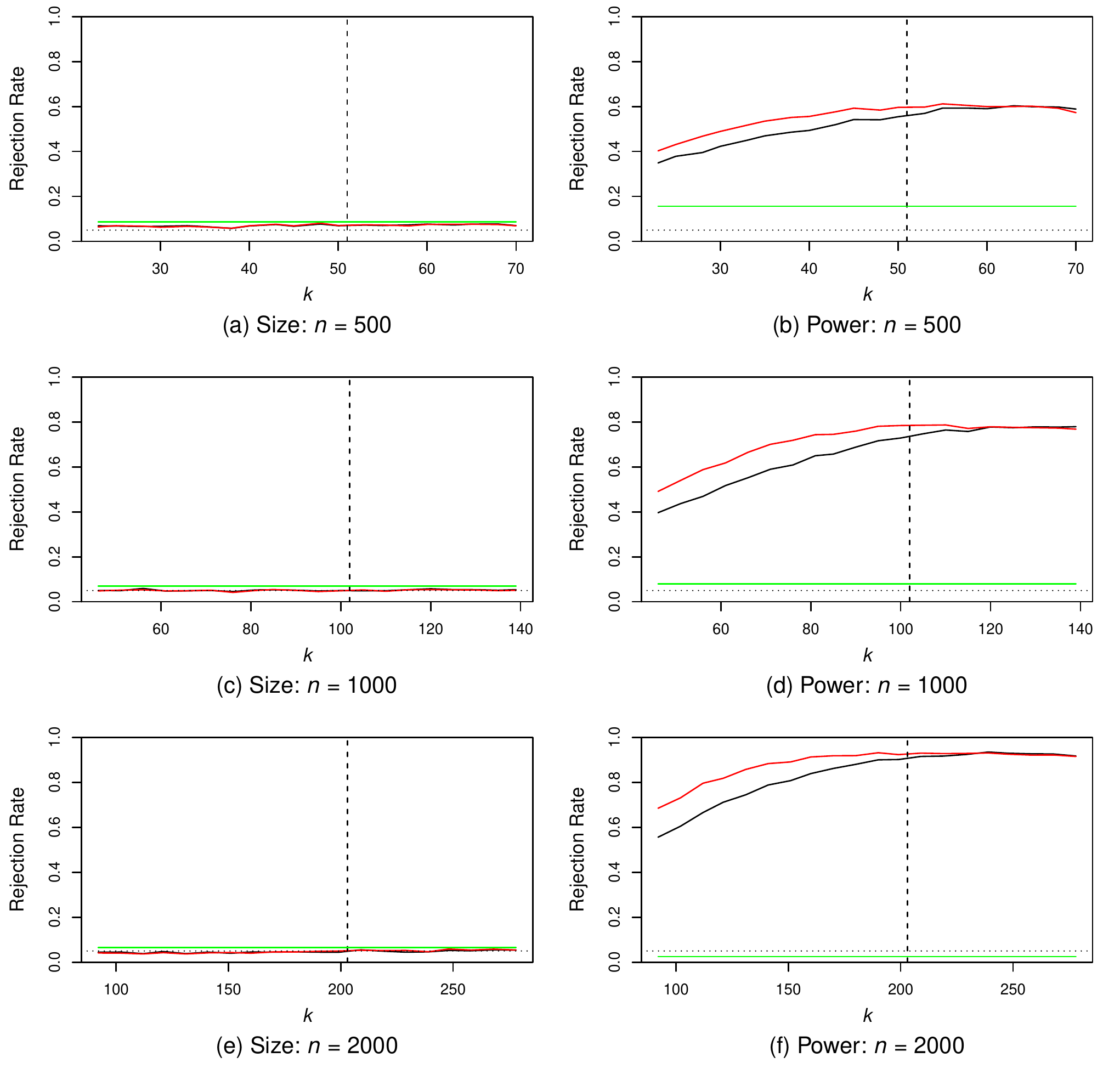}
	\caption{Rejection frequencies at the 5\%-level under the null and alternative for $\mathcal{P}_n^{(D)}$-test (black), $\mathcal{F}_n^{(D)}$-test (red) and $\LB_n^{(D)}$-test (green) for $D=5$ and $k=\lfloor\varrho n^{0.99}\rfloor$ with $\varrho\in[0.05,\, 0.15]$. Nominal level of 5\% indicated by the dotted horizontal line. Dashed vertical line indicates value of $k=\lfloor 0.11\cdot n^{0.99}\rfloor$. Results under the null (alternative) are for model \eqref{eq:GARCH(1,1)} with $(a_1,b_1,c_1)^\prime=(-3, 0, 0)^\prime$ and $(a_2,b_2,c_2)^\prime=(-1, 0, 0)^\prime$ ($(a_1,b_1,c_1)^\prime=(-3, -6, 0.6)^\prime$ and $(a_2,b_2,c_2)^\prime=(-1, -2, 0.6)^\prime$).}
	\label{fig:sev k1}
\end{figure}

Figure~\ref{fig:sev k1} displays the results. For all tests, size is very close to the nominal level, such that power is directly comparable. While for the $\LB_n^{(D)}$-test power is low and even decreases below the nominal level for $n=2000$, the power of our tests is high and approaches one. The huge difference in power may be explained as follows. The $\LB_n^{(D)}$-test only looks at the autocorrelations of the squared residuals $\widehat{\varepsilon}_t^2$. However, the $\varepsilon_t^2$ satisfy $\E[\varepsilon_t^2\mid\mathcal{F}_{t-1}]=1$, such that---despite being serially dependent---they have zero autocorrelations. Hence, the power of the $\LB_n^{(D)}$-test is seriously impaired. Comparing the $\mathcal{P}_n^{(D)}$- and the $\mathcal{F}_n^{(D)}$-test, we find that the latter has higher power for small $k$, yet the difference becomes negligible for larger $k$. For these two tests, the results are now more sensitive to the choice of $k$ than in Figure~\ref{fig:sev k}. Yet, our suggestion $k=\lfloor0.11\cdot n^{0.99}\rfloor$ once again leads to good size and power.

To summarize the simulation results, we find that size and power of our tests may depend somewhat on the choice of $k$. However, setting $k=\lfloor0.11\cdot n^{0.99}\rfloor$, which approximates \citeauthor{DuM83}'s \citeyearpar{DuM83} rule, gives good results---irrespective of the type of misspecification. Although our tests work well for $k=\lfloor0.11\cdot n^{0.99}\rfloor$, in practice we recommend to apply the tests for several values of $k$ as a `robustness check' (see Figure~\ref{fig:Fig6} in the empirical application for an example), much like results for standard Ljung--Box tests are also routinely reported for several lags. Our tests also depend on the number of included lags and, although we find $D=5$ to lead to good results in Appendix~\ref{Simulation Results for Varying $D$}, we again recommend---if possible---to report results for different $D$. We generally find that the $\mathcal{F}_n^{(D)}$-test is more powerful than the $\mathcal{P}_n^{(D)}$-test, yet both have markedly higher power than the $\LB_n^{(D)}$-test for the types of misspecifications considered here. As we have argued, these misspecifications may be empirically plausible, such that our tests should have substantial merit in practice. We investigate this next.

\section{Diagnosing S\&P~500 Constituents}\label{Application}

Consider the 500 components of the Standard \& Poor's 500  (S\&P~500) as of 31/5/2021. We apply the $\mathcal{F}_n^{(D)}$- and $\LB_n^{(D)}$-test to the log-returns of each of the constituents. We do so for the recommended values $D=5$, $k=\lfloor 0.11\cdot n^{0.99}\rfloor$ and $\iota=0.1$. For brevity, we do not report results for $\mathcal{P}_n^{(D)}$, which showed slightly inferior performance to $\mathcal{F}_n^{(D)}$ in the simulations. We consider the 495 components of the S\&P~500 for which we have complete data for our sample period from 1/1/2010 to 31/12/2019. We split the 10-year sample into an `in-sample' period of 8 years (used for specification testing) and an `out-of-sample' period of 2 years (used for backtesting VaR forecasts). The goal of our analysis is twofold. First, we want to illustrate that (similarly as in the simulations) cases may arise where $\mathcal{F}_n^{(D)}$ warrants a rejection, but $\LB_{n}^{(D)}$ does not. Our second goal is to show that in these cases, a rejection by $\mathcal{F}_n^{(D)}$ is not spurious, but---on the contrary---indicative of `out-of-sample' forecast failure of the model. 

Specifically, we fit an APARCH(1,1) model without covariates as in \eqref{eq:APARCH-X(1,1)} (i.e., with $\pi_1^{\circ}=0$) based on the `in-sample' period, and use the fitted model for one-step-ahead VaR forecasting in the `out-of-sample' part. The VaR at level $\theta$ is the loss in $(t+1)$ that is only exceeded with probability $\theta$ given the current state of the market (embodied by some information set $\mathcal{F}_t$). Formally, $\VaR_t$ is the $\mathcal{F}_{t}$-measurable random variable satisfying
\[
	\p\{Y_{t+1}\leq\VaR_t\mid\mathcal{F}_t\} = \theta.
\]
It immediately follows from the multiplicative structure in \eqref{eq:APARCH-X(1,1)} that $\VaR_t=\sigma_{t+1}\VaR_{\varepsilon}$, where $\VaR_{\varepsilon}$ is the (unconditional) $\theta$-quantile of the $\varepsilon_t$. Thus, we can easily compute the `out-of-sample' VaR forecasts from the volatility forecasts of the fitted model (say, $\widehat{\sigma}_{t+1}$) and an estimate of $\VaR_{\varepsilon}$ from the standardized residuals of the `in-sample' period. Then, we backtest the VaR forecasts using the DQ test of \citet{EM04}.\footnote{We use the DQ test as implemented in the \texttt{GAS} package (\citealp{GAS}) with 4 lags.}

\begin{table}[t!]
	\begin{center}
		\begin{tabular}{cccc}
			\toprule
$\LB_{n}^{(D)}$	& $\mathcal{F}_{n}^{(D)}$    &	Total & DQ-test \\
\midrule	
0 & 0   &           181   &   60   \\
0 & 1   &           122   &   69  \\
1 & 0   &           113   &   31  \\
1 & 1   &            79   &   38  \\
	\bottomrule
		\end{tabular}
	\end{center}
\caption{\label{tab:app}Results of specification and DQ tests carried out at 5\%-level.}
\end{table}

As we run two specification tests, there are four potential outcomes because the $\LB_{n}^{(D)}$- and the $\mathcal{F}_{n}^{(D)}$-test can either accept (indicated by a 0 in Table~\ref{tab:app}) or reject (indicated by a 1 in Table~\ref{tab:app}). The column `Total' in Table~\ref{tab:app} indicates the total number of times each of the four cases occurs for the 495 S\&P~500 stocks. In 181 out of 495 cases, none of the tests reject. For these 181 stocks, the `out-of-sample' VaR forecasts are rejected 60 times by the DQ-test. One reason for this rather large number of rejections may be the long out-of-sample period of 2 years. Typically, models are re-estimated at least quarterly \citep{AH14}, yet we only fitted the model once based on the `in-sample' data. Nonetheless, the rather long out-of-sample period of 2 years (yielding roughly 500 daily observations) is needed for the DQ-test to have reasonable power. Table~\ref{tab:app} further shows that for 122 stocks, $\mathcal{F}_{n}^{(D)}$ but not $\LB_{n}^{(D)}$ leads to a rejection. The fact that for these 122 stocks the DQ test rejects in 69 cases suggests some unmodeled dynamics in these stocks, that were picked up by our $\mathcal{F}_{n}^{(D)}$-test, yet not the Ljung--Box test. When both tests agree in rejecting a model (which happens in 79 cases), the model produces inadequate VaR forecasts 38 times. For the remaining 113 stocks, only $\LB_{n}^{(D)}$ rejects the null, yet the DQ-test rejects only 31 times.

\begin{figure}[t!]
	\centering
		\includegraphics[width=\textwidth]{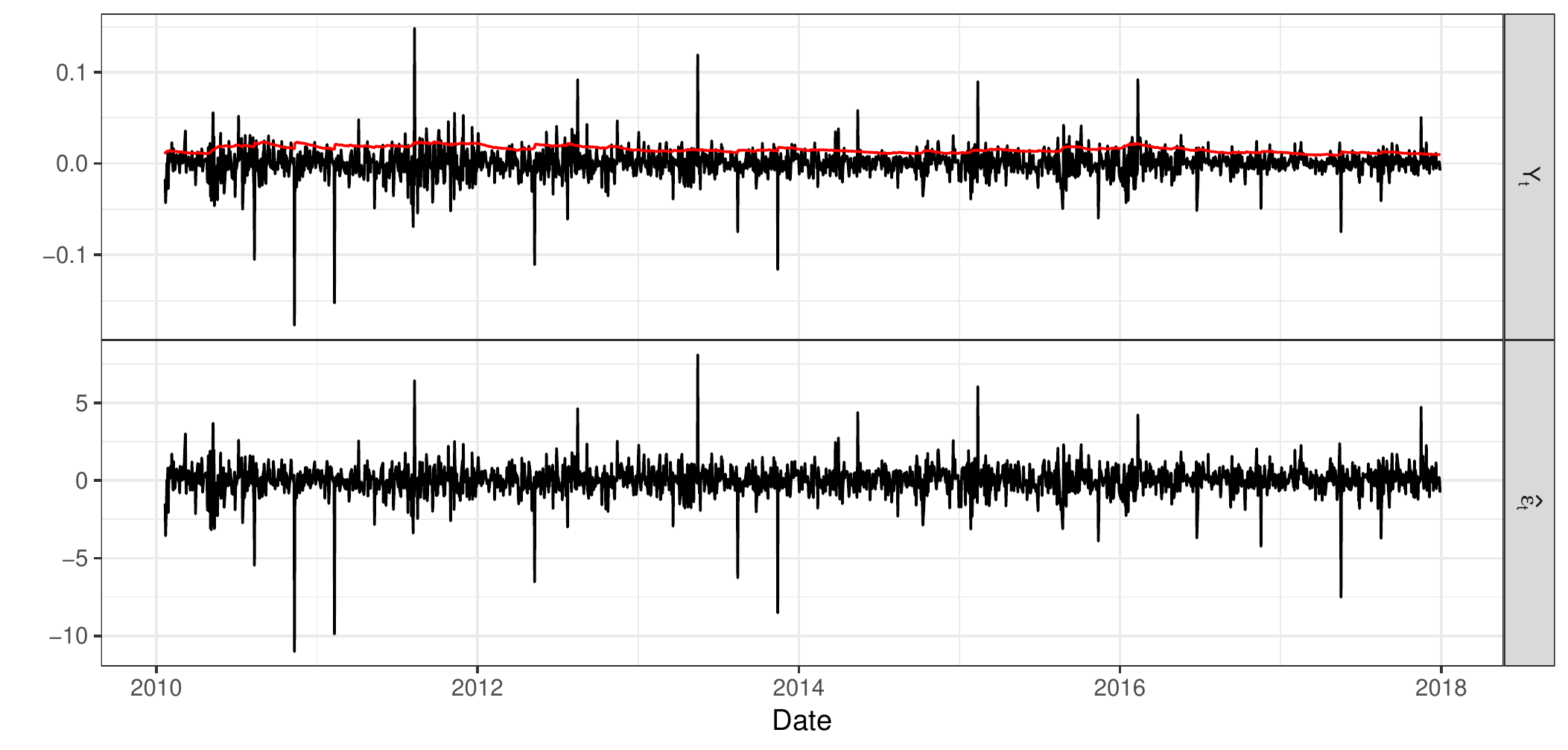}
	\caption{Top: Cisco returns ($Y_t$) in black and estimated volatility ($\widehat{\sigma}_t$) in red. Bottom: Standardized residuals ($\widehat{\varepsilon}_t$).}
	\label{fig:MRKy}
\end{figure}

Overall, we find that in most of the 198 cases when the DQ-test rejects, at least one of the specification test signals problems in advance. Of these 138 cases, 107 cases are identified by our $\mathcal{F}_n^{(D)}$-test, yet the $\LB_n^{(D)}$-test rejects in only 69 of these cases. This suggests that our extreme value-based tests provide complementary information to more classical specification tests.


\begin{figure}[t!]
	\centering
		\includegraphics[width=\textwidth]{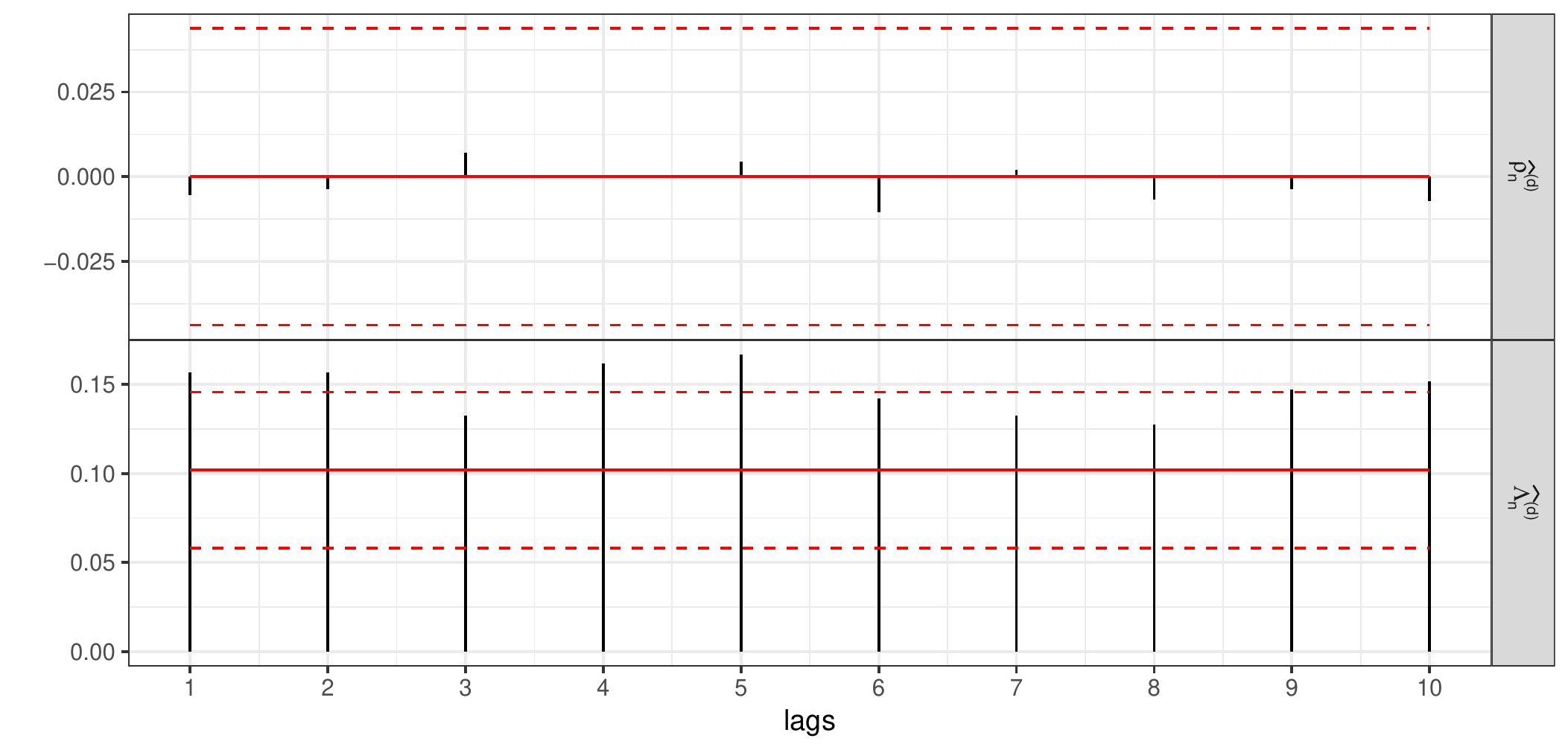}
	\caption{Top: Estimates $\widehat{\rho}_n^{(d)}$ of the autocorrelation at lag $d$ for the squared standardized residuals; bottom: Estimates $\widehat{\Lambda}_n^{(d)}(1,1)$ for the standardized residuals $\widehat{\varepsilon}_t$.}
	\label{fig:MRKacf}
\end{figure}

Next, we illustrate one of the cases where $\LB_{n}^{(D)}$ does not reject, yet $\mathcal{F}_{n}^{(D)}$ does. We do so for Cisco log-returns, shown in the top panel of Figure~\ref{fig:MRKy}. Volatility estimates are superimposed in red. During our in-sample period from 2010-2017 there are only mild signs of volatility clustering. A cursory inspection of the residuals in the lower panel suggests that the model successfully captures this. This is also confirmed by the plots of the squared residual autocorrelations in the top part of Figure~\ref{fig:MRKacf}, which shows that autocorrelations up to lag 10 are insignificant. However, focusing on the extremes, we find that the $\widehat{\Lambda}_n^{(d)}(1,1)$-estimates in Figure~\ref{fig:MRKacf} show strong (positive) serial extremal dependence at all lags. This may be due to occasional exogenous shocks to the share price. For instance, the largest drop on 11/11/2010 followed a profit warning of Cisco. This and other outliers apparent from the plot of the $Y_t$ in Figure~\ref{fig:MRKy} seem to have been unpredictable from past information in $\mathcal{F}_{t-1}$, thus violating the modeling paradigm in \eqref{eq:ls model}. Similarly as for model \eqref{eq:APARCH-X(1,1)} in the simulations, such exogenous shocks may have caused the serial extremal dependence in the standardized residuals.

\begin{figure}[t!]
	\centering
		\includegraphics[width=\textwidth]{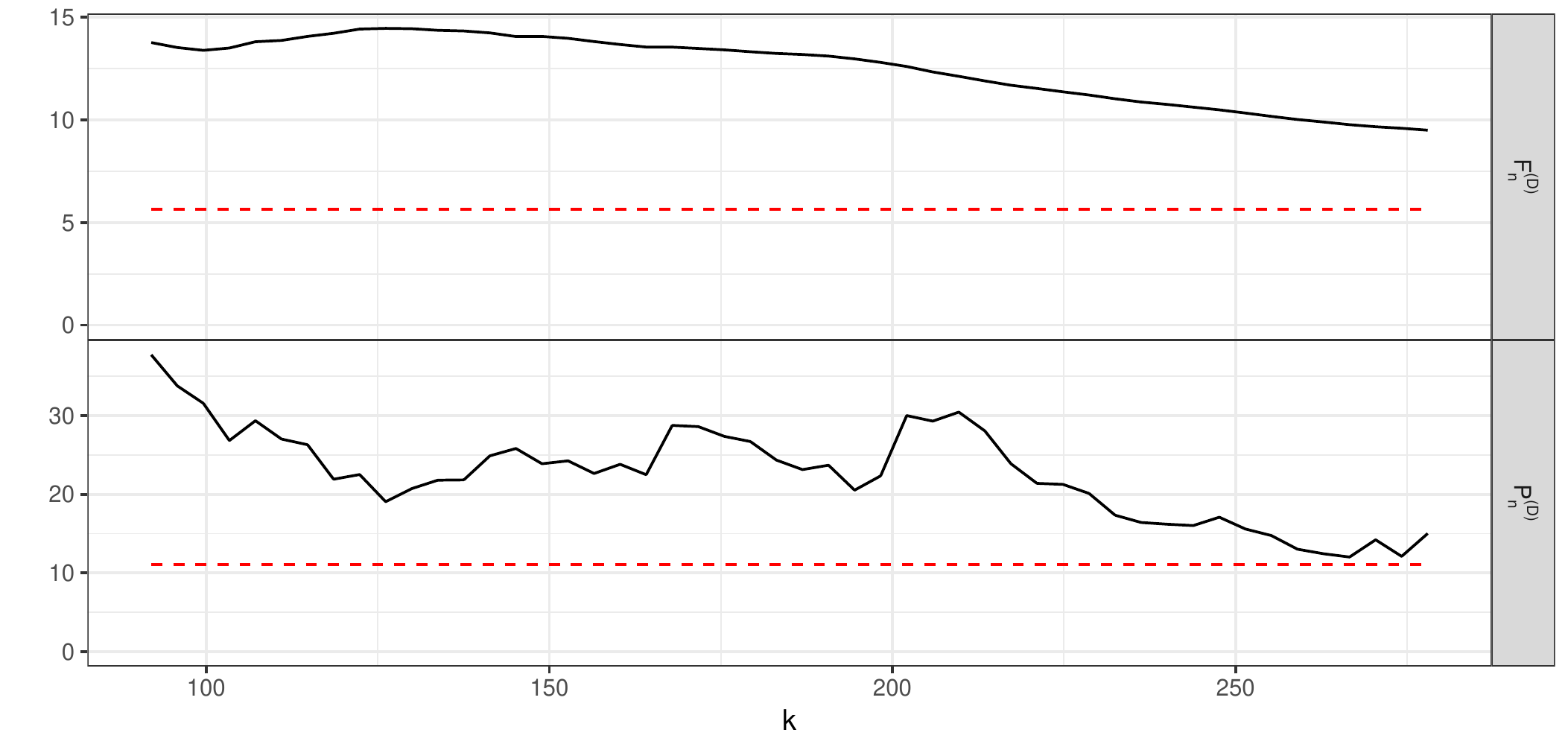}
	\caption{$\mathcal{F}_{n}^{(D)}$ (top) and $\mathcal{P}_{n}^{(D)}$ (top) for Cisco returns as a function of $k$. Horizontal red lines indicate respective 5\%-critical values.}
	\label{fig:Fig6}
\end{figure}

The results for the S\&P~500 components of this section are for $k=\lfloor0.11\cdot n^{0.99}\rfloor$. For a single time series one may want to assess the robustness of the result with respect to $k$. To do so, it is common in extreme value theory to plot the results as a function of $k$. We illustrate this for the Cisco returns in Figure~\ref{fig:Fig6}. There, the test statistics $\mathcal{F}_{n}^{(D)}$ and $\mathcal{P}_{n}^{(D)}$ are plotted as functions of $k$. As in the simulations, the test result is quite robust to $k$ with any reasonable choice leading to a rejection.

\section{Conclusion}\label{Conclusion}

We propose two specification tests based on the tail copula of time series residuals. By relying on the tail copula, our tests direct the power to any remaining serial \textit{extremal} dependence and, thus, complement the evidence by more classical Ljung--Box-type tests. The test statistics are easy to compute and their limiting distributions are free of nuisance parameters. Thus, our tests are simple to implement and enjoy broad applicability. This contrasts with more classical diagnostic tests, which require involved bootstrap procedures for valid inference and/or depend on the specific estimator used to fit the model, hence, limiting practical applications. Simulations demonstrate the good size of our proposals. They also show that when a misspecified model captures well the serial dependence in the body of the distribution but not in the tail, our tests have higher power than Ljung--Box tests. The misspecified models considered in the simulations are realistic in that they display exogenous shocks in the volatility equation or unmodeled higher-order dynamics in the innovations, both of which being empirically relevant sources of misspecification \citep{Han94,HK14}. We also exemplify the better detection properties of our tests on S\&P~500 constituents, where a rejection of our test is a more reliable indicator of poor out-of-sample risk predictions. Thus, our tests help to identify unsuitable risk forecasting models in advance, which is economically desirable under the Basel framework \citep{BCBSBF19}. Consequently, by focusing on the extremes, our tests are useful complements to standard specification tests, where valuable information on the serial extremal dependence may be `washed-out'.

\bigskip
\begin{center}
{\large\bf SUPPLEMENTARY MATERIAL}
\end{center}

\begin{description}

\item[Title:] Appendix containing verification of Assumption~\ref{ass:UA} for APARCH and ARMA--GARCH models (Appendices~\ref{APARCH and ARMA--GARCH Models} and \ref{Proof of Theorem}), the proofs of Theorems~\ref{thm:mainresult}--\ref{thm:mainresult2} (Appendices~\ref{Proof of Theorem 1}
 and \ref{Proof of Theorem 2}), and additional simulations (Appendix~\ref{Simulation Results for Varying $D$}). (pdf-file)

\item[R code:] File containing the \texttt{R} code to reproduce the simulation study and the empirical application. (zipped file)

\end{description}

\singlespacing

\bibliographystyle{jaestyle2}
\bibliography{thebib}

\begin{thebibliography}{58}
\providecommand{\natexlab}[1]{#1}

\bibitem[{Anderson and Darling(1952)}]{AD52}
Anderson TW, Darling DA. 1952. Asymptotic theory of certain "goodness of fit"
  criteria based on stochastic processes. \emph{The Annals of Mathematical
  Statistics} \textbf{23}: 193--212.

\bibitem[{Ardia \emph{et~al.}(2020)Ardia, Boudt and Catania}]{GAS}
Ardia D, Boudt K, Catania L. 2020. \emph{Generalized Autoregressive Score
  Models in {R}: The {GAS} Package}. R package version 0.3.3.

\bibitem[{Ardia and Hoogerheide(2014)}]{AH14}
Ardia D, Hoogerheide LF. 2014. {GARCH} models for daily stock returns: Impact
  of estimation frequency on value-at-risk and expected shortfall forecasts.
  \emph{Economics Letters} \textbf{123}: 187--190.

\bibitem[{Bali \emph{et~al.}(2008)Bali, Mo and Tang}]{BMT08}
Bali TG, Mo H, Tang Y. 2008. The role of autoregressive conditional skewness
  and kurtosis in the estimation of conditional {VaR}. \emph{Journal of Banking
  \& Finance} \textbf{32}: 269--282.

\bibitem[{{Basel Committee on Banking Supervision}(2019)}]{BCBSBF19}
{Basel Committee on Banking Supervision}. 2019. \emph{Basel Framework}. Basel:
  Bank for International Settlements
  (\url{http://www.bis.org/basel\_framework/index.htm?export=pdf}).

\bibitem[{Berkes and Horv\'ath(2004)}]{BH04}
Berkes I, Horv\'ath L. 2004. The efficiency of the estimators of the parameters
  in {GARCH} processes. \emph{Annals of Statistics} \textbf{32}: 633--655.

\bibitem[{Berkes \emph{et~al.}(2003)Berkes, Horv{\'a}th and Kokoszka}]{BHK03b}
Berkes I, Horv{\'a}th L, Kokoszka P. 2003. Asymptotics for {GARCH} squared
  residual correlations. \emph{Econometric Theory} \textbf{19}: 515--540.

\bibitem[{Bickel and Wichura(1971)}]{BW71}
Bickel PJ, Wichura MJ. 1971. Convergence criteria for multiparameter stochastic
  processes and some applications. \emph{The Annals of Statistics} \textbf{42}:
  1656--1670.

\bibitem[{Bollerslev(1986)}]{Bol86}
Bollerslev T. 1986. Generalized autoregressive conditional heteroskedasticity.
  \emph{Journal of Econometrics} \textbf{31}: 307--327.

\bibitem[{Box and Pierce(1970)}]{BP70}
Box GEP, Pierce DA. 1970. Distribution of residual autocorrelations in
  autoregressive-integreted moving average time series models. \emph{Journal of
  the American Statististical Association} \textbf{65}: 1509--1526.

\bibitem[{B\"ucher \emph{et~al.}(2015)B\"ucher, J\"aschke and Wied}]{BJW15}
B\"ucher A, J\"aschke S, Wied D. 2015. Nonparametric tests for constant tail
  dependence with an application to energy and finance. \emph{Journal of
  Econometrics} \textbf{187}: 154--168.

\bibitem[{Carbon and Francq(2011)}]{CF11}
Carbon M, Francq C. 2011. Portmanteau goodness-of-fit test for asymmetric power
  {GARCH} models. \emph{Austrian Journal of Statistics} \textbf{40}: 55--64.

\bibitem[{Chan \emph{et~al.}(2007)Chan, Deng, Peng and Xia}]{Cea07}
Chan NH, Deng SJ, Peng L, Xia Z. 2007. Interval estimation of value-at-risk
  based on {GARCH} models with heavy-tailed innovations. \emph{Journal of
  Econometrics} \textbf{137}: 556--576.

\bibitem[{Chavez-Demoulin \emph{et~al.}(2014)Chavez-Demoulin, Embrechts and
  Sardy}]{CES14}
Chavez-Demoulin V, Embrechts P, Sardy S. 2014. {Extreme-quantile tracking for
  financial time series}. \emph{Journal of Econometrics} \textbf{181}: 44--52.

\bibitem[{Davidson(1994)}]{Dav94}
Davidson J. 1994. \emph{Stochastic Limit Theory}. Oxford: Oxford University
  Press.

\bibitem[{Davis and Mikosch(2009)}]{DM09}
Davis RA, Mikosch T. 2009. The extremogram: A correlogram for extreme events.
  \emph{Bernoulli} \textbf{15}: 977--1009.

\bibitem[{Davis \emph{et~al.}(2012)Davis, Mikosch and Cribben}]{DMC12}
Davis RA, Mikosch T, Cribben I. 2012. Towards estimating extremal serial
  dependence via the bootstrapped extremogram. \emph{Journal of Econometrics}
  \textbf{170}: 142--152.

\bibitem[{Davis \emph{et~al.}(2013)Davis, Mikosch and Zhao}]{DMZ13}
Davis RA, Mikosch T, Zhao Y. 2013. Measures of serial extremal dependence and
  their estimation. \emph{Stochastic Processes and their Applications}
  \textbf{123}: 2575--2602.

\bibitem[{de~Haan and Ferreira(2006)}]{HF06}
de~Haan L, Ferreira A. 2006. \emph{Extreme Value Theory}. New York: Springer.

\bibitem[{Ding \emph{et~al.}(1993)Ding, Granger and Engle}]{DGE93}
Ding Z, Granger CWJ, Engle RF. 1993. A long memory property of stock market
  returns and a new model. \emph{Journal of Empirical Finance} \textbf{1}:
  83--106.

\bibitem[{DuMouchel(1983)}]{DuM83}
DuMouchel WH. 1983. Estimating the stable index $\alpha$ in order to measure
  tail thickness: A critique. \emph{The Annals of Statistics} \textbf{11}:
  1019--1031.

\bibitem[{Engle(1982)}]{Eng82}
Engle RF. 1982. Autoregressive conditional heteroscedasticity with estimates of
  the variance of {U}nited {K}ingdom inflation. \emph{Econometrica}
  \textbf{50}: 987--1007.

\bibitem[{Engle and Manganelli(2004)}]{EM04}
Engle RF, Manganelli S. 2004. {CAViaR}: Conditional autoregressive value at
  risk by regression quantiles. \emph{Journal of Business \& Economic
  Statistics} \textbf{22}: 367--381.

\bibitem[{Escanciano(2009)}]{Esc09}
Escanciano JC. 2009. Quasi-maximum likelihood estimation of semi-strong {GARCH}
  models. \emph{Econometric Theory} \textbf{25}: 561--570.

\bibitem[{Fisher and Gallagher(2012)}]{FG12}
Fisher TJ, Gallagher CM. 2012. New weighted {P}ortmanteau statistics for time
  series goodness of fit testing. \emph{Journal of the American Statistical
  Association} \textbf{107}: 777--787.

\bibitem[{Francq \emph{et~al.}(2006)Francq, Roy and Zako\"{i}an}]{FRZ06}
Francq C, Roy R, Zako\"{i}an JM. 2006. Diagnostic checking in {ARMA} models
  with uncorrelated errors. \emph{Journal of the American Statistical
  Association} \textbf{100}: 532--544.

\bibitem[{Francq and Thieu(2019)}]{FT19}
Francq C, Thieu LQ. 2019. {QML} inference for volatility models with
  covariates. \emph{Econometric Theory} \textbf{35}: 37--72.

\bibitem[{Francq and Zako\"{i}an(2004)}]{FZ04}
Francq C, Zako\"{i}an JM. 2004. Maximum likelihood estimation of pure {GARCH}
  and {ARMA--GARCH} processes. \emph{Bernoulli} \textbf{10}: 605--637.

\bibitem[{Francq and Zako\"{i}an(2010)}]{FZ10}
Francq C, Zako\"{i}an JM. 2010. \emph{{GARCH} Models: Structure, Statistical
  Inference and Financial Applications}. Chichester: Wiley.

\bibitem[{Ghalanos(2020)}]{rugarch}
Ghalanos A. 2020. \emph{rugarch: Univariate {GARCH} models}. R package version
  1.4-2.

\bibitem[{Glosten \emph{et~al.}(1993)Glosten, Jagannathan and Runkle}]{GJR93}
Glosten LR, Jagannathan R, Runkle DE. 1993. On the relation between the
  expected value and the volatility of the nominal excess return on stocks.
  \emph{The Journal of Finance} \textbf{48}: 1779--1801.

\bibitem[{Hall and Yao(2003)}]{HY03}
Hall P, Yao Q. 2003. Inference in {ARCH} and {GARCH} models with heavy-tailed
  errors. \emph{Econometrica} \textbf{71}: 285--317.

\bibitem[{Han and Kristensen(2014)}]{HK14}
Han H, Kristensen D. 2014. Asymptotic theory for the {QMLE} in {GARCH-X} models
  with stationary and nonstationary covariates. \emph{Journal of Business \&
  Economic Statistics} \textbf{32}: 416--429.

\bibitem[{Hansen(1994)}]{Han94}
Hansen BE. 1994. Autoregressive conditional density estimation.
  \emph{International Economic Review} \textbf{35}: 705--730.

\bibitem[{Heffernan(2000)}]{Hef00}
Heffernan JE. 2000. A directory of coefficients of tail dependence.
  \emph{Extremes} \textbf{3}: 279--290.

\bibitem[{Hidalgo and Zaffaroni(2007)}]{HZ07}
Hidalgo J, Zaffaroni P. 2007. A goodness-of-fit test for {ARCH}($\infty$)
  models. \emph{Journal of Econometrics} \textbf{141}: 973--1013.

\bibitem[{Hill(2011{\natexlab{a}})}]{Hil11a}
Hill JB. 2011{\natexlab{a}}. Extremal memory of stochastic volatility with an
  application to tail shape inference. \emph{Journal of Statistical Planning
  and Inference} \textbf{141}: 663--676.

\bibitem[{Hill(2011{\natexlab{b}})}]{Hil11}
Hill JB. 2011{\natexlab{b}}. Tail and nontail memory with applications to
  extreme value and robust statistics. \emph{Econometric Theory} \textbf{27}:
  844--884.

\bibitem[{Hill(2015)}]{Hil15}
Hill JB. 2015. Tail index estimation for a filtered dependent time series.
  \emph{Statistica Sinica} \textbf{25}: 609--629.

\bibitem[{Hoga(2019)}]{Hog18+}
Hoga Y. 2019. Confidence intervals for conditional tail risk measures in
  {ARMA--GARCH} models. \emph{Journal of Business \& Economic Statistics}
  \textbf{37}: 613--624.

\bibitem[{Jondeau and Rockinger(2003)}]{JR03a}
Jondeau E, Rockinger M. 2003. Conditional volatility, skewness, and kurtosis:
  Existence, persistence, and comovements. \emph{Journal of Economic Dynamics
  \& Control} \textbf{27}: 1699--1737.

\bibitem[{Kim and Lee(2016)}]{KL16}
Kim L, Lee S. 2016. On the tail index inference for heavy-tailed {GARCH}-type
  innovations. \emph{Annals of the Institute of Statistical Mathematics}
  \textbf{68}: 237--267.

\bibitem[{Li and Mak(1994)}]{LM94}
Li WK, Mak TK. 1994. On the squared residual autocorrelations in non-linear
  time series with conditional heteroskedasticity. \emph{Journal of Time Series
  Analysis} \textbf{15}: 627--636.

\bibitem[{Ling and Li(1997)}]{LL97}
Ling S, Li WK. 1997. On fractionally integrated autoregressive moving-average
  time series models with conditional heteroscedasticity. \emph{Journal of the
  American Statistical Association} \textbf{92}: 1184--1194.

\bibitem[{Ling and McAleer(2002)}]{LM02}
Ling S, McAleer M. 2002. Necessary and sufficient conditions for the
  {GARCH}($r,s$) and {Asymmetric Power GARCH}($r,s$) models. \emph{Econometric
  Theory} \textbf{18}: 722--729.

\bibitem[{Ljung and Box(1978)}]{LB78}
Ljung GM, Box GEP. 1978. On a measure of lack of fit in time series models.
  \emph{Biometrika} \textbf{65}: 297--303.

\bibitem[{Mancini and Trojani(2011)}]{MT11}
Mancini L, Trojani F. 2011. Robust value at risk prediction. \emph{Journal of
  Financial Econometrics} \textbf{9}: 281--313.

\bibitem[{McNeil and Frey(2000)}]{MF00}
McNeil AJ, Frey R. 2000. Estimation of tail-related risk measures for
  heteroscedastic financial time series: An extreme value approach.
  \emph{Journal of Empirical Finance} \textbf{7}: 271--300.

\bibitem[{Pan \emph{et~al.}(2008)Pan, Wang and Tong}]{PWT08}
Pan J, Wang H, Tong H. 2008. Estimation and tests for power-transformed and
  threshold {GARCH} models. \emph{Journal of Econometrics} \textbf{142}:
  352--378.

\bibitem[{Quintos \emph{et~al.}(2001)Quintos, Fan and Phillips}]{QFP01}
Quintos C, Fan Z, Phillips PCB. 2001. Structural change tests in tail behaviour
  and the {Asian} crisis. \emph{Review of Economic Studies} \textbf{68}:
  633--663.

\bibitem[{{R Core Team}(2020)}]{R4.0.3}
{R Core Team}. 2020. \emph{R: A Language and Environment for Statistical
  Computing}. R Foundation for Statistical Computing, Vienna, Austria.

\bibitem[{Resnick and St\v{a}ric\v{a}(1997)}]{RS97a}
Resnick S, St\v{a}ric\v{a} C. 1997. Smoothing the {H}ill estimator.
  \emph{Advances in Applied Probability} \textbf{29}: 271--293.

\bibitem[{Schmidt and Stadtm\"uller(2006)}]{SS06}
Schmidt R, Stadtm\"uller U. 2006. Nonparametric estimation of tail dependence.
  \emph{Scandinavian Journal of Statistics} \textbf{33}: 307--335.

\bibitem[{Shao(1993)}]{Sha93}
Shao QM. 1993. Almost sure invariance principles for mixing sequences of random
  variables. \emph{Stochastic Processes and their Applications} \textbf{48}:
  319--334.

\bibitem[{Sibuya(1960)}]{Sib60}
Sibuya M. 1960. Bivariate extreme statistics. \emph{Annals of the Institute of
  Statistical Mathematics} \textbf{11}: 195--210.

\bibitem[{Su \emph{et~al.}(2021+)Su, Qin, Peng and Qin}]{Sea21+}
Su Q, Qin Z, Peng L, Qin G. 2021+. Efficiently backtesting conditional
  value-at-risk and conditional expected shortfall. \emph{Journal of the
  American Statistical Association
  \textup{(\url{https://doi.org/10.1080/01621459.2020.1763804})}} : 1--12.

\bibitem[{White(2001)}]{Whi01}
White H. 2001. \emph{Asymptotic Theory for Econometricians}. San Diego:
  Academic Press, {F}irst edn.

\bibitem[{Zako\"{i}an(1994)}]{Zak94}
Zako\"{i}an JM. 1994. Threshold heteroskedastic models. \emph{Journal of
  Economic Dynamics and Control} \textbf{18}: 931--955.

\end{thebibliography}

\onehalfspacing

\newpage

\begin{appendices}


\section{APARCH and ARMA--GARCH Models}\label{APARCH and ARMA--GARCH Models}

\renewcommand{\theequation}{A.\arabic{equation}}	
\setcounter{equation}{0}	
\setcounter{page}{1}

We first verify Assumption~\ref{ass:UA} for the asymmetric power ARCH (APARCH) model of \citet{DGE93}; see also \citet{PWT08}. To that end, consider
\begin{align}
	Y_t &= \sigma_t(\vtheta^\circ) \varepsilon_t,\qquad \varepsilon_t\overset{\text{i.i.d.}}{\sim}(0,1),\notag\\
	\sigma_t^{\delta^{\circ}}(\vtheta^\circ)	&= \omega^{\circ}+\sum_{j=1}^{p}\left\{\alpha_{+,j}^{\circ}(Y_{t-j})_{+}^{\delta^{\circ}}+\alpha_{-,j}^{\circ}(Y_{t-j})_{-}^{\delta^{\circ}}\right\}+\sum_{j=1}^{q}\beta_j^{\circ}\sigma_{t-j}^{\delta^{\circ}}(\vtheta^\circ),\label{eq:APARCH}
\end{align}
where $y_{+}=\max\{y,0\}$, $y_{-}=\max\{-y,0\}$, $\omega^{\circ}>0$, $\delta^{\circ}>0$, and $\alpha_{+,j}^{\circ}$, $\alpha_{-,j}^{\circ}$ and $\beta_{j}^{\circ}$ are non-negative. Note that $\mu_t(\vtheta^{\circ})\equiv0$ here. As the likelihood is typically quite flat in the $\delta^{\circ}$-direction, we follow \citet{FT19} and assume the power $\delta^{\circ}$ to be a fixed known constant.\footnote{\citet[Sec.~2.5]{FT19} provide some theoretically-backed guidance on the choice of $\delta^{\circ}$ in practice.} Thus, the vector of the true unknown parameters is $\vtheta^{\circ}=(\omega^{\circ},\alpha_{+,1}^{\circ},\ldots,\alpha_{+,p}^{\circ},\alpha_{-,1}^{\circ},\ldots,\alpha_{-,p}^{\circ},\beta_1^{\circ},\ldots,\beta_{q}^{\circ})^\prime$. A generic parameter vector from the parameter space $\mTheta\subset(0,\infty)\times[0,\infty)^{d-1}$ ($d=2p+q+1$) is denoted by
\[
	\vtheta=(\omega,\alpha_{+,1},\ldots,\alpha_{+,p},\alpha_{-,1},\ldots,\alpha_{-,p},\beta_1,\ldots,\beta_{q})^\prime.
\]

The APARCH($p,q$) model in \eqref{eq:APARCH} nests several popular GARCH variants. We obtain \citeauthor{Zak94}'s \citeyearpar{Zak94} TARCH model for $\delta^{\circ}=1$, and the GJR--GARCH of \citet{GJR93} for $\delta^{\circ}=2$. When $\delta^{\circ}=2$ and $\alpha_{+,j}^{\circ}=\alpha_{-,j}^{\circ}$, \eqref{eq:APARCH} is the classic GARCH specification of \citet{Bol86}.

As volatility depends on the infinite past in the APARCH model, we rely on the truncated recursion
\begin{equation}\label{eq:vola}
	\widehat{\sigma}_t^{\delta^{\circ}}(\vtheta)=\omega+\sum_{j=1}^{p}\left\{\alpha_{+,j}(Y_{t-j})_{+}^{\delta^{\circ}}+\alpha_{-,j}(Y_{t-j})_{-}^{\delta^{\circ}}\right\}+\sum_{j=1}^{q}\beta_j\widehat{\sigma}_{t-j}^{\delta^{\circ}}(\vtheta),\qquad t\geq1,
\end{equation}
with initial values $Y_{1-p}=\ldots=Y_0=0$ and $\widehat{\sigma}_{1-q}^{\delta^{\circ}}(\vtheta)=\ldots=\widehat{\sigma}_{0}^{\delta^{\circ}}(\vtheta)=0$.

\begin{thm}\label{thm:UA}
Suppose for the APARCH model in \eqref{eq:APARCH} that 
\begin{enumerate}
	\item[(a)] $\min\{\omega^{\circ},\alpha_{+,1}^{\circ},\ldots,\alpha_{+,p}^{\circ},\alpha_{-,1}^{\circ},\ldots,\alpha_{-,p}^{\circ},\beta_1^{\circ},\ldots,\beta_{q}^{\circ}\}>0$, $\sum_{i=1}^{q}\beta_{i}^{\circ}<1$;
	\item[(b)] the top Lyapunov exponent $\gamma$ \citep[p.~41]{FT19} satisfies $\gamma<0$.
\end{enumerate}
Suppose further that Assumption~\ref{ass:estimator} is satisfied. Then, Assumption~\ref{ass:UA} holds for $\widehat{\varepsilon}_t(\vtheta)=Y_t/\widehat{\sigma}_t(\vtheta)$ with $\widehat{\sigma}_t(\vtheta)$ defined in \eqref{eq:vola}.
\end{thm}

Appendix~\ref{Proof of Theorem} contains the proof of Theorem~\ref{thm:UA}. Condition~\textit{(b)} ensures strict stationarity of the APARCH model, while \textit{(a)} bounds the true parameters away from the boundary of the parameter space to guarantee the approximability of the innovations.

Next, we consider the ARMA--GARCH model
\begin{equation}\label{eq:ARMA-GARCH}
	Y_t=\sum_{j=1}^{\overline{p}}\phi_j^{\circ}Y_{t-j}+ X_t-\sum_{j=1}^{\overline{q}}\vartheta_j^{\circ}X_{t-j},
\end{equation}
where the $X_t$ follow the GARCH model
\begin{align}
	X_t&=\sigma_t(\vtheta^{\circ})\varepsilon_t,\qquad\varepsilon_t\overset{\text{i.i.d.}}{\sim}(0,1),\notag\\
	\sigma_t^2(\vtheta^{\circ})&=\omega^{\circ}+\sum_{j=1}^{p}\alpha_j^{\circ}X_{t-j}^2+\sum_{j=1}^{q}\beta_j^{\circ}\sigma_{t-j}^2(\vtheta^{\circ}).\label{eq:GARCH part}
\end{align}
Here, $\vtheta^{\circ}=(\phi_1^{\circ}, \ldots, \phi_{\overline{p}}^{\circ}, \vartheta_1^{\circ}, \ldots, \vartheta_{\overline{q}}^{\circ},\omega^{\circ}, \alpha_1^{\circ},\ldots, \alpha_p^{\circ}, \beta_1^{\circ},\ldots, \beta_q^{\circ})^\prime$ is the true parameter vector contained is some parameter space $\mTheta=\mTheta_{\ARMA}\times\mTheta_{\GARCH}$, where $\mTheta_{\ARMA}\subset\mathbb{R}^{\overline{p}+\overline{q}}$ and $\mTheta_{\GARCH}\subset(0,\infty)\times[0,\infty)^{p+q}$. One can show that the model can be rewritten in the form \eqref{eq:ls model} with $\mu_t(\vtheta^{\circ}) = \sum_{j=1}^{\overline{p}}\phi_j^{\circ}Y_{t-j} - \sum_{j=1}^{\overline{q}}\vartheta_j^{\circ}[Y_{t-j}-\mu_{t-j}(\vtheta^{\circ})]$. We estimate the GARCH errors for a generic parameter vector $\vtheta\in\mTheta$ via
\[
	\widehat{X}_t(\vtheta)=Y_t-\sum_{j=1}^{\overline{p}}\phi_j Y_{t-j} + \sum_{j=1}^{\overline{q}}\vartheta_j \widehat{X}_{t-j}(\vtheta),
\]
where we artificially set $\widehat{X}_t(\vtheta)=Y_t=0$ for $t\leq0$. Based on the $\widehat{X}_t(\vtheta)$, we approximate volatility by the truncated recursion
\begin{equation}\label{eq:vola2}
	\widehat{\sigma}_t^2(\vtheta)=\omega+\sum_{j=1}^{p}\alpha_j\widehat{X}_{t-j}^2(\vtheta)+\sum_{j=1}^{q}\beta_j\widehat{\sigma}_{t-j}^2(\vtheta),
\end{equation}
where $\widehat{\sigma}_t^2(\vtheta)=0$ for $t\leq0$. 

\begin{thm}\label{thm:UA2}
Suppose for the ARMA--GARCH model in \eqref{eq:ARMA-GARCH} that
\begin{itemize}
	\item[(a)] the polynomials $\phi(z)=1-\sum_{j=1}^{\overline{p}}\phi_j^{\circ} z^{j}$ and $\vartheta(z)=1-\sum_{j=1}^{\overline{q}}\vartheta_j^{\circ} z^{j}$ have no common roots and no roots on the unit circle;
	\item[(b)] the GARCH parameters satisfy $\sum_{j=1}^{p}\alpha_j^{\circ}+\sum_{j=1}^{q}\beta_j^{\circ}<1$;
	\item[(c)] $\E|X_t|^{2+\delta}<\infty$ for some $\delta>0$.
\end{itemize}
Suppose further that Assumption~\ref{ass:estimator} is satisfied. Then, Assumption~\ref{ass:UA} holds for $\widehat{\varepsilon}_t(\vtheta)=\widehat{X}_t(\vtheta)/\widehat{\sigma}_t(\vtheta)$ with $\widehat{\sigma}_t(\vtheta)$ defined in \eqref{eq:vola2}.
\end{thm}

The proof of Theorem~\ref{thm:UA2} is also in Appendix~\ref{Proof of Theorem}. Assumption~\textit{(a)} is a standard stationarity, invertibility and identifiability condition for model~\eqref{eq:ARMA-GARCH}. Assumption~\textit{(b)} ensures a strictly stationary solution to the GARCH recurrence in \eqref{eq:GARCH part} with $\E|X_t|^{2}<\infty$. Thus, \textit{(c)} is only a mild additional requirement, for which Theorem~2.1 in \citet{LM02} provides necessary and sufficient conditions.

\section{Proofs of Theorems~\ref{thm:UA} and \ref{thm:UA2}}\label{Proof of Theorem}

\renewcommand{\theequation}{B.\arabic{equation}}	
\setcounter{equation}{0}	

All $o_{(\p)}$- and $O_{(\p)}$-symbols in Appendices~\ref{Proof of Theorem}--\ref{Proof of Theorem 2} are to be understood with respect to $n\to\infty$.

\begin{proof}[{\textbf{Proof of Theorem~\ref{thm:UA}:}}]
To avoid superscripts, we put $h_t(\vtheta)=\sigma_t^{\delta^{\circ}}(\vtheta)$, $h_{t}=h_{t}(\vtheta^{\circ})=\sigma_{t}^{\delta^{\circ}}(\vtheta^{\circ})$ and $\widehat{h}_t(\vtheta)=\widehat{\sigma}_{t}^{\delta^{\circ}}(\vtheta)$, where $\widehat{\sigma}_{t}^{\delta^{\circ}}(\vtheta)$ is defined in \eqref{eq:vola} in the main paper. Hence, $\widehat{h}_t(\vtheta)$ satisfies 
\[
	\widehat{h}_t(\vtheta)=\omega+\sum_{j=1}^{p}\left\{\alpha_{+,j}(Y_{t-j})_{+}^{\delta^{\circ}}+\alpha_{-,j}(Y_{t-j})_{-}^{\delta^{\circ}}\right\}+\sum_{j=1}^{q}\beta_j\widehat{h}_{t-j}(\vtheta),\qquad t\geq1,
\]
for initial values $Y_{1-p}=\ldots=Y_0=0$ and $\widehat{h}_{1-q}(\vtheta)=\ldots=\widehat{h}_{0}(\vtheta)=0$. Thus, while $h_t(\vtheta)$ denotes the true volatility (raised to the power of $\delta^{\circ}$) if $\vtheta$ were the true parameter vector, $\widehat{h}_t(\vtheta)$ approximates $h_t(\vtheta)$ using artificial initial values. Recalling that $h_t=h_t(\vtheta^{\circ})$, write
\begin{align}
\widehat{\varepsilon}_t(\vtheta) &= \frac{Y_t}{\widehat{\sigma}_t(\vtheta)}=\varepsilon_t\frac{h_t^{1/\delta^\circ}}{[\widehat{h}_t(\vtheta)]^{1/\delta^\circ}}\notag\\
&=\varepsilon_t\Big[1+\frac{h_t(\vtheta) - \widehat{h}_t(\vtheta)}{\widehat{h}_t(\vtheta)}\Big]^{1/\delta^\circ}\Big[1+\frac{h_t - h_t(\vtheta)}{h_t(\vtheta)}\Big]^{1/\delta^\circ}.\label{eq:(A.1)}
\end{align}
Note that the quantities in square brackets are positive almost surely and, hence, the power is well defined for $\delta^{\circ}>0$. Finally, define $N_n^{-}(\eta)=N_n(\eta)\setminus\{\vtheta^{\circ}\}$.

We require two results that follow directly from Lemma~13 and its proof in \citet{KL16}. First, for
\[
	\Delta_{n,t}=\sup_{\vtheta\in N_n^{-}(\eta)}\left|\frac{h_t-h_t(\vtheta)}{|\vtheta-\vtheta^{\circ}|h_t(\vtheta)}\right|,
\]
it holds that
\begin{equation}\label{eq:(3.11)}
	\max_{t=1,\ldots,n}\frac{\Delta_{n,t}}{n^{1/2}}=o_{\p}(1).
\end{equation}
Second, there exists $r_0\in[0,1)$, such that for sufficiently large $n$
\begin{equation}\label{eq:(4.11)}
\sup_{\vtheta\in N_n(\eta)}\left|\frac{h_t(\vtheta)-\widehat{h}_t(\vtheta)}{\widehat{h}_t(\vtheta)}\right|\leq r_0^{t}V_0,\qquad t=1,\ldots,n,
\end{equation}
where $V_0=V_0(\eta)\geq0$ with $\E|V_0|^{\nu_0}<\infty$ for some $\nu_0>0$. We have $\max_{t=\ell_{n},\ldots,n}r_0^{t}V_0=r_0^{\ell_n}V_0=o_{\p}(1)$.

Consider the two right-hand side factors in \eqref{eq:(A.1)}. Observe that
\[
	\sup_{\vtheta\in N_n(\eta)}\left|\frac{h_t - h_t(\vtheta)}{h_t(\vtheta)}\right|\leq \sup_{\vtheta\in N_n^{-}(\eta)}\big\{n^{1/2}|\vtheta-\vtheta^{\circ}|\big\}\Delta_{n,t}/n^{1/2}\leq\eta \Delta_{n,t} /n^{1/2}.
\]
Set $s_{1t}=r_0^t V_0$ and $s_{2t}=\eta\Delta_{n,t}/n^{1/2}$, whence $\max_{t=\ell_n,\ldots,n}s_{1t}=o_{\p}(1)$ and $\max_{t=\ell_n,\ldots,n}s_{2t}=o_{\p}(1)$. Then, for sufficiently large $n$,
\begin{align*}
	1-\underline{s}_{t}&:=[1-\min\{s_{1t},1\}]^{1/\delta^{\circ}}[1-\min\{s_{2t},1\}]^{1/\delta^{\circ}}\\
	&\leq \Big[1+\frac{h_t(\vtheta) - \widehat{h}_t(\vtheta)}{\widehat{h}_t(\vtheta)}\Big]^{1/\delta^\circ}\Big[1+\frac{h_t - h_t(\vtheta)}{h_t(\vtheta)}\Big]^{1/\delta^\circ}\\
	&\leq [1+s_{1t}]^{1/\delta^{\circ}}[1+s_{2t}]^{1/\delta^{\circ}}=:1+\overline{s}_t.
\end{align*}
Put $s_t=\max\{\underline{s}_{t},\ \overline{s}_t\}$. Then, $\max_{t=\ell_n,\ldots,n}s_{t}=o_{\p}(1)$ and 
\[
	1-s_{t}\leq \Big[1+\frac{h_t(\vtheta) - \widehat{h}_t(\vtheta)}{\widehat{h}_t(\vtheta)}\Big]^{1/\delta^\circ}\Big[1+\frac{h_t - h_t(\vtheta)}{h_t(\vtheta)}\Big]^{1/\delta^\circ} \leq 1+s_{t}
\]
holds with probability approaching 1 (w.p.a.~1), as $n\to\infty$. Hence, the conclusion follows with our choice of $s_t$ and $m_t\equiv0$.
\end{proof}

\begin{proof}[{\textbf{Proof of Theorem~\ref{thm:UA2}:}}]
Denoting true volatility by $\sigma_t^2=\sigma_t^{2}(\vtheta^{\circ})$, we write
\begin{align}
	\widehat{\varepsilon}_t(\vtheta) &= \frac{\widehat{X}_t(\vtheta)}{\widehat{\sigma}_t(\vtheta)}=\frac{X_t}{\widehat{\sigma}_t(\vtheta)} + \frac{\widehat{X}_t(\vtheta)-X_t}{\widehat{\sigma}_t(\vtheta)}\notag\\
	&= \varepsilon_t\Big[1+\frac{\sigma_t^2(\vtheta)-\widehat{\sigma}_t^2(\vtheta)}{\widehat{\sigma}_t^2(\vtheta)}\Big]^{1/2}\Big[1+\frac{\sigma_t^2-\sigma_t^2(\vtheta)}{\sigma_t^2(\vtheta)}\Big]^{1/2}+\frac{\widehat{X}_t(\vtheta)-X_t}{\widehat{\sigma}_t(\vtheta)}.\label{eq:decomp1}
\end{align}
We again set $N_n^{-}(\eta)=N_n(\eta)\setminus\{\vtheta^{\circ}\}$. \citet[p.~264]{KL16} show that for $\vtheta\in N_n(\eta)$
\begin{equation}\label{eq:2ndterm}
	\frac{|\sigma_t^2-\sigma_t^2(\vtheta)|}{\sigma_t^2(\vtheta)}\leq|\vtheta-\vtheta^{\circ}|\Pi_{1,n,t}^{*} + |\vtheta-\vtheta^{\circ}|^2\Pi_{2,n,t}^{*}
\end{equation}
for some $\Pi_{i,n,t}^{*}=\Pi_{i,n,t}^{*}(\eta)\geq0$ ($i=1,2$), both independent of $\vtheta$. Furthermore, define $X_t(\vtheta)=\sigma_t(\vtheta)\varepsilon_t$ and set
\[
	\Pi_{3,n,t}=\Pi_{3,n,t}(\eta)=\sup_{\vtheta\in N_n^{-}(\eta)}\frac{1}{|\vtheta-\vtheta^{\circ}|}\frac{|X_t(\vtheta)-X_t|}{\widehat{\sigma}_t(\vtheta)}.
\]
By Equation~(A.26) in \citet{Hog18+}, there exists $v_0>2$ such that
\begin{equation}\label{eq:(A.9m)}
	\limsup_{n\to\infty}\E|\Pi_{1,n,t}^{*}|^{v_0}<\infty,\quad\limsup_{n\to\infty}\E|\Pi_{2,n,t}^{*}|^{v_0/2}<\infty,\quad \limsup_{n\to\infty}\E|\Pi_{3,n,t}|^{v_0}<\infty
\end{equation}
uniformly in $t\in\mathbb{N}$. From this it follows that 
\begin{equation}\label{eq:op1 all}
	\max_{t=\ell_{n},\ldots,n}\Pi_{1,n,t}^{*}/\sqrt{n}=o_{\p}(1),\quad \max_{t=\ell_{n},\ldots,n}\Pi_{2,n,t}^{*}/n=o_{\p}(1),\quad \max_{t=\ell_{n},\ldots,n}\Pi_{3,n,t}/\sqrt{n}=o_{\p}(1).
\end{equation}
We only show the first claim in \eqref{eq:op1 all}, as the others follow along similar lines. Using subadditivity, Markov's inequality and \eqref{eq:(A.9m)}, we get
\begin{align*}
	\p\Big\{ \max_{t=\ell_{n},\ldots,n}\Pi_{1,n,t}^{*}/\sqrt{n} \geq\epsilon \Big\} &= \p\Big\{ \bigcup_{t=\ell_{n},\ldots,n}\big\{\Pi_{1,n,t}^{*}/\sqrt{n} \geq\epsilon\big\} \Big\}\\
	&\leq\sum_{t=\ell_{n}}^{n}\p\big\{ \Pi_{1,n,t}^{*}/\sqrt{n} \geq\epsilon \big\}\\
	&\leq\sum_{t=\ell_{n}}^{n}\epsilon^{-v_0}n^{-v_0/2}\E|\Pi_{1,n,t}^{*}|^{v_0}=o(1).
\end{align*}

By Lemma~A~1 in \citet{Hog18+}, there exist $r\in(0,1)$ and r.v.s $V_t\geq0$ with $\sup_{t\in\mathbb{N}}\E|V_t|^{v}<\infty$ for some $v>0$, such that for sufficiently large $n$,
\begin{equation}\label{eq:boundXs}
	\sup_{\vtheta\in N_n(\eta)}\max\big\{|\widehat{X}_t(\vtheta)-X_t(\vtheta)|,\ |\widehat{\sigma}_t^2(\vtheta)-\sigma_t^2(\vtheta)|\big\}\leq r^{t}V_t\qquad\text{for all }t\in\mathbb{N}.
\end{equation}
We easily obtain that $\max_{t=\ell_{n},\ldots,n}r^{t}V_t=o_{\p}(1)$, since
\begin{align*}
	\p\Big\{\max_{t=\ell_{n},\ldots,n}r^{t}V_t\geq\epsilon\Big\} &=\p\Big\{\bigcup_{t=\ell_{n},\ldots,n}\big\{r^{t}V_t\geq\epsilon\big\}\Big\} \\
	&\leq\sum_{t=\ell_{n}}^{n}\p\big\{r^{t}V_t\geq\epsilon\big\} \\
	&\leq\sum_{t=\ell_{n}}^{n}\epsilon^{-v}(r^{t})^v\E[V_t^{v}] \\
	&\leq C\Big\{\sup_{t=\ell_{n},\ldots,n}\E[V_t^{v}]\Big\}\sum_{t=\ell_{n}}^{n}(r^{t})^v=o(1),
\end{align*}
where the last equality follows from properties of the geometric series.

Equipped with these results, we now consider the terms on the right-hand side of \eqref{eq:decomp1} separately. First, by \eqref{eq:vola2} and the fact that $\omega^{\circ}>0$, we have for sufficiently large $n$ that $\widehat{\sigma}_t^2(\vtheta)\geq\underline{\omega}>0$ for all $\vtheta\in N_n(\eta)$ and all $t$. Combine this with \eqref{eq:boundXs} to obtain
\begin{align}
	\max_{t=\ell_{n},\ldots,n}\frac{|\sigma_t^2(\vtheta)-\widehat{\sigma}_t^2(\vtheta)|}{\widehat{\sigma}_t^2(\vtheta)}&\leq\underline{\omega}^{-1}\max_{t=\ell_{n},\ldots,n}r^{t}V_t=:\max_{t=\ell_{n},\ldots,n}s_{1t}=o_{\p}(1),\label{eq:beg1}\\
		\max_{t=\ell_{n},\ldots,n}\frac{|\widehat{X}_t(\vtheta)-X_t|}{\widehat{\sigma}_t(\vtheta)}&\leq\max_{t=\ell_{n},\ldots,n}\Big\{\underline{\omega}^{-1/2}|\widehat{X}_t(\vtheta)-X_t(\vtheta)|+ \frac{|X_t(\vtheta)-X_t|}{\widehat{\sigma}_t(\vtheta)}\Big\}\notag\\
		&\leq \max_{t=\ell_{n},\ldots,n}\Big\{\underline{\omega}^{-1/2}r^{t}V_t + \sup_{\vtheta\in N_n^{-}(\eta)}\big\{\sqrt{n}|\vtheta-\vtheta^{\circ}|\big\}\Pi_{3,n,t}/\sqrt{n}\Big\}\notag\\
		&\leq \max_{t=\ell_{n},\ldots,n}\Big\{\underline{\omega}^{-1/2}r^{t}V_t + \eta\Pi_{3,n,t}/\sqrt{n}\Big\}\notag\\		
		&=:\max_{t=\ell_{n},\ldots,n}m_t=o_{\p}(1).\notag
\end{align}
Finally, we have from \eqref{eq:2ndterm} and \eqref{eq:op1 all} that
\begin{align}
	\max_{t=\ell_{n},\ldots,n}\frac{|\sigma_t^2-\sigma_t^2(\vtheta)|}{\sigma_t^2(\vtheta)}&\leq\max_{t=\ell_{n},\ldots,n}\Big\{ \sup_{\vtheta\in N_n^{-}(\eta)}\big\{\sqrt{n}|\vtheta-\vtheta^{\circ}|\big\}\Pi_{1,n,t}^{*}/\sqrt{n} +  \sup_{\vtheta\in N_n^{-}(\eta)}\big\{n|\vtheta-\vtheta^{\circ}|^2\big\}\Pi_{2,n,t}^{*}/n\Big\}\notag\\
	&\leq  \max_{t=\ell_{n},\ldots,n}\Big\{\eta\Pi_{1,n,t}^{*}/\sqrt{n} + \eta^2\Pi_{2,n,t}^{*}/n\Big\}\notag\\
	&=:\max_{t=\ell_{n},\ldots,n}s_{2t}=o_{\p}(1).\label{eq:beg3}
\end{align}
Combining \eqref{eq:beg1} and \eqref{eq:beg3}, we get
\begin{align*}
	1-\underline{s}_{t}&:=[1-\min\{s_{1t},1\}]^{1/2}[1-\min\{s_{2t},1\}]^{1/2}\\
	&\leq \Big[1+\frac{\sigma_t^2(\vtheta)-\widehat{\sigma}_t^2(\vtheta)}{\widehat{\sigma}_t^2(\vtheta)}\Big]^{1/2}\Big[1+\frac{\sigma_t^2-\sigma_t^2(\vtheta)}{\sigma_t^2(\vtheta)}\Big]^{1/2}\\
	&\leq [1+s_{1t}]^{1/2}[1+s_{2t}]^{1/2}=:1+\overline{s}_t.
\end{align*}
Hence, for $s_t=\max\{\underline{s}_{t},\ \overline{s}_t\}$ with $\max_{t=\ell_n,\ldots,n}s_{t}=o_{\p}(1)$, it holds w.p.a.~1, as $n\to\infty$, that
\[
	1-s_{t}\leq \Big[1+\frac{\sigma_t^2(\vtheta)-\widehat{\sigma}_t^2(\vtheta)}{\widehat{\sigma}_t^2(\vtheta)}\Big]^{1/2}\Big[1+\frac{\sigma_t^2-\sigma_t^2(\vtheta)}{\sigma_t^2(\vtheta)}\Big]^{1/2} \leq 1+s_{t}.
\]
Thus, the conclusion follows with the above choices of $m_t$ and $s_t$.
\end{proof}

\section{Proof of Theorem~\ref{thm:mainresult}}\label{Proof of Theorem 1}

\renewcommand{\theequation}{C.\arabic{equation}}	
\setcounter{equation}{0}	

Let $C>0$ denote a large constant that may change from line to line. For notational brevity, we put 
\[
	U_t=|\varepsilon_t|,\quad \widehat{U}_t=|\widehat{\varepsilon}_t|,\quad \widehat{U}_{(i)}=|\widehat{\varepsilon}|_{(i)}.
\]
We denote the survivor function of $U_t$ by $\overline{F}(\cdot)=1-F(\cdot)$, where $F(\cdot)$ is from Assumption~\ref{ass:U distr 2}. For $\vx=(x,y)^\prime\in(0,\infty)^2$ and $\vxi=(\xi_1, \xi_2)^\prime\in\mathbb{R}^2$, define
\begin{align*}
	M_n(\vxi) &=M_n(\vx,\vxi)= \frac{1}{k}\sum_{t=d+1}^{n}I_{\big\{U_t>e^{\xi_1/\sqrt{k}}b(\frac{n}{kx}),\ U_{t-d}>e^{\xi_2/\sqrt{k}}b(\frac{n}{ky})\big\}},\\
	\widehat{M}_n(\vxi) &=\widehat{M}_n(\vx,\vxi)= \frac{1}{k}\sum_{t=d+1}^{n}I_{\big\{\widehat{U}_t>e^{\xi_1/\sqrt{k}}b(\frac{n}{kx}),\ \widehat{U}_{t-d}>e^{\xi_2/\sqrt{k}}b(\frac{n}{ky})\big\}}.
\end{align*}
If not specified otherwise, we assume the conditions of Theorem~\ref{thm:mainresult} to hold for the following lemmas and propositions, even though some of them may hold under a subset of the assumptions of Theorem~\ref{thm:mainresult}. 

The proof of Theorem~\ref{thm:mainresult} requires the following four propositions, which are proved in Subsection~\ref{Proofs of Propositions} of this Appendix. The first proposition derives an (infeasible) estimate of the PA-tail copula for the (unobserved) innovations. The limit theory developed by \citet{SS06} cannot be applied to prove the next proposition, since it only applies to tail dependent sequences of random variables. Recall for the following that $b(x)=F^{\leftarrow}(1-1/x)$.

\begin{prop}\label{prop:base convergence}
For i.i.d.~$U_t$ with d.f.~$F(\cdot)$ satisfying Assumption~\ref{ass:U distr 2}~(ii), it holds that, as $n\to\infty$,
\begin{equation*}
	\sqrt{\frac{n}{xy}}\begin{pmatrix}
	\frac{1}{k}\sum_{t=2}^{n}\Big[I_{\left\{U_t>b(\frac{n}{kx}),\ U_{t-1}>b(\frac{n}{ky})\right\}}
	-\Big(\frac{k}{n}\Big)^2xy\Big]\\
	\vdots \\
	\frac{1}{k}\sum_{t=D+1}^{n}\Big[I_{\left\{U_t>b(\frac{n}{kx}),\ U_{t-D}>b(\frac{n}{ky})\right\}}
	-\Big(\frac{k}{n}\Big)^2xy\Big]\\
	\end{pmatrix}\overset{d}{\longrightarrow}N(\vzeros,\mI_{D\times D}),
\end{equation*}
where $\mI_{D\times D}$ denotes the $(D\times D)$-identity matrix, and $\vzeros$ is a $(D\times 1)$-vector of zeros.
\end{prop}

In the following, we show that in the indicators in Proposition~\ref{prop:base convergence} we can replace the $U_t$ with the $\widehat{U}_t$, and the $b(n/[kz])$ by the empirical analog $\widehat{U}_{(\lfloor kz\rfloor +1)}$ ($z\in\{x,y\}$) without changing the limit. Then, Theorem~\ref{thm:mainresult} follows immediately from an application of the continuous mapping theorem. The following three propositions serve to justify these replacements.

\begin{prop}\label{prop:xi convergence}
For any $K>0$, it holds that
\[
	\sup_{\vxi\in[-K, K]^2}|M_n(\vxi)-M_n(\vzeros)|=o_{\p}(n^{-1/2}).
\]
\end{prop}

\begin{prop}\label{prop:hat convergence}
For any $K>0$, it holds that
\[
	\sup_{\vxi\in[-K, K]^2}|\widehat{M}_n(\vxi)-M_n(\vxi)|=o_{\p}(n^{-1/2}).
\]
\end{prop}

\begin{prop}\label{prop:emp an}
For any $x>0$, it holds that
\[
	|\xi_n(x)|=\left|\sqrt{k}\log\Bigg(\frac{\widehat{U}_{(\lfloor kx\rfloor+1)}}{b(n/[kx])}\Bigg)\right|=O_{\p}(1).
\]
\end{prop}

Equipped with these four propositions, we proceed to prove Theorem~\ref{thm:mainresult}.

\begin{proof}[{\textbf{Proof of Theorem~\ref{thm:mainresult}:}}]
We show that for each $d\in\{1,\ldots,D\}$,
\begin{equation}\label{eq:to show1}
	\Bigg|\frac{1}{k}\sum_{t=d+1}^{n}I_{\{U_t>b(\frac{n}{kx}),\ U_{t-d}>b(\frac{n}{ky})\}}-\frac{1}{k}\sum_{t=d+1}^{n}I_{\{\widehat{U}_t>\widehat{U}_{(\lfloor kx\rfloor+1)},\ \widehat{U}_{t-d}>\widehat{U}_{(\lfloor ky\rfloor+1)}\}}\Bigg|=o_{\p}(n^{-1/2}).
\end{equation}
Combining this with Proposition~\ref{prop:base convergence}, the desired convergence follows from Lemma~4.7 of \citet{Whi01} and the continuous mapping theorem \citep[e.g.,][Theorem~7.20]{Whi01}. Decompose
\begin{align}
M_n(\vzeros) &= \frac{1}{k}\sum_{t=d+1}^{n}I_{\{U_t>b(\frac{n}{kx}),\ U_{t-d}>b(\frac{n}{ky})\}} \notag\\
&= \frac{1}{k}\sum_{t=d+1}^{n}I_{\{U_t>b(\frac{n}{kx}),\ U_{t-d}>b(\frac{n}{ky})\}}-\frac{1}{k}\sum_{t=d+1}^{n}I_{\{U_t>e^{\xi_1/\sqrt{k}}b(\frac{n}{kx}),\ U_{t-d}>e^{\xi_2/\sqrt{k}}b(\frac{n}{ky})\}}\notag\\
&\hspace{0.4cm} +\frac{1}{k}\sum_{t=d+1}^{n}I_{\{U_t>e^{\xi_1/\sqrt{k}}b(\frac{n}{kx}),\ U_{t-d}>e^{\xi_2/\sqrt{k}}b(\frac{n}{ky})\}} - \frac{1}{k}\sum_{t=d+1}^{n}I_{\{\widehat{U}_t>e^{\xi_1/\sqrt{k}}b(\frac{n}{kx}),\ \widehat{U}_{t-d}>e^{\xi_2/\sqrt{k}}b(\frac{n}{ky})\}}\notag\\
&\hspace{0.4cm} +\frac{1}{k}\sum_{t=d+1}^{n}I_{\{\widehat{U}_t>e^{\xi_1/\sqrt{k}}b(\frac{n}{kx}),\ \widehat{U}_{t-d}>e^{\xi_2/\sqrt{k}}b(\frac{n}{ky})\}}\notag\\
&=o_{\p}(n^{-1/2}) + \widehat{M}_n(\vxi),\label{eq:(2.1.1)}
\end{align}
where the $o_{\p}(n^{-1/2})$-term is uniform on $\vxi\in[-K,K]^2$ by Propositions~\ref{prop:xi convergence} and~\ref{prop:hat convergence}. 

Define $\vxi_n=\big(\xi_n(x),\ \xi_n(y)\big)^\prime$, where $\xi_n(z)=\sqrt{k}\log\big(\widehat{U}_{(\lfloor kz\rfloor+1)}/b(n/[kz])\big)$ ($z\in\{x,y\}$) as in Proposition~\ref{prop:emp an}. Then, $\widehat{M}_n(\vxi_n)=\frac{1}{k}\sum_{t=d+1}^{n}I_{\{\widehat{U}_t>\widehat{U}_{(\lfloor kx\rfloor+1)},\ \widehat{U}_{t-d}>\widehat{U}_{(\lfloor ky\rfloor+1)}\}}$. To show \eqref{eq:to show1} or, equivalently, that $\widehat{M}_n(\vxi_n)-M_n(\vzeros)=o_{\p}(n^{-1/2})$ we have to prove that for any $\varepsilon>0$ and $\delta>0$
\[
	\p\left\{\sqrt{n}|\widehat{M}_n(\vxi_n)-M_n(\vzeros)|>\varepsilon\right\}\leq\delta
\]
for sufficiently large $n$. By Proposition~\ref{prop:emp an}, we can choose $K>0$ such that $\p\{|\vxi_n|>K\}\leq\delta/2$ for all sufficiently large $n$. Furthermore, by \eqref{eq:(2.1.1)}, $\p\{\sqrt{n}\sup_{\vxi\in[-K, K]^2}|\widehat{M}_n(\vxi)-M_n(\vzeros)|>\varepsilon\}\leq\delta/2$ for large $n$. Using these two results, we obtain
\begin{align*}
	\p\left\{\sqrt{n}|\widehat{M}_n(\vxi_n)-M_n(\vzeros)|>\varepsilon\right\} &\leq\p\Big\{|\vxi_n|>K\Big\} + \p\Big\{\sqrt{n}\sup_{\vxi\in[-K, K]^2}|\widehat{M}_n(\vxi)-M_n(\vzeros)|>\varepsilon,\ |\vxi_n|\leq K\Big\}\\
	&\leq\frac{\delta}{2}+\frac{\delta}{2}=\delta,
\end{align*}
i.e., \eqref{eq:to show1}. The conclusion follows.
\end{proof}

\subsection{Proofs of Propositions~\ref{prop:base convergence}--\ref{prop:emp an}}\label{Proofs of Propositions}

\begin{proof}[{\textbf{Proof of Proposition~\ref{prop:base convergence}:}}]
For $d\in\{1,\ldots,D\}$, we define
\begin{align*}
	Z_{n,t}^{(d)}&=Z_{n,t}^{(d)}(x,y)=\begin{cases}0, & t=1,\ldots,d,\\
	\frac{n}{k}\left[I_{\left\{U_t>b(\frac{n}{kx}),\ U_{t-d}>b(\frac{n}{ky})\right\}}
	-\p\left\{U_t>b\Big(\frac{n}{kx}\Big),\ U_{t-d}>b\Big(\frac{n}{ky}\Big)\right\}\right],& t>d,
	\end{cases}\\
	\mZ_{n,t}&=\big(Z_{n,t}^{(1)},\ldots, Z_{n,t}^{(D)}\big)^\prime.
\end{align*}
Note that by Assumption~\ref{ass:U distr 2}~(ii)
\[
	\p\left\{U_t>b\Big(\frac{n}{kx}\Big),\ U_{t-d}>b\Big(\frac{n}{ky}\Big)\right\}=\p\left\{U_t>b\Big(\frac{n}{kx}\Big)\right\}\p\left\{U_{t-d}>b\Big(\frac{n}{ky}\Big)\right\}=\left(\frac{k}{n}\right)^2xy
\]
for sufficiently large $n$.

Let $\vlambda=(\lambda_1,\ldots,\lambda_D)^\prime\in\mathbb{R}^{D}$ be any vector with $\vlambda^\prime\vlambda=1$. Then, by the Cram\'{e}r--Wold device \citep[e.g.,][Proposition~5.1]{Whi01} it suffices to show that $n^{-1/2}\sum_{t=1}^{n}\vlambda^\prime\mZ_{n,t}\overset{d}{\longrightarrow}\sqrt{xy}\vlambda^\prime \mZ$ for $\mZ\sim N(\vzeros,\mI_{D\times D})$. To do so, we verify the conditions of Theorem~5.20 of \citet{Whi01}.

Since the $\vlambda^\prime \mZ_{n,t}$ are (row-wise) $D$-dependent, the array $\{\vlambda^\prime \mZ_{n,t}\}$ is obviously $\phi$-mixing of any rate \citep[Sec.~3]{Whi01}. Furthermore, $\E[\vlambda^\prime \mZ_{n,t}]=0$ and 
\begin{align*}
	\E[(\vlambda^\prime \mZ_{n,t})^2] &= \Var(\vlambda^\prime \mZ_{n,t})=\Var\Big(\sum_{d=1}^{D}\lambda_d Z_{n,t}^{(d)}\Big)\\
	&=\sum_{d=1}^{D}\lambda_d^2\Var( Z_{n,t}^{(d)}) + \sum_{c\neq d}\lambda_c\lambda_d \Cov( Z_{n,t}^{(c)},\ Z_{n,t}^{(d)}).
\end{align*}
Consider the variances and covariances separately. For the variance we get for $t>d$
\begin{align}
\Var\big( Z_{n,t}^{(d)}\big) &= \frac{n^2}{k^2}\Var\Big(I_{\left\{U_t>b(\frac{n}{kx}),\ U_{t-d}>b(\frac{n}{ky})\right\}}\Big)\notag\\
	&= \frac{n^2}{k^2}\Big(\frac{k}{n}\Big)^2xy\left[1-\Big(\frac{k}{n}\Big)^2xy\right]\notag\\
	&=xy+o(1),\label{eq:Var1}
\end{align}
where the second line exploits that $I_{\left\{U_t>b(\frac{n}{kx}),\ U_{t-d}>b(\frac{n}{ky})\right\}}$ is Bernoulli-distributed with success probability $p=(k/n)^2xy$ and, hence,
\begin{equation*}
	\Var\Big(I_{\left\{U_t>b(\frac{n}{kx}),\ U_{t-d}>b(\frac{n}{ky})\right\}}\Big)=\Big(\frac{k}{n}\Big)^2xy\left[1-\Big(\frac{k}{n}\Big)^2xy\right].
\end{equation*}
Using similar arguments, we obtain
\begin{align}
	\Cov\big( Z_{n,t}^{(c)}, Z_{n,t}^{(d)}\big) &= O\left(\frac{n^2}{k^2}\p\left\{U_r>b\Big(\frac{n}{kx}\Big),\ U_{s}>b\Big(\frac{n}{ky}\Big),\ U_{t}>b\Big(\frac{n}{kz}\Big)\right\}\right)\notag\\
	&=O\Big(\frac{n^2}{k^2}\frac{k^3}{n^3}\Big)=O\Big(\frac{k}{n}\Big)=o(1)\label{eq:cov bound in sum},
\end{align}
where $r,s,t$ are pairwise distinct. Overall, we have $\E[(\vlambda^\prime \mZ_{n,t})^2]\leq 2<\infty$ for sufficiently large $n$.

Finally, exploiting $D$-dependence of $\vlambda^\prime\mZ_{n,t}$ again,
\begin{align}
\overline{\sigma}_n^2&:= \Var\Big(n^{-1/2}\sum_{t=1}^{n}\vlambda^\prime \mZ_{n,t}\Big)\notag\\
&=\frac{1}{n}\sum_{t=1}^{n}\Var(\vlambda^\prime \mZ_{n,t})+\frac{1}{n}\sum_{s\neq t}\Cov(\vlambda^\prime \mZ_{n,s}, \vlambda^\prime \mZ_{n,t})\notag\\
&=\frac{1}{n}\sum_{t=1}^{n}\Var(\vlambda^\prime \mZ_{n,t})+\frac{2}{n}\sum_{s=1}^{n}\sum_{t=s+1}^{\min\{s+D, n\}}\Cov(\vlambda^\prime \mZ_{n,s}, \vlambda^\prime \mZ_{n,t}).\label{eq:(A.5-)}
\end{align}
It is easy to check that the sum on the right-hand side of \eqref{eq:(A.5-)} vanishes, because the covariance terms are of order $o(1)$ by \eqref{eq:cov bound in sum}. For the sum involving the variances, we get using \eqref{eq:Var1}, \eqref{eq:cov bound in sum} and $\vlambda^\prime\vlambda=1$ that for $t>D$
\begin{align*}
\Var(\vlambda^\prime \mZ_{n,t}) 	&= \Var\Big(\sum_{d=1}^{D}\lambda_{d} Z_{n,t}^{(d)}\Big)\\
															&=\sum_{d=1}^{D}\lambda_{d}^2\Var\big( Z_{n,t}^{(d)}\big) + \sum_{c\neq d}\lambda_{c}\lambda_{d}\Cov\big(Z_{n,t}^{(c)}, Z_{n,t}^{(d)}\big)\\
															&= \sum_{d=1}^{D}\lambda_{d}^2[xy+o(1)]+o(1)=xy+o(1).							
\end{align*}
Overall, we obtain that $\overline{\sigma}_n^2=xy+o(1)\longrightarrow xy$, as $n\to\infty$.

Thus, we may apply Theorem~5.20 of \citet{Whi01}, yielding that $n^{-1/2}\sum_{t=1}^{n}\vlambda^\prime\mZ_{n,t}\overset{d}{\longrightarrow}\sqrt{xy}\vlambda^\prime \mZ$. This concludes the proof.
\end{proof}

The following lemma bounds the variation in the tail of the $U_t$, and is used throughout the proofs. 

\begin{lem}\label{lem:Lem F}
Suppose Assumption~\ref{ass:U distr 2} holds. Then, for any $\nu\in(-1,\infty)$ there exists $\widetilde{\nu}$ between $1$ and $\nu$, such that
\begin{align*}
	\overline{F}\Big((1+\nu)b\Big(\frac{n}{kx}\Big)\Big) &= \frac{kx}{n}\Big[1-\frac{\alpha\nu}{(1+\widetilde{\nu})^{1+\alpha}}+\nu o(1)\Big],
\intertext{where the $o(1)$-term vanishes uniformly on compact $\nu$-sets in $(-1,\infty)$ and on compact $x$-sets in $(0,\infty)$. In particular,}
	\overline{F}\Big(e^{\xi/\sqrt{k}}b\Big(\frac{n}{kx}\Big)\Big) &= \frac{kx}{n}\Bigg[1-\frac{\alpha\xi}{\sqrt{k}}+o\Big(\frac{1}{\sqrt{k}}\Big)\Bigg]
\end{align*}
uniformly on compact $\xi$-sets in $\mathbb{R}$ and on compact $x$-sets in $(0,\infty)$.
\end{lem}

Lemma~\ref{lem:Lem F} is proved in Appendix~\ref{Proofs of Auxiliary Lemmas}.

\begin{proof}[{\textbf{Proof of Proposition~\ref{prop:xi convergence}:}}]
We only consider the supremum over $\vxi\in[-K,0]^2$; that over the other three quadrants can be dealt with similarly. Observe that
\begin{align*}
	\sup_{\vxi\in[-K,0]^2}|M_n(\vxi)-M_n(\vzeros)|&\leq M_n(-(K,K)^\prime)-M_n(\vzeros)\\
	&=\frac{1}{k}\sum_{t=d+1}^{n}\Big[I_{\{U_t>e^{-K/\sqrt{k}}b(\frac{n}{kx}),\ U_{t-d}>e^{-K/\sqrt{k}}b(\frac{n}{ky})\}}-I_{\{U_t>b(\frac{n}{kx}),\ U_{t-d}>b(\frac{n}{ky})\}}\Big]\\
	&=:\frac{1}{k}\sum_{t=d+1}^{n}I_{1t},
\end{align*}
where $I_{1t}\in\{0,1\}$. Markov's inequality implies
\begin{equation}\label{eq:Cheby}
	\p\left\{\sqrt{n}\sup_{\vxi\in[-K,0]^2}|M_n(\vxi)-M_n(\vzeros)|>\varepsilon\right\}
	\leq\frac{\varepsilon^{-4}n^2}{k^4}\E\Big[\sum_{t=d+1}^{n}I_{1t}\Big]^4.
\end{equation}
Due to $d$-dependence, the $I_{1t}$ are trivially $\rho$-mixing, and thus we may apply Lemma~2.3 of \citet{Sha93} (for, in his notation, $q=4$) to deduce that
\begin{equation}\label{eq:Shao1}
	\E\Big[\sum_{t=d+1}^{n}I_{1t}\Big]^4\leq n^2 C \big\{\E[I_{1t}^2]\big\}^2+n C \E[I_{1t}^4].
\end{equation}
Since $I_{1t}\in\{0,1\}$ and, hence, $\E[I_{1t}^4]=\E[I_{1t}^2]=\E[I_{1t}]$, it suffices to compute $\E[I_{1t}]$. Using Lemma~\ref{lem:Lem F}, we get that
\begin{align*}
	\E[I_{1t}]	&=\p\Big\{U_t>e^{-K/\sqrt{k}}b\Big(\frac{n}{kx}\Big)\Big\} \p\Big\{U_{t-d}>e^{-K/\sqrt{k}}b\Big(\frac{n}{ky}\Big)\Big\}-\p\Big\{U_t>b\Big(\frac{n}{kx}\Big)\Big\}\p\Big\{U_{t-d}>b\Big(\frac{n}{ky}\Big)\Big\}\\
	&=\frac{k^2}{n^2}xy\Big[1+\frac{\alpha K}{\sqrt{k}}+o\Big(\frac{1}{\sqrt{k}}\Big)\Big]^2-\frac{k^2}{n^2}xy\\
	&=\frac{k^2}{n^2}xy\Big[1+\frac{2\alpha K}{\sqrt{k}}+o\Big(\frac{1}{\sqrt{k}}\Big)-1\Big]\\
	&\leq C\frac{k^{3/2}}{n^2}xy,
\end{align*}
whence, from \eqref{eq:Shao1},
\[
	\E\Big[\sum_{t=d+1}^{n}I_{1t}\Big]^4\leq C\frac{k^3}{n^2}+C\frac{k^{3/2}}{n}.
\]
Together with \eqref{eq:Cheby}, this implies that
\[
	\p\Big\{\sqrt{n}\sup_{\vxi\in[-K,0]^2}|M_n(\vxi)-M_n(\vzeros)|>\varepsilon\Big\}=O\Big(\frac{1}{k}+\frac{n}{k^{5/2}}\Big)=o(1)
\]
by Assumption~\ref{ass:k}. Hence, the conclusion follows.
\end{proof}

Define
\begin{align}
	A_t(\vxi,\vtheta) &= I_{\left\{\widehat{U}_t(\vtheta)>e^{\xi_1/\sqrt{k}}b(n/[kx]),\ \widehat{U}_{t-d}(\vtheta)>e^{\xi_2/\sqrt{k}}b(n/[ky])\right\}},\notag\\
	A_t(\vxi) &= I_{\left\{U_t>e^{\xi_1/\sqrt{k}}b(n/[kx]),\ U_{t-d}>e^{\xi_2/\sqrt{k}}b(n/[ky])\right\}},\notag\\
	A_t(\vxi,\eta,\eta_0) &= I_{\left\{U_t[1+\eta_0\s_{t}]+\eta_0 m_t>e^{\xi_1/\sqrt{k}}b(n/[kx]),\ U_{t-d}[1+\eta_0\s_{t-d}]+\eta_0 m_{t-d}>e^{\xi_2/\sqrt{k}}b(n/[ky])\right\}},\qquad \eta_0\in\{-1, 1\},\notag\\
	B_t(\vxi,\eta,\eta_0) &= B_t(\vx,\vxi,\eta,\eta_0)= A_t(\vxi,\eta,\eta_0) - A_t(\vxi),\label{eq:(A.10+)}
\end{align}
where $m_{t}=m_{n,t}(\eta)\geq0$ and $s_{t}=s_{n,t}(\eta)\geq0$ are from Assumption~\ref{ass:UA}.

The proof of Proposition~\ref{prop:hat convergence} requires two further preliminary lemmas, which are both proven in Appendix~\ref{Proofs of Auxiliary Lemmas}.

\begin{lem}\label{lem:wpa1}
Let $\eta>0$. Then, w.p.a.~1, as $n\to\infty$,
\begin{equation*}
	A_t(\vxi,\eta,-1) \leq A_t(\vxi,\vtheta) \leq A_t(\vxi,\eta, 1)
\end{equation*}
for all $\vtheta\in N_n(\eta)$ and $t=\ell_{n},\ldots,n$. 
\end{lem}

\begin{lem}\label{lem:Bop1}
Let $\eta>0$, $\eta_0\in\{-1, 1\}$. Then, for any $K>0$,
\begin{equation*}
\sup_{\vxi\in[-K,K]^2}\left|\frac{1}{k}\sum_{t=d+1}^{n}B_t(\vxi,\eta,\eta_0)\right|=o_{\p}(n^{-1/2}).
\end{equation*}
\end{lem}

\begin{proof}[{\textbf{Proof of Proposition~\ref{prop:hat convergence}:}}]
Let $\eta>0$ and $\ell_{n}\to\infty$ with $\ell_{n}=o(k/\sqrt{n})$ as $n\to\infty$. Then, due to Lemma~\ref{lem:wpa1}, we get that w.p.a.~1, as $n\to\infty$,
\[
	B_t(\vxi,\eta,-1)=A_t(\vxi,\eta,-1)-A_t(\vxi)\leq A_t(\vxi,\vtheta)-A_t(\vxi) \leq A_t(\vxi,\eta,1)-A_t(\vxi)=B_t(\vxi,\eta,1)
\]
for all $\vtheta\in N_n(\eta)$ and $t=\ell_{n},\ldots,n$. Hence, from Lemma~\ref{lem:Bop1},
\[
	\sup_{\vtheta\in N_n(\eta)}\sup_{\vxi\in[-K, K]^2}\frac{1}{k}\left|\sum_{t=\ell_{n}}^{n}\left\{A_t(\vxi,\vtheta)-A_t(\vxi)\right\}\right|=o_{\p}(n^{-1/2}).
\]
Using additionally that $\ell_{n}/k=o(n^{-1/2})=o(1)$, we also have that (since $|A_t(\vxi,\vtheta)-A_t(\vxi)|\leq 1$)
\[
	\sup_{\vtheta\in N_n(\eta)}\sup_{\vxi\in[-K, K]^2}\frac{1}{k}\left|\sum_{t=d+1}^{n}\left\{A_t(\vxi,\vtheta)-A_t(\vxi)\right\}\right|=o_{\p}(n^{-1/2}).
\]
Because by Assumption~\ref{ass:estimator}, $\widehat{\vtheta}\in N_n(\eta)$ w.p.a.~1 as $n\to\infty$ followed by $\eta\to\infty$, we get
\[
	\sup_{\vxi\in[-K, K]^2}\frac{1}{k}\left|\sum_{t=d+1}^{n}\left\{A_t(\vxi,\widehat{\vtheta})-A_t(\vxi)\right\}\right|=o_{\p}(n^{-1/2}).
\]
This, however, is just the conclusion.
\end{proof}

\begin{proof}[{\textbf{Proof of Proposition~\ref{prop:emp an}:}}]
The result follows immediately from Proposition~\ref{prop:unif emp an}, which is proven in Appendix~\ref{Proofs of Propositions 2}.
\end{proof}

\subsection{Proofs of Lemmas~\ref{lem:Lem F}--\ref{lem:Bop1}}\label{Proofs of Auxiliary Lemmas}

\begin{proof}[{\textbf{Proof of Lemma~\ref{lem:Lem F}:}}]
By \citet[Theorem~1.1.11 and Theorem~B.1.10]{HF06}, we have that for any $\widetilde{\widetilde{\nu}}\in(-1,\infty)$, any $\varepsilon>0$ and any $\delta>0$
\begin{equation}\label{eq:unif in F}
	\left|\frac{\overline{F}\Big((1+\widetilde{\widetilde{\nu}})b\Big(\frac{n}{kx}\Big)\Big)}{\overline{F}\Big(b\Big(\frac{n}{kx}\Big)\Big)}-(1+\widetilde{\widetilde{\nu}})^{-\alpha}\right|\leq\varepsilon\max\{(1+\widetilde{\widetilde{\nu}})^{-\alpha-\delta}, (1+\widetilde{\widetilde{\nu}})^{-\alpha+\delta}\}
\end{equation}
for sufficiently large $n$. Use the mean value theorem to deduce that there is some $\widetilde{\nu}$ between $1$ and $\nu$, such that
\begin{align}
	\overline{F}\Big(b\Big(\frac{n}{kx}\Big)\Big) &- \overline{F}\Big((1+\nu)b\Big(\frac{n}{kx}\Big)\Big)\notag\\
	&=F\Big((1+\nu)b\Big(\frac{n}{kx}\Big)\Big) - F\Big(b\Big(\frac{n}{kx}\Big)\Big)\notag\\
	&=f\Big((1+\widetilde{\nu})b\Big(\frac{n}{kx}\Big)\Big)\left[(1+\nu)b\Big(\frac{n}{kx}\Big)-b\Big(\frac{n}{kx}\Big)\right]\notag\\
	&= \frac{(1+\widetilde{\nu})b\Big(\frac{n}{kx}\Big)f\Big((1+\widetilde{\nu})b\Big(\frac{n}{kx}\Big)\Big)}{1-F\Big((1+\widetilde{\nu})b\Big(\frac{n}{kx}\Big)\Big)}\cdot\frac{(1+\nu)b\Big(\frac{n}{kx}\Big)-b\Big(\frac{n}{kx}\Big)}{(1+\widetilde{\nu})b\Big(\frac{n}{kx}\Big)}\cdot\frac{\overline{F}\Big((1+\widetilde{\nu})b\Big(\frac{n}{kx}\Big)\Big)}{\overline{F}\Big(b\Big(\frac{n}{kx}\Big)\Big)}\cdot\overline{F}\Big(b\Big(\frac{n}{kx}\Big)\Big)\notag\\
	&=\big[\alpha+o(1)\big]\frac{1+\nu-1}{1+\widetilde{\nu}}\big[(1+\widetilde{\nu})^{-\alpha}+o(1)\big]\frac{kx}{n}\notag\\
	&=\big[1+o(1)\big]\Bigg[\frac{\alpha\nu}{(1+\widetilde{\nu})^{1+\alpha}}+\nu o(1)\Bigg]\frac{kx}{n}\notag\\
	&=\Bigg[\frac{\alpha\nu}{(1+\widetilde{\nu})^{1+\alpha}}+\nu o(1)\Bigg]\frac{kx}{n},\label{eq:(B.X)}
\end{align}
where the fourth step uses that the convergence
\[
	\frac{\overline{F}\Big((1+\widetilde{\nu})b\Big(\frac{n}{kx}\Big)\Big)}{\overline{F}\Big(b\Big(\frac{n}{kx}\Big)\Big)}-(1+\widetilde{\nu})^{-\alpha}\longrightarrow0
\]
is uniform in $\widetilde{\nu}$ and uniform in $x$ by \eqref{eq:unif in F}. By continuity of $F(\cdot)$ in the tail (see Assumption~\ref{ass:U distr 2}~(ii)), \eqref{eq:(B.X)} implies
\begin{align*}
\overline{F}\Big((1+\nu)b\Big(\frac{n}{kx}\Big)\Big) &= \overline{F}\Big(b\Big(\frac{n}{kx}\Big)\Big)-\Big[\frac{\alpha\nu}{(1+\widetilde{\nu})^{1+\alpha}}+\nu o(1)\Big]\frac{kx}{n}\\
&= \frac{kx}{n}\Big[1-\frac{\alpha\nu}{(1+\widetilde{\nu})^{1+\alpha}}+\nu o(1)\Big].
\end{align*}
To prove the second statement, it suffices to recall the Taylor expansion $e^{\xi/\sqrt{k}}=1+\xi/\sqrt{k}+o(1/\sqrt{k})$ and apply a Taylor expansion to $(1+\widetilde{\nu})^{1+\alpha}$.
\end{proof}

\begin{proof}[{\textbf{Proof of Lemma~\ref{lem:wpa1}:}}]
We have to show that w.p.a.~1, as $n\to\infty$,
\begin{align*}
	&I_{\left\{U_t[1-\s_{t}]-m_t>e^{\xi_1/\sqrt{k}}b(n/[kx]),\ U_{t-d}[1-\s_{t-d}]-m_{t-d}>e^{\xi_2/\sqrt{k}}b(n/[ky])\right\}}\\
	 &\hspace{1cm}\leq I_{\left\{\widehat{U}_t(\vtheta)>e^{\xi_1/\sqrt{k}}b(n/[kx]),\ \widehat{U}_{t-d}(\vtheta)>e^{\xi_2/\sqrt{k}}b(n/[ky])\right\}}\\
	&\hspace{1cm}\leq I_{\left\{U_t[1+\s_{t}]+m_t>e^{\xi_1/\sqrt{k}}b(n/[kx]),\ U_{t-d}[1+\s_{t-d}]+m_{t-d}>e^{\xi_2/\sqrt{k}}b(n/[ky])\right\}}
\end{align*}
for all $\vtheta\in N_n(\eta)$ and $t=\ell_{n},\ldots,n$. However, this follows immediately from Assumption~\ref{ass:UA}.
\end{proof}

\begin{proof}[{\textbf{Proof of Lemma~\ref{lem:Bop1}:}}]
Let $\eta_0=1$; the case $\eta_0=-1$ can be dealt with similarly. Define
\[
	w_t=I_{\left\{m_{t}\leq\varepsilon_0,\ m_{t-d}\leq\varepsilon_0,\ s_{t}\leq\varepsilon_0,\ s_{t-d}\leq\varepsilon_0\right\}},\qquad\varepsilon_0>0.
\]
If $w_t=1$, there exists $\nu=\nu(\varepsilon_0)>0$, such that
\begin{equation*}
	1-\nu/2<(1+\s_{t})^{-1}\leq1\qquad\text{and}\qquad 1-\nu/2<(1+\s_{t-d})^{-1}\leq1.
\end{equation*}
Furthermore, $b(x)\to\infty$ for $x\to\infty$ by Assumption~\ref{ass:U distr 2}, such that
\[
	1-\nu/2<1-e^{K/\sqrt{k}}\varepsilon_0/b\big(n/[k\max\{x,\, y\}]\big)\leq1
\]
for sufficiently large $n$. Then, for $\vxi\in[-K,0]^2$ (the other quadrants can be dealt with similarly) and sufficiently large $n$,
\begin{align*}
&0\leq\frac{\sqrt{n}}{k}\sum_{t=\ell_{n}}^{n}w_tB_t(\vxi,\eta,\eta_0)\\
&=\frac{\sqrt{n}}{k}\sum_{t=\ell_{n}}^{n}w_t\Big[I_{\Big\{U_t>e^{\xi_1/\sqrt{k}}(1+\s_{t})^{-1}\big(1-e^{-\xi_1/\sqrt{k}}\frac{m_t}{b(n/[kx])}\big) b(n/[kx]),}\\
&\hspace{3cm} _{U_{t-d}>e^{\xi_2/\sqrt{k}}(1+\s_{t-d})^{-1}\big(1-e^{-\xi_2/\sqrt{k}}\frac{m_{t-d}}{b(n/[ky])}\big)b(n/[ky])\Big\}}\\
&\hspace{8.22cm}-I_{\left\{U_t>e^{\xi_1/\sqrt{k}}b(n/[kx]),\ U_{t-d}>e^{\xi_2/\sqrt{k}}b(n/[ky])\right\}}\Big]\\
&\leq \frac{\sqrt{n}}{k}\sum_{t=\ell_{n}}^{n}\Big[I_{\left\{U_t>(1-\nu) b(n/[kx]),\ U_{t-d}>(1-\nu)b(n/[ky])\right\}}-I_{\left\{U_t>b(n/[kx]),\ U_{t-d}>b(n/[ky])\right\}}\Big]\\
&=: \frac{\sqrt{n}}{k}\sum_{t=\ell_{n}}^{n}I_{2t}.
\end{align*}
Use Markov's inequality and Lemma~2.3 of \citet{Sha93}, to obtain that
\begin{align}
\p\Big\{\frac{\sqrt{n}}{k}\sum_{t=\ell_{n}}^{n}I_{2t}>\varepsilon\Big\}&\leq\varepsilon^{-4}\frac{n^2}{k^4}\E\Big[\sum_{t=\ell_{n}}^{n}I_{2t}\Big]^4\notag\\
&=C\frac{n^2}{k^4}\Big\{n^2\big\{\E[I_{2t}^2]\big\}^2 + n\E[I_{2t}^4] \Big\}.\label{eq:Cheby2}
\end{align}
Similarly as for $I_{1t}$ in the proof of Proposition~\ref{prop:xi convergence}, we obtain using Lemma~\ref{lem:Lem F} that for sufficiently large $n$,
\begin{align*}
\E\big[I_{2t}\big] &= \p\left\{U_t>(1-\nu)b(n/[kx])\right\}\p\left\{U_{t-d}>(1-\nu)b(n/[ky])\right\}\\
&\hspace{6cm}-\p\left\{U_t>b(n/[kx])\right\}\p\left\{U_{t-d}>b(n/[ky])\right\}\\
&=\frac{k^2}{n^2}xy\Big[1+\frac{\alpha\nu}{(1+\widetilde{\nu})^{1+\alpha}}+\nu o(1)\Big]^2-\frac{k^2}{n^2}xy\\
&\leq C\frac{k^2}{n^2}\nu.
\end{align*}
Thus, the right-hand side of \eqref{eq:Cheby2} can be bounded by
\begin{equation}\label{eq:HHH}
	C\frac{n^2}{k^4}\Big[n^2\frac{k^4}{n^4}+n\frac{k^2}{n^2}\Big]\nu=O(1)\nu.
\end{equation}
For a suitable choice of $\varepsilon_0$ in $w_t$, $\nu$ in \eqref{eq:HHH} can be chosen arbitrarily close to zero. Thus, the expression on the right-hand side can be made arbitrarily small, yielding that 
\begin{equation*}
	\sup_{\vxi\in[-K,0]^2}\frac{\sqrt{n}}{k}\Bigg|\sum_{t=\ell_{n}}^{n}w_tB_t(\vxi,\eta,\eta_0=1)\Bigg|=o_{\p}(1). 
\end{equation*}
By Assumption~\ref{ass:UA}, $w_{\ell_{n}}=\ldots=w_{n}=1$ w.p.a.~1, as $n\to\infty$. Thus, we obtain
\begin{equation}\label{eq:(13.1)}
\sup_{\vxi\in[-K,0]^2}\frac{\sqrt{n}}{k}\Bigg|\sum_{t=\ell_{n}}^{n}B_t(\vxi,\eta,\eta_0=1)\Bigg|=o_{\p}(1). 
\end{equation}
By choosing $\ell_{n}\to\infty$ with $\ell_{n}=o(k/\sqrt{n})$ (recall Assumption~\ref{ass:k}), we also have by boundedness of the $B_t(\vxi,\eta,\eta_0=1)$ that
\begin{equation}\label{eq:(13.2)}
\sup_{\vxi\in[-K,0]^2}\frac{\sqrt{n}}{k}\Bigg|\sum_{t=d+1}^{\ell_{n}-1}B_t(\vxi,\eta,\eta_0=1)\Bigg|\leq \frac{\sqrt{n}}{k}\ell_{n}=o(1).
\end{equation}
Combining \eqref{eq:(13.1)} and \eqref{eq:(13.2)} gives the desired conclusion.
\end{proof}

%
%

\section{Proof of Theorem~\ref{thm:mainresult2}}\label{Proof of Theorem 2}

\renewcommand{\theequation}{D.\arabic{equation}}	
\setcounter{equation}{0}	

The proof of Theorem~\ref{thm:mainresult2} is structured similarly as that of Theorem~\ref{thm:mainresult}. It requires four preliminary propositions, whose proofs are relegated to Appendix~\ref{Proofs of Propositions 2}. We follow the same notational conventions as in Appendix~\ref{Proof of Theorem 1}. Furthermore, we define $D[0,1]^2$ to be the space of $\mathbb{R}^{D}$-valued functions on $[0,1]\times[0,1]$ that are continuous from above, with limits from below \citep[see][Sec.~3]{BW71}.

\begin{prop}\label{prop:base functional convergence}
For i.i.d.~$U_t$ with d.f.~$F(\cdot)$ satisfying Assumption~\ref{ass:U distr 2}~(ii), it holds that, as $n\to\infty$,
\begin{equation*}
	\mM_n(x,y)=\sqrt{n}\begin{pmatrix}
	\frac{1}{k}\sum_{t=2}^{n}\Big[I_{\left\{U_t>b(\frac{n}{kx}),\ U_{t-1}>b(\frac{n}{ky})\right\}}
	-\Big(\frac{k}{n}\Big)^2xy\Big]\\
	\vdots \\
	\frac{1}{k}\sum_{t=D+1}^{n}\Big[I_{\left\{U_t>b(\frac{n}{kx}),\ U_{t-D}>b(\frac{n}{ky})\right\}}
	-\Big(\frac{k}{n}\Big)^2xy\Big]\\
	\end{pmatrix}\overset{d}{\longrightarrow}\mB(x,y)\quad\text{in}\ D[0,1]^2,
\end{equation*}
where $\mB=(B_1,\ldots,B_D)^\prime$ is a $D$-dimensional Brownian sheet with independent components, i.e., a zero-mean Gaussian process with $\Cov\big(\mB(x_1,y_1), \mB(x_2,y_2)\big)=\min(x_1,x_2)\min(y_1,y_2)\mI_{D\times D}$.
\end{prop}

For the proof of Theorem~\ref{thm:mainresult2}, we have to justify the replacement of $U_t$ by $\widehat{U}_t$, and $b(\frac{n}{kz})$ by $\widehat{U}_{(\lfloor kz\rfloor +1)}$ ($z\in\{x,y\}$) in the indicators in Proposition~\ref{prop:base functional convergence}. The following three propositions serve that purpose.

\begin{prop}\label{prop:Mn unif x xi}
For any $K>0$, it holds that
\[
	\sup_{\vxi\in[-K, K]^2}\sup_{\vx\in[0,K]^2}|M_n(\vx,\vxi)-M_n(\vx,\vzeros)|=o_{\p}(n^{-1/2}).
\]
\end{prop}

\begin{prop}\label{prop:hat unif convergence}
For any $K>0$, it holds that
\[
	\sup_{\vxi\in[-K, K]^2}\sup_{\vx\in[0,K]^2}|\widehat{M}_n(\vx,\vxi)-M_n(\vx,\vxi)|=o_{\p}(n^{-1/2}).
\]
\end{prop}

\begin{prop}\label{prop:unif emp an}
For any $0<\iota<K<\infty$, it holds that
\[
	\sup_{x\in[\iota,K]}|\xi_n(x)|=\sup_{x\in[\iota,K]}\left|\sqrt{k}\log\Bigg(\frac{\widehat{U}_{(\lfloor kx\rfloor+1)}}{b(n/[kx])}\Bigg)\right|=O_{\p}(1).
\]
\end{prop}

These four propositions allow us to prove Theorem~\ref{thm:mainresult2}.

\begin{proof}[{\textbf{Proof of Theorem~\ref{thm:mainresult2}:}}]
The outline of the proof is similar to that of Theorem~\ref{thm:mainresult}. We show that for each $d\in\{1,\ldots,D\}$,
\begin{equation}\label{eq:to show}
	\sup_{\vx=(x,y)^\prime\in[\iota,1]^2}\Bigg|\frac{1}{k}\sum_{t=d+1}^{n}I_{\{U_t>b(\frac{n}{kx}),\ U_{t-d}>b(\frac{n}{ky})\}}-\frac{1}{k}\sum_{t=d+1}^{n}I_{\{\widehat{U}_t>\widehat{U}_{(\lfloor kx\rfloor+1)},\ \widehat{U}_{t-d}>\widehat{U}_{(\lfloor ky\rfloor+1)}\}}\Bigg|=o_{\p}(n^{-1/2}).
\end{equation}
Combining this with Proposition~\ref{prop:base functional convergence}, Lemma~4.7 of \citet{Whi01} and the continuous mapping theorem \citep[e.g.,][Theorem~7.20]{Whi01} gives us
\[
	\mathcal{F}_n^{(D)}\overset{d}{\longrightarrow}\sum_{d=1}^{D}\int_{[\iota,1-\iota]}W_d^2(2-2z, 2z)\D z,
\]
where $\{W_d(\cdot,\cdot)\}_{d=1,\ldots,D}$ are mutually independent Brownian sheets, i.e., zero-mean Gaussian processes with covariance function $\Cov\big(W_d(x_1, y_1), W_d(x_2, y_2)\big)=\min(x_1, x_2)\min(y_1, y_2)$. Computing the covariance function of $\{W_d(2-2z, 2z)\}_{z\in[0,1]}$, we find that $\{W_d(2-2z, 2z)\}_{z\in[0,1]}\overset{d}{=}\{2B_d(z)\}_{z\in[0,1]}$, where $B_d(\cdot)$ is a standard Brownian bridge. Due to this,
\[
	\sum_{d=1}^{D}\int_{[\iota,1-\iota]}W_d^2(2-2z, 2z)\D z\overset{d}{=}4\sum_{d=1}^{D}\int_{[\iota,1-\iota]}B_d^2(z)\D z
\]
and, hence, the limit is as claimed.

So it remains to show \eqref{eq:to show}. Decompose
\begin{align}
M_n(\vx,\vzeros) &= \frac{1}{k}\sum_{t=d+1}^{n}I_{\{U_t>b(\frac{n}{kx}),\ U_{t-d}>b(\frac{n}{ky})\}} \notag\\
&= \frac{1}{k}\sum_{t=d+1}^{n}I_{\{U_t>b(\frac{n}{kx}),\ U_{t-d}>b(\frac{n}{ky})\}}-\frac{1}{k}\sum_{t=d+1}^{n}I_{\{U_t>e^{\xi_1/\sqrt{k}}b(\frac{n}{kx}),\ U_{t-d}>e^{\xi_2/\sqrt{k}}b(\frac{n}{ky})\}}\notag\\
&\hspace{0.4cm} +\frac{1}{k}\sum_{t=d+1}^{n}I_{\{U_t>e^{\xi_1/\sqrt{k}}b(\frac{n}{kx}),\ U_{t-d}>e^{\xi_2/\sqrt{k}}b(\frac{n}{ky})\}} - \frac{1}{k}\sum_{t=d+1}^{n}I_{\{\widehat{U}_t>e^{\xi_1/\sqrt{k}}b(\frac{n}{kx}),\ \widehat{U}_{t-d}>e^{\xi_2/\sqrt{k}}b(\frac{n}{ky})\}}\notag\\
&\hspace{0.4cm} +\frac{1}{k}\sum_{t=d+1}^{n}I_{\{\widehat{U}_t>e^{\xi_1/\sqrt{k}}b(\frac{n}{kx}),\ \widehat{U}_{t-d}>e^{\xi_2/\sqrt{k}}b(\frac{n}{ky})\}}\notag\\
&=o_{\p}(n^{-1/2}) + \widehat{M}_n(\vx,\vxi),\label{eq:(2.1)}
\end{align}
where the $o_{\p}(n^{-1/2})$-term is uniform on $\vxi\in[-K,K]^2$ and $\vx\in[\iota,1]^2$ by Propositions~\ref{prop:Mn unif x xi} and~\ref{prop:hat unif convergence}. 

Define $\vxi_n:=\vxi_n(\vx):=\big(\xi_n(x),\ \xi_n(y)\big)^\prime$, such that $\widehat{M}_n(\vx,\vxi_n)=\frac{1}{k}\sum_{t=d+1}^{n}I_{\{\widehat{U}_t>\widehat{U}_{(\lfloor kx\rfloor+1)},\ \widehat{U}_{t-d}>\widehat{U}_{(\lfloor ky\rfloor+1)}\}}$. Then, to show \eqref{eq:to show} or, equivalently, that $\sup_{\vx\in[\iota,1]^2}|\widehat{M}_n(\vx,\vxi_n)-M_n(\vx,\vzeros)|=o_{\p}(n^{-1/2})$, we have to prove that for any $\varepsilon>0$ and $\delta>0$
\[
	\p\left\{\sqrt{n}\sup_{\vx\in[\iota,1]^2}|\widehat{M}_n(\vx,\vxi_n)-M_n(\vx,\vzeros)|>\varepsilon\right\}\leq\delta
\]
for sufficiently large $n$. By Proposition~\ref{prop:unif emp an}, we can choose $K>0$ such that $\p\{\sup_{\vx\in[\iota,1]^2}|\vxi_n|>K\}\leq\delta/2$ for sufficiently large $n$. Furthermore, by \eqref{eq:(2.1)}, $\p\{\sqrt{n}\sup_{\vxi\in[-K, K]^2}\sup_{\vx\in[\iota,1]^2}|\widehat{M}_n(\vx,\vxi)-M_n(\vx,\vzeros)|>\varepsilon\}\leq\delta/2$ for large $n$. Using these two results, we obtain
\begin{align*}
	\p&\Bigg\{\sqrt{n}\sup_{\vx\in[\iota,1]^2}|\widehat{M}_n(\vx,\vxi_n)-M_n(\vx,\vzeros)|>\varepsilon\Bigg\}\\
	&\leq\p\Bigg\{\sup_{\vx\in[\iota,1]^2}|\vxi_n|>K\Bigg\} + \p\Bigg\{\sqrt{n}\sup_{\vxi\in[-K, K]^2}\sup_{\vx\in[\iota,1]^2}|\widehat{M}_n(\vx,\vxi)-M_n(\vx,\vzeros)|>\varepsilon,\ \sup_{\vx\in[\iota,1]^2}|\vxi_n|\leq K\Bigg\}\\
	&\leq\frac{\delta}{2}+\frac{\delta}{2}=\delta,
\end{align*}
i.e., \eqref{eq:to show}. The conclusion follows.
\end{proof}

\subsection{Proofs of Propositions~\ref{prop:base functional convergence}--\ref{prop:unif emp an}}\label{Proofs of Propositions 2}

\begin{proof}[{\textbf{Proof of Proposition~\ref{prop:base functional convergence}:}}]
The corollary on p.~1664 of \citet{BW71} shows that to prove the claim, it suffices to show convergence of the finite-dimensional distributions and tightness. First, to prove convergence of the finite-dimensional distributions, we have to show that for all pairwise distinct $(x_i, y_i)\in[0,1]^2$, $i=1,\ldots,m$,
\begin{equation}\label{eq:(1.1)}
	\big(\mM_n(x_1,y_1),\ldots,\mM_n(x_m, y_m)\big)^\prime\overset{d}{\longrightarrow}\big(\mB(x_1,y_1),\ldots,\mB(x_m,y_m)\big)^\prime.
\end{equation}
Second, to prove tightness it suffices by Theorem~26.23 of \citet{Dav94} to do so for each of the marginals separately. To that end, denote the components of $\mM_n(x,y)$ by $M_n^{(d)}(x,y)$ $(d=1,\ldots,D)$ and define
\begin{align*}
B&= (x_1,x_2]\times (y_1, y_2],\\
M_n^{(d)}(B) &= M_n^{(d)}(x_2,y_2) - M_n^{(d)}(x_1,y_2) - M_n^{(d)}(x_2,y_1) + M_n^{(d)}(x_1,y_1).
\end{align*}
By their Theorem~3 and the subsequent comment together with their Equation~(3), \citet{BW71} show that for tightness it suffices to prove that for $d\in\{1,\ldots,D\}$
\begin{align}
	\E\Big[\big\{M_n^{(d)}(B_1)\big\}^2\big\{M_n^{(d)}(B_2)\big\}^2\Big] &\leq K(x_2-x_1)^2(y_2-y)(y-y_1),\label{eq:(B.4m)}\\
	\E\Big[\big\{M_n^{(d)}(B_3)\big\}^2\big\{M_n^{(d)}(B_4)\big\}^2\Big] &\leq K(x-x_1)(x_2-x)(y_2-y_1)^2,\label{eq:tight2}
\end{align}
where, for $0\leq x_1<x\leq x_2\leq 1$ and $0\leq y_1<y\leq y_2\leq 1$,
\begin{align*}
B_1&= (x_1,x_2]\times (y_1, y],\quad B_2= (x_1,x_2]\times (y, y_2],\\
B_3&= (x_1,x]\times (y_1, y_2],\quad B_4= (x,x_2]\times (y_1, y_2].
\end{align*}

We first show convergence of the finite-dimensional distributions, i.e., \eqref{eq:(1.1)}. We do so for $m=2$; the case $m>2$ being only notationally more complicated. By the Cram\'{e}r--Wold device \citep[Thm.~25.5]{Dav94} it suffices to show that, as $n\to\infty$,
\begin{equation*}
\vlambda_1^\prime\mM_n(x_1,y_1)+\vlambda_2^\prime\mM_n(x_2,y_2)\overset{d}{\longrightarrow}\vlambda_1^\prime\mB(x_1,y_1)+\vlambda_2^\prime\mB(x_2,y_2)
\end{equation*}
for any $\vlambda=(\vlambda_1^\prime, \vlambda_2^\prime)^\prime=(\lambda_{1,1},\ldots,\lambda_{1,D},\lambda_{2,1},\ldots,\lambda_{2,D})^\prime$ with $\vlambda^\prime\vlambda=1$. The proof of this convergence is similar to that of Proposition~\ref{prop:base convergence}, and uses Theorem~5.20 of \citet{Whi01}. We re-define
\begin{align}
	\overline{\sigma}_n^2&=\Var\Big(\vlambda_1^\prime\mM_n(x_1,y_1)+\vlambda_2^\prime\mM_n(x_2,y_2)\Big)\notag\\
	&= \Var\Big(\vlambda_1^\prime\mM_n(x_1,y_1)\Big)+\Var\Big(\vlambda_2^\prime\mM_n(x_2,y_2)\Big)+2\Cov\Big(\vlambda_1^\prime\mM_n(x_1,y_1),\ \vlambda_2^\prime\mM_n(x_2,y_2)\Big).\label{eq:(3.1)}
\end{align}
From the proof of Proposition~\ref{prop:base convergence}, we get that
\begin{equation}\label{eq:(C.14)}
	\Var\Big(\vlambda_i^\prime\mM_n(x_i,y_i)\Big)=\sum_{d=1}^{D}\lambda_{i,d}^2x_iy_i+o(1) \qquad (i=1,2).
\end{equation}
Thus, it remains to compute the covariance term. Recall from the proof of Proposition~\ref{prop:base convergence} that
\[
	Z_{n,t}^{(d)}(x,y)=\begin{cases}0, & t=1,\ldots,d,\\
	\frac{n}{k}\Big[I_{\left\{U_t>b(\frac{n}{kx}),\ U_{t-d}>b(\frac{n}{ky})\right\}}
	-\Big(\frac{k}{n}\Big)^2xy\Big],& t>d.
	\end{cases}
\]
With this
\begin{align}
	&\Cov\Big(\vlambda_1^\prime\mM_n(x_1,y_1),\ \vlambda_2^\prime\mM_n(x_2,y_2)\Big) \notag\\
	&\hspace{1cm}= \Cov\Big(\sqrt{n}\sum_{d_1=1}^{D}\frac{1}{n}\sum_{t=1}^{n}\lambda_{1,d_1}Z_{n,t}^{(d_1)}(x_1,y_1),\ \sqrt{n}\sum_{d_2=1}^{D}\frac{1}{n}\sum_{t=1}^{n}\lambda_{2,d_2}Z_{n,t}^{(d_2)}(x_2,y_2)\Big)\notag\\
	&\hspace{1cm}= \sum_{d_1=1}^{D}\sum_{d_2=1}^{D}\frac{1}{n}\Cov\Big(\sum_{t=1}^{n}\lambda_{1,d_1}Z_{n,t}^{(d_1)}(x_1,y_1),\ \sum_{t=1}^{n}\lambda_{2,d_2}Z_{n,t}^{(d_2)}(x_2,y_2)\Big).\label{eq:(4.1)}
\end{align}
For $d_1\neq d_2$, we obtain that
\begin{align*}
	\frac{1}{n}\Cov&\Big(\sum_{t=1}^{n}\lambda_{1,d_1}Z_{n,t}^{(d_1)}(x_1,y_1),\ \sum_{t=1}^{n}\lambda_{2,d_2}Z_{n,t}^{(d_2)}(x_2,y_2)\Big)\\
	&= \frac{1}{n}\sum_{s=1}^{n}\sum_{t=1}^{n}\lambda_{1,d_1}\lambda_{2,d_2}\Cov\Big(Z_{n,s}^{(d_1)}(x_1,y_1),\ Z_{n,t}^{(d_2)}(x_2,y_2)\Big)\\	
	&=\frac{1}{n}\Bigg\{\sum_{t=1}^{n}\lambda_{1,d_1}\lambda_{2,d_2}\Cov\Big(Z_{n,t}^{(d_1)}(x_1,y_1),\ Z_{n,t}^{(d_2)}(x_2,y_2)\Big)\\
	&\hspace{1cm}+\Bigg(\sum_{s-t=d_1} + \sum_{t-s=d_2} + \sum_{t-s=d_2-d_1} \Bigg)\lambda_{1,d_1}\lambda_{2,d_2}\Cov\Big(Z_{n,s}^{(d_1)}(x_1,y_1),\ Z_{n,t}^{(d_2)}(x_2,y_2)\Big)\Bigg\}\\
	&=O(k/n)=o(1)
\end{align*}
by a relation similar to \eqref{eq:cov bound in sum}. This shows that the right-hand side of \eqref{eq:(4.1)} equals
\begin{align*}
\sum_{d=1}^{D}\frac{1}{n}&\Cov\Big(\sum_{t=1}^{n}\lambda_{1,d}Z_{n,t}^{(d)}(x_1,y_1),\ \sum_{t=1}^{n}\lambda_{2,d}Z_{n,t}^{(d)}(x_2,y_2)\Big)+o(1)\\
&=\sum_{d=1}^{D}\frac{1}{n}\sum_{t=1}^{n}\lambda_{1,d}\lambda_{2,d}\Cov\Big(Z_{n,t}^{(d)}(x_1,y_1),\ Z_{n,t}^{(d)}(x_2,y_2)\Big)\\
&\hspace{1cm}+\sum_{d=1}^{D}\frac{1}{n}\sum_{|t-s|=d}\lambda_{1,d}\lambda_{2,d}\Cov\Big(Z_{n,s}^{(d)}(x_1,y_1),\ Z_{n,t}^{(d)}(x_2,y_2)\Big)+o(1)\\
&=\sum_{d=1}^{D}\lambda_{1,d}\lambda_{2,d}\min(x_1,x_2)\min(y_1,y_2)+o(1),
\end{align*}
where the last line follows from \eqref{eq:cov bound in sum} and the fact that
\begin{align*}
\Cov&\Big(Z_{n,t}^{(d)}(x_1,y_1),\ Z_{n,t}^{(d)}(x_2,y_2)\Big)=\E\big[Z_{n,t}^{(d)}(x_1,y_1) Z_{n,t}^{(d)}(x_2,y_2)\big]-\E\big[Z_{n,t}^{(d)}(x_1,y_1)\big]\E\big[Z_{n,t}^{(d)}(x_2,y_2)\big]\\
&=\Big(\frac{n}{k}\Big)^2\p\Big\{U_t>b\big(n/[k\min(x_1,x_2)]\big),\ U_{t-d}>b\big(n/[k\min(y_1,y_2)]\big)\Big\}-\Big(\frac{n}{k}\Big)^2\Big(\frac{k}{n}\Big)^2x_1y_1 \Big(\frac{k}{n}\Big)^2x_2y_2\\
&= \min(x_1,x_2)\min(y_1,y_2)+o(1).
\end{align*}
Overall, we get that
\[
	\Cov\Big(\vlambda_1^\prime\mM_n(x_1,y_1),\ \vlambda_2^\prime\mM_n(x_2,y_2)\Big)=\sum_{d=1}^{D}\lambda_{1,d}\lambda_{2,d}\min(x_1,x_2)\min(y_1,y_2)+o(1).
\]
Plugging this and \eqref{eq:(C.14)} into \eqref{eq:(3.1)} gives
\begin{equation*}
\overline{\sigma}_n^2 = \sum_{d=1}^{D}\big[\lambda_{1,d}^2 x_1y_1 + \lambda_{2,d}^2x_2y_2+2\lambda_{1,d}\lambda_{2,d}\min(x_1,x_2)\min(y_1,y_2)\big]+o(1).
\end{equation*}
The limit of $\overline{\sigma}_n^2$ then is easily seen to equal
\begin{align*}
\Var&\big(\vlambda_1^\prime\mB(x_1,y_1)+\vlambda_2^\prime\mB(x_2,y_2)\big)\\
&=\Var\big(\vlambda_1^\prime\mB(x_1,y_1)\big)+\Var\big(\vlambda_2^\prime\mB(x_2,y_2)\big)+2\Cov\big(\vlambda_1^\prime\mB(x_1,y_1),\ \vlambda_2^\prime\mB(x_2,y_2)\big)\\
&=\sum_{d=1}^{D}\lambda_{1,d}^2\Var\big(B_d(x_1,y_1)\big) + \sum_{d=1}^{D}\lambda_{2,d}^2\Var\big(B_d(x_1,y_1)\big)\\
&\hspace{5cm}+2\Cov\Big(\sum_{d=1}^{D}\lambda_{1,d}B_d(x_1, y_1),\ \sum_{d=1}^{D}\lambda_{2,d}B_d(x_2, y_2)\Big)\\
&=\sum_{d=1}^{D}\big[\lambda_{1,d}^2 x_1y_1 + \lambda_{2,d}^2x_2y_2+2\lambda_{1,d}\lambda_{2,d}\min(x_1,x_2)\min(y_1,y_2)\big],
\end{align*}
as desired. The remaining conditions of Theorem~5.20 of \citet{Whi01} can be checked easily and, hence, \eqref{eq:(1.1)} follows.  

It remains to show tightness of the marginals. We only show \eqref{eq:(B.4m)}, because \eqref{eq:tight2} can be proved in exactly the same manner. Set
\begin{align*}
G_t &= I_{\Big\{U_t\in\big(b(\frac{n}{kx_2}),\ b(\frac{n}{kx_1})\big],\ U_{t-d}\in\big(b(\frac{n}{ky}),\ b(\frac{n}{ky_1})\big]\Big\}}-\Big(\frac{k}{n}\Big)^2(x_2-x_1)(y-y_1)\\
&= I_{\Big\{U_t\in\big(b(\frac{n}{kx_2}),\ b(\frac{n}{kx_1})\big],\ U_{t-d}\in\big(b(\frac{n}{ky}),\ b(\frac{n}{ky_1})\big]\Big\}}-p_{G},\\
H_t &= I_{\Big\{U_t\in\big(b(\frac{n}{kx_2}),\ b(\frac{n}{kx_1})\big],\ U_{t-d}\in\big(b(\frac{n}{ky_2}),\ b(\frac{n}{ky})\big]\Big\}}-\Big(\frac{k}{n}\Big)^2(x_2-x_1)(y_2-y)\\
&= I_{\Big\{U_t\in\big(b(\frac{n}{kx_2}),\ b(\frac{n}{kx_1})\big],\ U_{t-d}\in\big(b(\frac{n}{ky_2}),\ b(\frac{n}{ky})\big]\Big\}}-p_{H}.
\end{align*}
 It is easy to check that
\begin{align*}
M_n^{(d)}(B_1) &= M_n^{(d)}(x_2,y) - M_n^{(d)}(x_1,y) - M_n^{(d)}(x_2,y_1) + M_n^{(d)}(x_1,y_1)\\
&=\frac{\sqrt{n}}{k}\sum_{t=d+1}^{n}G_t,\\
M_n^{(d)}(B_2)&=\frac{\sqrt{n}}{k}\sum_{t=d+1}^{n}H_t.
\end{align*}
Thus, we can write
\begin{align}
\E&\Big[\big\{M_n^{(d)}(B_1)\big\}^2\big\{M_n^{(d)}(B_2)\big\}^2\Big]\notag\\
&=\frac{n^2}{k^4}\sum_{t_1=d+1}^{n}\sum_{t_2=d+1}^{n}\sum_{t_3=d+1}^{n}\sum_{t_4=d+1}^{n}\E[G_{t_1}G_{t_2}H_{t_3}H_{t_4}]\notag\\
&=\frac{n^2}{k^4}\left\{\sum_{t=d+1}^{n}\E[G_{t}G_{t}H_{t}H_{t}]+ \sum_{t_1\neq t_2}\E[G_{t_1}^2H_{t_2}^2]+2\sum_{t_1\neq t_2}\E[G_{t_1}H_{t_1}G_{t_2}H_{t_2}]\right\}+R_n.\label{eq:(7.2)}
\end{align}
Here, the first term in curly brackets on the right-hand side of \eqref{eq:(7.2)} collects those terms that would arise for serially independent $\{(G_t,H_t)^\prime\}$, whereas the remainder $R_n$ collects those terms that arise due to the $d$-dependence of $\{(G_t,H_t)^\prime\}$. Exploiting the fact that $G_t$ and $H_t$ are centered Bernoulli-distributed random variables with respective success probabilities $p_G$ and $p_H$, the first term on the right-hand side of \eqref{eq:(7.2)} reduces to (with $\widetilde{n}=n-d$)
\begin{align*}
\frac{n^2}{k^4}&\left\{\widetilde{n}\E[G_1^2H_1^2]+\widetilde{n}(\widetilde{n}-1)\E[G_1^2]\E[H_1^2]+2\widetilde{n}(\widetilde{n}-1)\E^2[G_1H_1]\right\}\\
&=\frac{n^2}{k^4}\left\{\widetilde{n}[p_G^2p_{H}+p_G p_H^2+p_G^2 p_H^2]+\widetilde{n}(\widetilde{n}-1)[p_G p_H(1-p_G)(1-p_H)]+2\widetilde{n}(\widetilde{n}-1)p_G^2p_H^2\right\}\\
&\leq K(x_2-x_1)^2(y_2-y)(y-y_1).
\end{align*}
Lengthy but straightforward calculations establish a similar bound for $R_n$. In sum, we obtain
\[
	\E\Big[\big\{M_n^{(d)}(B_1)\big\}^2\big\{M_n^{(d)}(B_2)\big\}^2\Big]\leq K(x_2-x_1)^2(y_2-y)(y-y_1).
\]
This completes the proof.
\end{proof}

\begin{proof}[{\textbf{Proof of Proposition~\ref{prop:Mn unif x xi}:}}]
Consider the supremum over $\vxi\in[-K,0]^2$; the cases where $\vxi$ lies in the other quadrants can be dealt with similarly. By definition of $M_n(\cdot,\cdot)$, it holds that 
\begin{equation}\label{eq:Mn bound}
	\sup_{\vxi\in[-K, 0]^2}\sup_{\vx\in[0,K]^2}|M_n(\vx,\vxi)-M_n(\vx,\vzeros)|\leq\sup_{\vx\in[0,K]^2}|M_n(\vx,-(K,K)^\prime)-M_n(\vx,\vzeros)|.
\end{equation}
Let $\rho>0$. Assume without loss of generality that $K/\rho\in\mathbb{N}$. Put
\[
	\mathcal{L}=\left\{(l_1,l_2)^\prime\in\mathbb{N}_0^2\ :\ (l_1,l_2)^\prime\in[0, K/\rho]^2\right\}.
\]
Then, bound the right-hand side of \eqref{eq:Mn bound} by
\begin{align*}
	&\max_{(l_1,l_2)^\prime\in\mathcal{L}}|M_n((l_1,l_2)^\prime\rho,-(K,K)^\prime)-M_n((l_1,l_2)^\prime\rho,\vzeros)|\\
	&\hspace{2cm}+\sup_{|\vx_1-\vx_2|\leq\rho}\big|\big\{M_n(\vx_1,-(K,K)^\prime)-M_n(\vx_1,\vzeros)\big\}-\big\{M_n(\vx_2,-(K,K)^\prime)-M_n(\vx_2,\vzeros)\}\big|\\
	&\leq\max_{(l_1,l_2)^\prime\in\mathcal{L}}|M_n((l_1,l_2)^\prime\rho,-(K,K)^\prime)-M_n((l_1,l_2)^\prime\rho,\vzeros)|+\sup_{|\vx_1-\vx_2|\leq\rho}|M_n(\vx_1,\vzeros)-M_n(\vx_2,\vzeros)|\\
	&\hspace{2cm}+\sup_{|\vx_1-\vx_2|\leq\rho}|M_n(\vx_1,-(K,K)^\prime)-M_n(\vx_2,-(K,K)^\prime)|\\
	&=:A_{1n}+B_{1n}+C_{1n}.
\end{align*}
Here, $A_{1n}$ is the maximum of the function $\vx\mapsto|M_n(\vx,-(K,K)^\prime)-M_n(\vx,\vzeros)|$ on the grid $\mathcal{L}$, and $B_{1n}+C_{1n}$ bound the function's maximum absolute increase over a grid quadrangle.

By subadditivity and Proposition~\ref{prop:xi convergence}, we obtain that
\begin{align*}
	\p\{\sqrt{n}A_{1n}\geq\varepsilon\} &= \p\Big\{\cup_{(l_1,l_2)^\prime\in\mathcal{L}}\big\{\sqrt{n}\big|M_n((l_1,l_2)^\prime\rho,-(K,K)^\prime)-M_n((l_1,l_2)^\prime\rho,\vzeros)\big|\geq\varepsilon\big\}\Big\}\\
	&\leq\sum_{(l_1,l_2)^\prime\in\mathcal{L}}\p\Big\{\sqrt{n}\big|M_n((l_1,l_2)^\prime\rho,-(K,K)^\prime)-M_n((l_1,l_2)^\prime\rho,\vzeros)\big|\geq\varepsilon\Big\}=o(1).
\end{align*}
Thus, $A_{1n}=o_{\p}(n^{-1/2})$.

Since $B_{1n}=o_{\p}(n^{-1/2})$ can be shown similarly as $C_{1n}=o_{\p}(n^{-1/2})$, we only prove the former. Define
\[
	\widetilde{\mathcal{L}}=\left\{(l_1,l_2)^\prime\in\mathbb{N}_0^2\ :\ [l_1,l_1+2]\times[l_2,l_2+2]\subset[0, K/\rho]^2\right\}.
\]
Use a monotonicity argument to deduce that
\begin{equation*}
B_{1n} \leq \max_{(l_1,l_2)^\prime\in\widetilde{\mathcal{L}}}\big|M_n((l_1+2,l_2+2)^\prime\rho,\vzeros)-M_n((l_1,l_2)^\prime\rho,\vzeros)\big|.
\end{equation*}
Put
\begin{equation*}
	I_{3t}:=I_{\Big\{U_t>b\big(\frac{n}{k[l_1+2]\rho}\big),\ U_{t-d}>b\big(\frac{n}{k[l_2+2]\rho}\big)\Big\}}-I_{\Big\{U_t>b\big(\frac{n}{kl_1\rho}\big),\ U_{t-d}>b\big(\frac{n}{kl_2\rho}\big)\Big\}}.
\end{equation*}
For fixed $(l_1,l_2)^\prime\in\widetilde{\mathcal{L}}$, use Markov's inequality to obtain that
\begin{align}
\p&\left\{\sqrt{n}\big|M_n((l_1+2,l_2+2)^\prime\rho,\vzeros)-M_n((l_1,l_2)^\prime\rho,\vzeros)\big|\geq\varepsilon\right\}\notag\\
&=\p\left\{\frac{\sqrt{n}}{k}\sum_{t=d+1}^{n}I_{3t}\geq\varepsilon\right\}\leq \varepsilon^{-4}\frac{n^2}{k^4}\E\Big[\sum_{t=d+1}^{n}I_{3t}\Big]^4.\label{eq:bound grid deriv}
\end{align}
Due to their $d$-dependence, the $I_{3t}$ are trivially $\rho$-mixing of any rate, and thus we may again apply Lemma~2.3 of \citet{Sha93} (for, in his notation, $q=4$) to deduce that
\begin{equation}\label{eq:Shao}
	\E\Big[\sum_{t=d+1}^{n}I_{3t}\Big]^4\leq n^2 C \big\{\E[I_{3t}^2]\big\}^2+n C \E[I_{3t}^4].
\end{equation}
Since $I_{3t}\in\{0,1\}$ and, hence, $\E[I_{3t}^4]=\E[I_{3t}^2]=\E[I_{3t}]$, it suffices to compute $\E[I_{3t}]$. Use Lemma~\ref{lem:Lem F} to obtain
\begin{align*}
\E[I_{3t}]	&=\p\Big\{U_t>b\Big(\frac{n}{k[l_1+2]\rho}\Big)\Big\} \p\Big\{U_{t-d}>b\Big(\frac{n}{k[l_2+2]\rho}\Big)\Big\}\\
	&\hspace{3cm}-\p\Big\{U_t>b\Big(\frac{n}{kl_1\rho}\Big)\Big\} \p\Big\{U_{t-d}>b\Big(\frac{n}{kl_2\rho}\Big)\Big\}\\
	&=\frac{k^2}{n^2}\big[(l_1+2)(l_2+2)-l_1l_2\big]\rho^2\\
	&\leq C\frac{k^2}{n^2}\rho^2.
\end{align*}
Together with \eqref{eq:Shao}, this implies 
\[
	\E\Big[\sum_{t=d+1}^{n}I_{3t}\Big]^4\leq C \frac{k^4}{n^2}\rho^4+C\frac{k^2}{n}\rho^2.
\]
With \eqref{eq:bound grid deriv}, this gives
\[
	\p\left\{\sqrt{n}\big|M_n((l_1+2,l_2+2)^\prime\rho,\vzeros)-M_n((l_1,l_2)^\prime\rho,\vzeros)\big|\geq\varepsilon\right\}\leq C\rho^4+C\frac{n}{k^2}\rho^2.
\]
Thus, by subadditivity and using that the cardinality of $\widetilde{\mathcal{L}}$ is of the order $\rho^{-2}$, we obtain with Assumption~\ref{ass:k} that
\begin{align*}
	\p\{\sqrt{n}B_{1n}\geq\varepsilon\} &= \p\Big\{\cup_{(l_1,l_2)^\prime\in\mathcal{L}}\big\{\sqrt{n}\big|M_n((l_1+2,l_2+2)^\prime\rho,\vzeros)-M_n((l_1,l_2)^\prime\rho,\vzeros)\big|\geq\varepsilon\big\}\Big\}\\
	&\leq\sum_{(l_1,l_2)^\prime\in\mathcal{L}}\p\Big\{\sqrt{n}\big|M_n((l_1+2,l_2+2)^\prime\rho,\vzeros)-M_n((l_1,l_2)^\prime\rho,\vzeros)\big|\geq\varepsilon\Big\}\\
	&\leq O(\rho^{-2})\Big[C\rho^4+C\frac{n}{k^2}\rho^2\Big]\leq C\rho^2+o(1).
\end{align*}
Since $\rho$ can be chosen arbitrarily small, it follows that $B_{1n}=o_{\p}(n^{-1/2})$. This ends the proof.
\end{proof}

For the following lemma, recall the definition of $B_t(\vx,\vxi,\eta,\eta_0)$ in \eqref{eq:(A.10+)}. The proof of the lemma is in Appendix~\ref{Auxiliary Lemmas for Proof of Theorem 2}.

\begin{lem}\label{lem:unif Bop1}
Let $\eta>0$, $\eta_0\in\{-1, 1\}$. Then, for any $K>0$,
\begin{equation*}
\sup_{\vxi\in[-K,K]^2}\sup_{\vx\in[0,K]^2}\left|\frac{1}{k}\sum_{t=d+1}^{n}B_t(\vx,\vxi,\eta,\eta_0)\right|=o_{\p}(n^{-1/2}).
\end{equation*}
\end{lem}

\begin{proof}[{\textbf{Proof of Proposition~\ref{prop:hat unif convergence}:}}]
The proof is identical to that of Proposition~\ref{prop:hat convergence} using Lemma~\ref{lem:unif Bop1} instead of Lemma~\ref{lem:Bop1}.
\end{proof}

The next lemma, whose proof is in Appendix~\ref{Auxiliary Lemmas for Proof of Theorem 2}, is required for the proof of Proposition~\ref{prop:unif emp an}.

\begin{lem}\label{lem:univ unif}
For any $0<\iota<K<\infty$, it holds that, as $n\to\infty$,
\[
	\sup_{x\in[\iota,K]}\Bigg|\frac{1}{\sqrt{k}}\sum_{t=1}^{n}\Big[I_{\big\{\widehat{U}_t>b(n/[kx])\big\}}-I_{\big\{U_t>b(n/[kx])\big\}}\Big]\Bigg|=o_{\p}(1).
\]
\end{lem}

\begin{proof}[{\textbf{Proof of Proposition~\ref{prop:unif emp an}:}}]
Observe that for all $\xi\in\mathbb{R}$
\begin{align}
&\sqrt{k}\log\Bigg(\frac{\widehat{U}_{(\lfloor kx\rfloor+1)}}{b(n/[kx])}\Bigg)\leq\xi\quad\Longleftrightarrow \quad \widehat{U}_{(\lfloor kx\rfloor+1)}\leq e^{\xi/\sqrt{k}}b(n/[kx])\notag\\
&\Longleftrightarrow\quad \sum_{t=1}^{n}I_{\big\{\widehat{U}_t>e^{\xi/\sqrt{k}}b(n/[kx])\big\}}\leq kx \notag\\
&\Longleftrightarrow\quad \frac{1}{\sqrt{k}}\sum_{t=1}^{n}\Big[I_{\big\{\widehat{U}_t>e^{\xi/\sqrt{k}}b(n/[kx])\big\}}-\p\big\{U_t>e^{\xi/\sqrt{k}}b(n/[kx])\big\}\Big]\leq\alpha\xi x+o(1),\label{eq:fin equiv}
\end{align}
where the $o(1)$-term is uniform in $x\in[\iota,K]$ by Lemma~\ref{lem:Lem F}.

Denote by $D[0,K]$ the space of real-valued functions on $[0,K]$ that are right-continuous with limits from below \citep{Dav94}. Then, it can be proven along similar lines as in Proposition~\ref{prop:base functional convergence} that 
\[
	\frac{1}{\sqrt{k}}\sum_{t=1}^{n}\Big[I_{\big\{U_t>b(n/[kx])\big\}}-\p\Big\{U_t>b(n/[kx])\Big\}\Big]\overset{d}{\longrightarrow}B(x)\qquad \text{in }D[0,K],
\]
where $B(\cdot)$ denotes a standard Brownian motion; see also the proof of Proposition~2.1 in \citet{RS97a}. This implies by the continuous mapping theorem that
\[
	\sup_{x\in[\iota,K]}\Bigg|\frac{1}{\sqrt{k}}\sum_{t=1}^{n}\Big[I_{\big\{U_t>b(n/[kx])\big\}}-\p\Big\{U_t>b(n/[kx])\Big\}\Big]\Bigg|=O_{\p}(1).
\]
From this, Lemma~\ref{lem:univ unif} and Assumption~\ref{ass:U distr 2} we deduce that
\[
	\sup_{x\in[\iota,K]}\Bigg
	|\frac{1}{\sqrt{k}}\sum_{t=1}^{n}\Big[I_{\big\{\widehat{U}_t>e^{\xi/\sqrt{k}}b(n/[kx])\big\}}-\p\Big\{U_t>e^{\xi/\sqrt{k}}b(n/[kx])\Big\}\Big]\Bigg|=O_{\p}(1).
\]
Combining this with \eqref{eq:fin equiv}, the conclusion follows.
\end{proof}

\subsection{Proofs of Lemmas~\ref{lem:unif Bop1}--\ref{lem:univ unif}}\label{Auxiliary Lemmas for Proof of Theorem 2}

\begin{proof}[{\textbf{Proof of Lemma~\ref{lem:unif Bop1}:}}]
It suffices to consider the supremum over $\vxi\in[-K,0]^2$; the cases where $\vxi$ lies in the other quadrants can be dealt with similarly. Using the same grid $\mathcal{L}$ as in the proof of Proposition~\ref{prop:Mn unif x xi}, we get that
\begin{align*}
\sup_{\vxi\in[-K,0]^2}&\sup_{\vx\in[0,K]^2}\left|\frac{1}{k}\sum_{t=d+1}^{n}B_t(\vx,\vxi,\eta,\eta_0)\right|\leq\sup_{\vx\in[0,K]^2}\left|\frac{1}{k}\sum_{t=d+1}^{n}B_t(\vx,-(K,K)^\prime,\eta,\eta_0)\right|\\
 &\leq\max_{(l_1,l_2)^\prime\in\mathcal{L}}\left|\frac{1}{k}\sum_{t=d+1}^{n}B_t((l_1, l_2)^\prime\rho,-(K,K)^\prime,\eta,\eta_0)\right|\\
&\hspace{2cm}+\sup_{|\vx_1-\vx_2|\leq\rho}\left|\frac{1}{k}\sum_{t=d+1}^{n}B_t(\vx_1,-(K,K)^\prime,\eta,\eta_0)-B_t(\vx_2,-(K,K)^\prime,\eta,\eta_0)\right|\\
&=:A_{2n}+B_{2n}.
\end{align*}
By subadditivity and Lemma~\ref{lem:Bop1}, we obtain that
\begin{align*}
	\p\{\sqrt{n}A_{2n}\geq\varepsilon\} &= \p\Bigg\{\bigcup_{(l_1,l_2)^\prime\in\mathcal{L}}\Big\{\sqrt{n}\Big|\frac{1}{k}\sum_{t=d+1}^{n}B_t\big((l_1,l_2)^\prime\rho,-(K,K)^\prime,\eta,\eta_0\big)\Big|\geq\varepsilon\Big\}\Bigg\}\\
	&\leq\sum_{(l_1,l_2)^\prime\in\mathcal{L}}\p\Bigg\{\sqrt{n}\Big|\frac{1}{k}\sum_{t=d+1}^{n}B_t\big((l_1,l_2)^\prime\rho,-(K,K)^\prime,\eta,\eta_0\big)\Big|\geq\varepsilon\Bigg\}=o(1).
\end{align*}
Thus, $A_{2n}=o_{\p}(n^{-1/2})$.

To show that $B_{2n}=o_{\p}(n^{-1/2})$, use a monotonicity argument to deduce that
\begin{equation*}
B_{2n} \leq \max_{(l_1,l_2)^\prime\in\widetilde{\mathcal{L}}}\left|\frac{1}{k}\sum_{t=d+1}^{n}B_t\big((l_1+2,l_2+2)^\prime\rho,-(K,K)^\prime,\eta,\eta_0\big)-B_t\big((l_1,l_2)^\prime\rho,-(K,K)^\prime,\eta,\eta_0\big)\right|,
\end{equation*}
where $\widetilde{\mathcal{L}}$ is again defined as in the proof of Proposition~\ref{prop:Mn unif x xi}. Let $\eta_0=1$; the case $\eta_0=-1$ can be treated analogously. Define $w_t$ and $\nu=\nu(\varepsilon_0)>0$ as in the proof of Lemma~\ref{lem:Bop1}. In particular, if $w_t=1$ then
\begin{equation*}
	1-\nu/2<(1+\s_{t})^{-1}\leq1\qquad\text{and}\qquad 1-\nu/2<(1+\s_{t-d})^{-1}\leq1.
\end{equation*}
For sufficiently large $n$, such that $(1-\nu/2)<\big[1-e^{K/\sqrt{k}}\varepsilon_0/b(n/[k K])\big]\leq1$,
\begin{align*}
	&w_t\big|B_t\big((l_1+2,l_2+2)^\prime\rho,-(K,K)^\prime,\eta,\eta_0=1\big)-B_t\big((l_1,l_2)^\prime\rho,-(K,K)^\prime,\eta,\eta_0=1\big)\big| \\
	&=w_t\Big|I_{\Big\{U_t> e^{-K/\sqrt{k}}(1+\s_{t})^{-1}\big(1-e^{K/\sqrt{k}}\frac{m_{t}}{b(n/[k(l_1+2)\rho])}\big)b\big(\frac{n}{k[l_1+2]\rho}\big),}\\ 
	& \hspace{1.5cm} _{U_{t-d}> e^{-K/\sqrt{k}}(1+\s_{t-d})^{-1}\big(1-e^{K/\sqrt{k}}\frac{m_{t-d}}{b(n/[k(l_2+2)\rho])}\big)b\big(\frac{n}{k[l_2+2]\rho}\big)\Big\}}\\
	&\hspace{8.4cm}-I_{\Big\{U_t> e^{-K/\sqrt{k}}b\big(\frac{n}{k[l_1+2]\rho}\big),\ U_{t-d}> e^{-K/\sqrt{k}}b\big(\frac{n}{k[l_2+2]\rho}\big)\Big\}}\\
	&\hspace{0.5cm}-I_{\Big\{U_t> e^{-K/\sqrt{k}}(1+\s_{t})^{-1}\big(1-e^{K/\sqrt{k}}\frac{m_{t}}{b(n/[kl_1\rho])}\big)b\big(\frac{n}{kl_1\rho}\big),}\\
	&\hspace{1.5cm} _{U_{t-d}> e^{-K/\sqrt{k}}(1+\s_{t-d})^{-1}\big(1-e^{K/\sqrt{k}}\frac{m_{t-d}}{b(n/[kl_2\rho])}\big)b\big(\frac{n}{kl_2\rho}\big)\Big\}}+I_{\Big\{U_t> e^{-K/\sqrt{k}}b\big(\frac{n}{kl_1\rho}\big),\ U_{t-d}> e^{-K/\sqrt{k}}b\big(\frac{n}{kl_2\rho}\big)\Big\}}\Big|\\
	&\leq \Big|I_{\Big\{U_t> (1-\nu)b\big(\frac{n}{k[l_1+2]\rho}\big),\ U_{t-d}> (1-\nu)b\big(\frac{n}{k[l_2+2]\rho}\big)\Big\}}-I_{\Big\{U_t> b\big(\frac{n}{k[l_1+2]\rho}\big),\ U_{t-d}> b\big(\frac{n}{k[l_2+2]\rho}\big)\Big\}}\Big|\\
	&\hspace{1cm}+\Big|I_{\Big\{U_t> (1-\nu)b\big(\frac{n}{kl_1\rho}\big),\ U_{t-d}> (1-\nu)b\big(\frac{n}{kl_2\rho}\big)\Big\}}-I_{\Big\{U_t> b\big(\frac{n}{kl_1\rho}\big),\ U_{t-d}> b\big(\frac{n}{kl_2\rho}\big)\Big\}}\Big|\\
	&=:I_{4t}(l_1,l_2)+I_{5t}(l_1,l_2).
\end{align*}
Since $w_{\ell_{n}}=\ldots=w_n=1$ w.p.a.~1 as $n\to\infty$, to prove $B_{2n}=o_{\p}(n^{-1/2})$, it suffices to show that
\begin{equation}\label{eq:(C.21)}
	\max_{(l_1,l_2)^\prime\in\widetilde{\mathcal{L}}}\frac{1}{k}\sum_{t=d+1}^{n}I_{jt}(l_1,l_2)=o_{\p}(n^{-1/2})\qquad(j=4,5).
\end{equation}
We only show this claim for $j=5$, as that for $j=4$ can be proven along similar lines. By subadditivity and Markov's inequality,
\begin{align*}
\p\Bigg\{\sqrt{n}\max_{(l_1,l_2)^\prime\in\widetilde{\mathcal{L}}}\frac{1}{k}\sum_{t=d+1}^{n}I_{5t}(l_1,l_2)\geq\varepsilon\Bigg\}&=\p\Bigg\{\bigcup_{(l_1,l_2)^\prime\in\widetilde{\mathcal{L}}}\Big\{\frac{\sqrt{n}}{k}\sum_{t=d+1}^{n}I_{5t}(l_1,l_2)\geq\varepsilon\Big\}\Bigg\}\\
&=\sum_{(l_1,l_2)^\prime\in\widetilde{\mathcal{L}}}\p\Big\{\frac{\sqrt{n}}{k}\sum_{t=d+1}^{n}I_{5t}(l_1,l_2)\geq\varepsilon\Big\}\\
&\leq\sum_{(l_1,l_2)^\prime\in\widetilde{\mathcal{L}}}\varepsilon^{-4}\frac{n^2}{k^4}\E\Big[\sum_{t=d+1}^{n}I_{5t}(l_1,l_2)\Big]^4.
\end{align*}
Use Lemma~\ref{lem:Lem F} to conclude that for sufficiently small $\nu>0$,
\begin{align*}
	\E[I_{5t}(l_1,l_2)]&=\p\Big\{U_t> (1-\nu)b\big(\frac{n}{kl_1\rho}\big)\Big\}\p\Big\{U_{t-d}> (1-\nu)b\big(\frac{n}{kl_2\rho}\big)\Big\}\\
	&\hspace{5cm}-\p\Big\{U_t> b\big(\frac{n}{kl_1\rho}\big)\Big\}\p\Big\{U_{t-d}> b\big(\frac{n}{kl_2\rho}\big)\Big\}\\
	&=\frac{k^2}{n^2}\rho^2 l_1 l_2\Big[1+\frac{\alpha\nu}{(1+\widetilde{\nu})^{1+\alpha}}+\nu o(1)\Big]^2-\frac{k^2}{n^2}\rho^2 l_1 l_2\\
	&\leq C\frac{k^2}{n^2}\rho^2 l_1 l_2\Big[\frac{\alpha\nu}{(1+\widetilde{\nu})^{1+\alpha}}+\nu o(1)\Big]\\
	&\leq C\frac{k^2}{n^2}\rho^2\nu.
\end{align*}
Arguing similarly as after \eqref{eq:Shao}, we obtain
\[
	\E\Big[\sum_{t=d+1}^{n}I_{5t}(l_1,l_2)\Big]^4\leq C \frac{k^4}{n^2}\rho^4\nu+C\frac{k^2}{n}\rho^2\nu.
\]
Using Assumption~\ref{ass:k} and that the cardinality of $\widetilde{\mathcal{L}}$ is of the order $\rho^{-2}$, we get from this that
\[
\p\Bigg\{\sqrt{n}\max_{(l_1,l_2)^\prime\in\widetilde{\mathcal{L}}}\frac{1}{k}\sum_{t=d+1}^{n}I_{5t}(l_1,l_2)\geq\varepsilon\Bigg\} \leq C\rho^2\nu + C\frac{n}{k^2}\nu\leq C\rho^2\nu,
\]
which can be made arbitrarily small by a suitable choice of $\nu$. This establishes \eqref{eq:(C.21)}, concluding the proof.
\end{proof}

The following lemmas are required for the proof of Lemma~\ref{lem:univ unif}. Set
\begin{align*}
C_t(x,\vtheta)     &= I_{\big\{\widehat{U}_t(\vtheta)>b(n/[kx])\big\}},\\
C_t(x)     				 &= I_{\big\{U_t>b(n/[kx])\big\}},\\
C_t(x,\eta,\eta_0) &= I_{\left\{U_t[1+\eta_0\s_{t}]+\eta_0m_t > b(n/[kx])\right\}},\qquad\eta_0\in\{-1, 1\},\\
D_t(x,\eta,\eta_0) &= C_t(x,\eta,\eta_0) - C_t(x),
\end{align*}
where $m_t=m_{n,t}(\eta)\geq0$ and $s_t=s_{n,t}(\eta)\geq0$ are from Assumption~\ref{ass:UA}.

\begin{lem}\label{lem:Ct bound}
Let $\eta>0$. Then, w.p.a.~1, as $n\to\infty$,
\begin{equation*}
	C_t(x,\eta,-1) \leq C_t(x,\vtheta) \leq C_t(x, \eta,1)
\end{equation*}
for all $\vtheta\in N_n(\eta)$ and $t=\ell_{n},\ldots,n$. 
\end{lem}

\begin{proof}
The proof closely resembles that of Lemma~\ref{lem:wpa1} and, hence, is omitted.
\end{proof}

\begin{lem}\label{lem:Dt unif}
Let $\eta>0$ and $\eta_0\in\{-1, 1\}$. Then, for any $0<\iota<K<\infty$,
\[
	\sup_{x\in[\iota,K]}\Bigg|\frac{1}{\sqrt{k}}\sum_{t=1}^{n}D_t(x,\eta,\eta_0)\Bigg|=o_{\p}(1).
\]
\end{lem}

\begin{proof}
The outline of the proof is similar to that of Lemma~\ref{lem:Bop1}. Let $\ell_{n}\to\infty$ with $\ell_{n}=o(\sqrt{k})$. Consider the case $\eta_0=1$; the case $\eta_0=-1$ can be dealt with similarly. Define
\[
	v_t=I_{\left\{m_{t}\leq\varepsilon_0,\ s_{t}\leq\varepsilon_0\right\}},\qquad\varepsilon_0>0.
\]
If $v_t=1$, there exists $\nu=\nu(\varepsilon_0)>0$, such that for $\s_{t}$ from Assumption~\ref{ass:UA} it holds that
\begin{equation*}
	1-\nu/2<(1+s_{t})^{-1}\leq1.
\end{equation*}
Then, for sufficiently large $n$, such that $1-\nu/2<1-\varepsilon_0/b(n/[kx])\leq1$ for all $x\in[\iota,K]$, we can bound
\begin{align*}
\Big|\frac{1}{\sqrt{k}}\sum_{t=\ell_{n}}^{n}v_tD_t(x,\eta,\eta_0=1)\Big|&=\Big|\frac{1}{\sqrt{k}}\sum_{t=\ell_{n}}^{n}v_t\big[I_{\left\{U_t>(1+\s_{t})^{-1}(1-m_t/b(n/[kx])) b(n/[kx])\right\}}-I_{\left\{U_t>b(n/[kx])\right\}}\big]\\
&\leq\frac{1}{\sqrt{k}}\sum_{t=\ell_{n}}^{n}\big[I_{\left\{U_t>(1-\nu) b(n/[kx])\right\}}-I_{\left\{U_t>b(n/[kx])\right\}}\big]\\
&=\frac{1}{\sqrt{k}}\sum_{t=\ell_{n}}^{n}I_{\left\{U_t\in\big( (1-\nu)b(n/[kx]),\ b(n/[kx])\big]\right\}}\\
&=: \frac{1}{\sqrt{k}}\sum_{t=\ell_{n}}^{n}I_{6t}(x).
\end{align*}
For $\rho>0$, define $\mathcal{M}=\{l\in\mathbb{N}_0\ :\ l\in[0, K/\rho]\}$. Then,
\begin{align*}
\frac{1}{\sqrt{k}}\sum_{t=\ell_{n}}^{n}I_{6t}(x) &= \max_{l\in\mathcal{M}}\frac{1}{\sqrt{k}}\sum_{t=\ell_{n}}^{n}I_{6t}(l\rho)+\sup_{|x_1-x_2|\leq\rho}\Big|\frac{1}{\sqrt{k}}\sum_{t=\ell_{n}}^{n}[I_{6t}(x_1)-I_{6t}(x_2)]\Big|\\
&=:A_{3n} + B_{3n}.
\end{align*}
Using in turn subadditivity, Markov's inequality, and Lemma~2.3 of \citet{Sha93}, we get that
\begin{align}
	\p\{A_{3n}>\varepsilon\} &\leq \sum_{l\in\mathcal{M}}\p\Big\{\frac{1}{\sqrt{k}}\sum_{t=\ell_{n}}^{n}I_{6t}(l\rho)>\varepsilon\Big\}\notag\\
	&\leq \sum_{l\in\mathcal{M}}\frac{1}{k^2}\E\Bigg[\sum_{t=\ell_{n}}^{n}I_{6t}(l\rho)\Bigg]^4\notag\\
	&\leq \sum_{l\in\mathcal{M}}\frac{C}{k^2}\Big\{n^2\big\{\E[I_{6t}^{2}(l\rho)]\big\}^2 + n\E[I_{6t}^4(l\rho)]\Big\}.\label{eq:(10.1)}
\end{align}
Since $I_{6t}(l\rho)\in\{0,1\}$, Lemma~\ref{lem:Lem F} implies
\begin{align}
 \E[I_{6t}^{4}(l\rho)] &= \E[I_{6t}^{2}(l\rho)] = \E[I_{6t}(l\rho)]\notag\\
&= \p\Big\{U_t\in\big((1-\nu)b(n/[kl\rho]),\ b(n/[kl\rho])\big]\Big\}\notag\\
&= \p\Big\{U_t>(1-\nu)b(n/[kl\rho])\Big\} - \p\Big\{U_t>b(n/[kl\rho])\Big\}\notag\\
&= \frac{kl\rho}{n}\Big\{1+\frac{\alpha\nu}{(1+\widetilde{\nu})^{1+\alpha}}+\nu o(1)-1\Big\}\notag\\
&= \frac{kl\rho}{n}\frac{\alpha\nu}{(1+\widetilde{\nu})^{1+\alpha}}+\nu l\rho o\Big(\frac{k}{n}\Big)\label{eq:(10.2)}\\
&\leq C\nu\rho\frac{k}{n}.\notag
\end{align}
Plugging this into \eqref{eq:(10.1)} gives
\begin{align*}
\p\{A_{3n}>\varepsilon\} &\leq \sum_{l\in\mathcal{M}}C\left\{\frac{n^2}{k^2}[\nu\rho]^2\frac{k^2}{n^2}+\frac{n}{k^2}\nu\rho\frac{k}{n}\right\}\\
&\leq C \left\{\nu^2\rho^2 + \frac{\nu}{k}\rho\right\},
\end{align*}
which can be made arbitrarily small by a suitable choice of $\nu$. Thus, $A_{3n}=o_{\p}(1)$.

To show $B_{3n}=o_{\p}(1)$, define
\[
	\widetilde{\mathcal{M}} = \{l\in\mathbb{N}_0\ :\ [l, l+2]\subset[0,K/\rho]\}.
\]
Use a monotonicity argument to deduce that 
\[
	B_{3n}\leq\max_{l\in\widetilde{\mathcal{M}}}\Big|\frac{1}{\sqrt{k}}\sum_{t=\ell_{n}}^{n}\big[I_{6t}([l+2]\rho)-I_{6t}(l\rho)\big]\Big|=:\max_{l\in\widetilde{\mathcal{M}}}\Big|\frac{1}{\sqrt{k}}\sum_{t=1}^{n}I_{7t}(l,\rho)\Big|.
\]
We get by Markov's inequality that
\begin{equation}\label{eq:(11.1)}
	\p\Bigg\{\Big|\frac{1}{\sqrt{k}}\sum_{t=\ell_{n}}^{n}I_{7t}(l,\rho)\Big|>\varepsilon\Bigg\}\leq\frac{\varepsilon^{-4}}{k^2}\E\Big[\sum_{t=\ell_{n}}^{n}I_{7t}(l,\rho)\Big]^4.
\end{equation}
Since, by \eqref{eq:(10.2)},
\begin{align*}
\E\Big[I_{7t}(l,\rho)\Big] &= 2\rho\frac{k}{n}\frac{\alpha\nu}{(1+\widetilde{\nu})^{1+\alpha}}+2\rho\nu o(k/n)\leq C\rho\nu\frac{k}{n},
\end{align*}
we obtain from Lemma~2.3 of \citet{Sha93} that
\begin{align*}
 \E\Big[\sum_{t=\ell_{n}}^{n}I_{7t}(l,\rho)\Big]^4 &\leq C\left\{n^2  \big\{\E[I_{7t}(l,\rho)]^2\big\}^2 + n  \E[I_{7t}(l,\rho)]^4 \right\} \\
&\leq C \Big\{ n^2\Big[\rho\nu\frac{k}{n}\Big]^2 + n\Big[C\rho\nu\frac{k}{n}\Big]\Big\}. 
\end{align*}
Thus, the right-hand side of \eqref{eq:(11.1)} can be bounded by $C\rho^2\nu^2+\frac{C}{k}\rho\nu$. Using this, we obtain
\begin{align*}
\p\{B_{3n}>\varepsilon\} &\leq \sum_{l\in\widetilde{\mathcal{M}}}\p\Big\{\Big|\frac{1}{\sqrt{k}}\sum_{t=\ell_n}^{n}I_{7t}(l,\rho)\Big|>\varepsilon\Big\}\\
&= \sum_{l\in\widetilde{\mathcal{M}}} \Big[C\rho^2\nu^2+\frac{C}{k}\rho\nu\Big]\\
&\leq C\{\rho\nu^2+\nu/k\},
\end{align*}
where we have used in the final step that the cardinality of $\widetilde{\mathcal{M}}$ is of the order $1/\rho$. This proves 
\[
	\Big|\frac{1}{\sqrt{k}}\sum_{t=\ell_{n}}^{n}v_tD_t(x,\eta,\eta_0)\Big| = o_{\p}(1),
\]
since $\nu>0$ can be chosen arbitrarily small. By Assumption~\ref{ass:UA}, we have $v_{\ell_{n}}=\ldots=v_n=1$ w.p.a.~1, as $n\to\infty$. Hence, the above display also implies that 
\[
	\Big|\frac{1}{\sqrt{k}}\sum_{t=\ell_{n}}^{n}D_t(x,\eta,\eta_0)\Big| = o_{\p}(1).
\]
Moreover, by boundedness of the $D_t(x,\eta,\eta_0)$, we easily get $\Big|\frac{1}{\sqrt{k}}\sum_{t=1}^{\ell_{n}-1}D_t(x,\eta,\eta_0)\Big|=O(\ell_{n}/\sqrt{k})=o(1)$. Putting these two results together, the conclusion follows.
\end{proof}

Now, we are in a position to prove Lemma~\ref{lem:univ unif}.

\begin{proof}[{\textbf{Proof of Lemma~\ref{lem:univ unif}:}}]
The proof is again similar to that of Proposition~\ref{prop:hat convergence}, where now we use Lemma~\ref{lem:Ct bound} (instead of Lemma~\ref{lem:wpa1}) and Lemma~\ref{lem:Dt unif} (instead of Lemma~\ref{lem:Bop1}).
\end{proof}

\section{Simulation Results for Varying $D$}\label{Simulation Results for Varying $D$}

In light of the simulation results in Section~\ref{Simulations}, we pick $k=\lfloor 0.11\cdot n^{0.99}\rfloor$. We do so to focus here on the sensitivity of our tests with respect to the number of included lags $D\in\{1,\ldots,10\}$. We reconsider the models of Section~\ref{Simulations} with misspecified volatility dynamics (in Section~\ref{App: Misspecified Volatility}) and with misspecified innovations (in Section~\ref{App: Misspecified Innovations}).

\subsection{Misspecified Volatility}\label{App: Misspecified Volatility}

\begin{table}[t!]
	\begin{center}
		\begin{tabular}{lllrrrrrrrrrr}
			\toprule
$n$	& Test    &	$\alpha$&\multicolumn{10}{c}{$D$}  \\[0.25ex]
\cline{4-13}\\[-2.25ex] 
		&         & 				& 1 & 2 & 3 & 4 & 5 & 6 & 7 & 8 & 9 & 10 \\
\midrule	
500 & $\mathcal{P}_n^{(D)}$	& 1\%		  &	0.3   &  0.5   &  0.6   &   0.5   &   0.5   &    0.6   &  0.6   &   0.5   &   0.4   &   0.4	\\
		&          							& 5\%	    &	2.6   &  2.5   &  2.5   &   2.5   &   2.2   &    2.3   &  2.1   &   1.9   &   1.8   &   1.7  \\
		&          							& 10\%    &	8.1   &  5.4   &  5.3   &   4.8   &   5.0   &    4.9   &  4.5   &   4.2   &   3.8   &   3.7  \\[0.25ex]
			\cline{2-13}\\[-2.25ex]                                                                                                      
		& $\mathcal{F}_n^{(D)}$	& 1\%		  &	0.5   &  0.4   &  0.4   &   0.3   &   0.3   &    0.4   &  0.3   &   0.2   &   0.2   &   0.2	\\
		&          							& 5\%	    &	2.0   &  1.7   &  1.8   &   1.8   &   1.6   &    1.6   &  1.4   &   1.2   &   1.0   &   0.8  \\
		&          							& 10\%    &	4.6   &  3.8   &  3.4   &   3.4   &   3.1   &    3.2   &  2.9   &   2.5   &   2.1   &   1.9  \\[0.25ex]
			\cline{2-13}\\[-2.25ex]                                                                                                      
		& $\LB_n^{(D)}$					& 1\%		  &	1.2   &  1.7   &  2.3   &   2.8   &   3.0   &    3.2   &  3.4   &   3.5   &   3.4   &   3.5	\\
		&          							& 5\%	    &	3.2   &  3.9   &  4.8   &   5.3   &   5.8   &    6.1   &  6.5   &   6.6   &   6.9   &   6.9  \\
		&          							& 10\%    &	5.7   &  6.1   &  7.3   &   7.4   &   8.2   &    8.5   &  9.0   &   9.4   &   9.3   &   9.8  \\
		\midrule                                                                                                                       
1000& $\mathcal{P}_n^{(D)}$	& 1\%		  &	0.8   &  1.0   &  0.9   &   0.9   &   0.7   &    0.5   &  0.5   &   0.6   &   0.5   &   0.5  \\
		&          							& 5\%	    &	4.8   &  3.4   &  3.5   &   3.3   &   3.1   &    2.8   &  2.7   &   2.5   &   2.2   &   2.0  \\
		&          							& 10\%    &	9.7   &  7.0   &  6.5   &   6.1   &   5.9   &    5.7   &  5.4   &   4.9   &   4.4   &   4.0  \\[0.25ex]
			\cline{2-13}\\[-2.25ex]                                                                                                      
		& $\mathcal{F}_n^{(D)}$	& 1\%		  &	0.7   &  0.7   &  0.7   &   0.6   &   0.5   &    0.5   &  0.6   &   0.5   &   0.4   &   0.4  \\
		&          							& 5\%	    &	3.5   &  2.9   &  2.9   &   2.2   &   2.0   &    2.0   &  1.8   &   1.4   &   1.4   &   1.3  \\
		&          							& 10\%    &	6.9   &  5.4   &  5.2   &   4.4   &   4.2   &    3.9   &  3.7   &   3.4   &   2.7   &   2.4  \\[0.25ex]
			\cline{2-13}\\[-2.25ex]                                                                                                        
		& $\LB_n^{(D)}$					& 1\%		  &	1.3   &  2.0   &  2.5   &   2.5   &   2.9   &    3.0   &  3.3   &   3.7   &   3.7   &   3.7  \\
		&          							& 5\%	    &	2.7   &  4.0   &  5.0   &   5.3   &   5.7   &    5.7   &  6.2   &   6.5   &   6.8   &   6.5  \\
		&          							& 10\%    &	4.8   &  6.4   &  7.2   &   7.6   &   8.3   &    8.5   &  8.6   &   8.7   &   9.1   &   9.3  \\
		\midrule                                                                                                                         
2000& $\mathcal{P}_n^{(D)}$	& 1\%		  &	0.9   &  1.0   &  1.1   &   1.1   &   1.0   &    0.9   &  1.0   &   0.9   &   0.8   &   0.7  \\
		&          							& 5\%	    &	4.2   &  4.0   &  3.8   &   3.6   &   3.4   &    3.0   &  2.9   &   2.8   &   2.5   &   2.6   \\
		&          							& 10\%    &	7.9   &  7.7   &  7.3   &   6.9   &   6.5   &    6.2   &  5.8   &   5.4   &   5.1   &   5.2   \\[0.25ex]
			\cline{2-13}\\[-2.25ex]                                                                                                        
		& $\mathcal{F}_n^{(D)}$	& 1\%		  &	1.0   &  1.3   &  1.3   &   1.3   &   1.1   &    1.2   &  1.0   &   0.9   &   0.8   &   0.8  \\
		&          							& 5\%	    &	3.9   &  3.5   &  3.6   &   3.5   &   3.2   &    3.0   &  2.7   &   2.3   &   2.2   &   2.1   \\
		&          							& 10\%    &	7.4   &  6.8   &  6.2   &   6.3   &   5.8   &    5.1   &  4.6   &   4.5   &   4.1   &   3.7   \\[0.25ex]
			\cline{2-13}\\[-2.25ex]                                                                                                        
		& $\LB_n^{(D)}$					& 1\%		  &	1.6   &  2.2   &  3.0   &   3.3   &   3.5   &    3.8   &  3.9   &   4.1   &   4.3   &   4.2  \\
		&          							& 5\%	    &	3.2   &  4.3   &  5.2   &   5.8   &   6.0   &    6.4   &  6.7   &   7.0   &   7.4   &   7.5   \\
		&          							& 10\%    &	4.9   &  6.0   &  7.3   &   7.8   &   8.4   &    8.6   &  9.0   &   9.3   &   9.7   &   9.9   \\
	\bottomrule
		\end{tabular}
	\end{center}
\caption{\label{tab:size}Size in \% of tests based on $\mathcal{P}_n^{(D)}$, $\mathcal{F}_n^{(D)}$ and $\LB_n^{(D)}$ for significance levels $\alpha\in\{1\%, 5\%, 10\%\}$ and $D=1,\ldots,10$. Results are for model \eqref{eq:APARCH-X(1,1)} with $\vtheta^{\circ}=\vtheta^{\circ}_{s}=(0.046, 0.027, 0.092, 0.843, 0)^\prime$.}
\end{table}

Reconsider the setup of Section~\ref{Misspecified Volatility} in the main paper. The only difference is that we now consider a fixed $k=\lfloor 0.11\cdot n^{0.99}\rfloor$ and varying $D\in\{1,\ldots,10\}$. When $\vtheta^{\circ}=\vtheta^{\circ}_{s}$ in \eqref{eq:APARCH-X(1,1)}, the results correspond to size, which is displayed in Table~\ref{tab:size}. As a benchmark, we have again included the $\LB_{n}^{(D)}$-test with the estimation correction of \citet{CF11}. We see that, particularly for small sample sizes, our tests tend to be undersized. Yet, the size distortions substantially decrease with increasing $n$. Except for $D=1$, size is reasonably stable across different choices of the number $D$ of included lags.

\begin{table}[t!]
	\begin{center}
		\begin{tabular}{lllrrrrrrrrrr}
			\toprule
$n$	& Test    &	$\alpha$&\multicolumn{10}{c}{$D$}  \\[0.25ex]
\cline{4-13}\\[-2.25ex] 
		&         & 				& 1 & 2 & 3 & 4 & 5 & 6 & 7 & 8 & 9 & 10 \\
\midrule	
500 & $\mathcal{P}_n^{(D)}$	& 1\%		  &	 3.8   &    8.3  &  11.0 &   13.0 &   14.3 &   14.7 &   14.8  &   14.9  &  14.9  &  14.7	\\
		&          							& 5\%	    &	12.1   &   15.6  &  20.1 &   22.2 &   23.9 &   25.0 &   24.9  &   25.0  &  24.6  &  24.0  \\
		&          							& 10\%    &	20.7   &   21.7  &  26.8 &   30.6 &   31.9 &   32.7 &   32.9  &   32.4  &  31.6  &  31.1  \\[0.25ex]
			\cline{2-13}\\[-2.25ex]                                                                                           
		& $\mathcal{F}_n^{(D)}$	& 1\%		  &	 5.3   &    8.5  &  11.5 &   13.8 &   15.3 &   15.9 &   15.9  &   15.6  &  15.5  &  15.1	\\
		&          							& 5\%	    &	11.2   &   16.4  &  20.6 &   23.7 &   25.2 &   26.1 &   26.1  &   25.8  &  25.0  &  24.3  \\
		&          							& 10\%    &	17.0   &   23.4  &  28.0 &   31.0 &   32.8 &   33.5 &   33.4  &   33.5  &  32.7  &  31.6  \\[0.25ex]
			\cline{2-13}\\[-2.25ex]                                                                                           
		& $\LB_n^{(D)}$					& 1\%		  &	 4.3   &    8.0  &  10.6 &   12.5 &   14.1 &   14.8 &   15.1  &   15.4  &  15.9  &  15.5	\\
		&          							& 5\%	    &	 7.3   &   12.3  &  15.2 &   17.6 &   19.5 &   20.8 &   20.9  &   20.8  &  21.1  &  20.9  \\
		&          							& 10\%    &	 9.6   &   15.1  &  18.9 &   21.2 &   22.9 &   24.1 &   24.5  &   24.7  &  24.7  &  24.2  \\
		\midrule                                                                                                            
1000& $\mathcal{P}_n^{(D)}$	& 1\%		  &	 9.5   &   13.8  &  19.9 &   24.0 &   27.4 &   29.8 &   31.2  &   31.4  &  31.6  &  31.0	\\
		&          							& 5\%	    &	19.1   &   27.0  &  34.1 &   39.0 &   43.5 &   46.2 &   47.0  &   47.3  &  47.6  &  46.9  \\
		&          							& 10\%    &	27.2   &   34.1  &  43.0 &   49.1 &   53.3 &   55.6 &   56.5  &   56.8  &  56.8  &  56.6  \\[0.25ex]
			\cline{2-13}\\[-2.25ex]                                                                                           
		& $\mathcal{F}_n^{(D)}$	& 1\%		  &	10.0   &   16.9  &  23.9 &   29.9 &   33.7 &   35.9 &   37.6  &   37.7  &  37.4  &  37.1	\\
		&          							& 5\%	    &	20.0   &   30.4  &  38.7 &   44.9 &   49.4 &   52.3 &   53.2  &   53.9  &  53.9  &  53.4  \\
		&          							& 10\%    &	27.5   &   39.0  &  48.5 &   55.5 &   59.2 &   61.1 &   62.1  &   62.2  &  62.3  &  61.8  \\[0.25ex]
			\cline{2-13}\\[-2.25ex]                                                                                           
		& $\LB_n^{(D)}$					& 1\%		  &	 5.9   &   10.3  &  13.7 &   15.6 &   17.3 &   18.3 &   18.8  &   19.2  &  19.1  &  19.0	\\
		&          							& 5\%	    &	 9.7   &   15.2  &  19.4 &   22.4 &   23.6 &   25.0 &   25.7  &   25.9  &  26.0  &  25.3  \\
		&          							& 10\%    &	12.4   &   18.7  &  23.3 &   26.1 &   27.7 &   28.8 &   29.2  &   29.7  &  29.4  &  29.1  \\
		\midrule                                                                                                            
2000& $\mathcal{P}_n^{(D)}$	& 1\%		  &	14.2   &   27.3  &  39.8 &   49.0 &   55.8 &   60.5 &   62.9  &   64.0  &  64.6  &  65.5	\\
		&          							& 5\%	    &	26.7   &   44.4  &  58.1 &   67.1 &   73.1 &   76.3 &   78.2  &   79.2  &  79.5  &  79.5  \\
		&          							& 10\%    &	33.0   &   54.4  &  67.4 &   75.3 &   80.9 &   83.3 &   84.7  &   85.6  &  85.9  &  85.2  \\[0.25ex]
			\cline{2-13}\\[-2.25ex]                                                                                           
		& $\mathcal{F}_n^{(D)}$	& 1\%		  &	18.8   &   35.1  &  50.2 &   61.4 &   68.6 &   72.5 &   74.7  &   75.5  &  76.2  &  76.3	\\
		&          							& 5\%	    &	31.5   &   52.3  &  67.7 &   77.2 &   82.5 &   85.4 &   87.2  &   87.2  &  87.4  &  87.1  \\
		&          							& 10\%    &	41.6   &   61.8  &  76.3 &   84.2 &   88.6 &   90.4 &   91.2  &   92.0  &  92.1  &  91.6  \\[0.25ex]
			\cline{2-13}\\[-2.25ex]                                                                                           
		& $\LB_n^{(D)}$					& 1\%		  &	06.7   &   12.9  &  16.2 &   18.2 &   19.7 &   20.4 &   21.0  &   21.1  &  21.2  &  20.6	\\
		&          							& 5\%	    &	10.4   &   18.2  &  22.4 &   24.2 &   25.3 &   26.0 &   26.7  &   26.9  &  26.9  &  26.5  \\
		&          							& 10\%    &	13.3   &   21.9  &  26.6 &   28.5 &   29.8 &   29.9 &   30.2  &   30.5  &  30.1  &  29.7  \\
	\bottomrule
		\end{tabular}
	\end{center}
\caption{\label{tab:power}Power in \% of tests based on $\mathcal{P}_n^{(D)}$, $\mathcal{F}_n^{(D)}$ and $\LB_n^{(D)}$ for significance levels $\alpha\in\{1\%, 5\%, 10\%\}$ and $D=1,\ldots,10$.  Results are for model \eqref{eq:APARCH-X(1,1)} with $\vtheta^{\circ}=\vtheta^{\circ}_{p}=(0.046, 0.027, 0.092, 0.843, 0.089)^\prime$.}
\end{table}

Comparing our $\mathcal{P}_n^{(D)}$- and $\mathcal{F}_n^{(D)}$-test with the Ljung--Box test, we find that all tests are comparable in their sensitivity to the choice of $D$. Furthermore, while our tests may have inferior size at the 5\%- and 10\%-level for small samples, size tends to be better overall for $n=2000$. E.g., for the $\LB_n^{(D)}$-test, the empirical rejection frequencies at the 1\%-level vary between 1.6\% an 4.3\% for $n=2000$, while the fluctuations for the $\mathcal{F}_n^{(D)}$-test are much smaller (size between 0.8\% and 1.3\%). The fact that size of our tests improves more for increasing $n$ than for the Ljung--Box test is as expected, because our tests are based on the \textit{tail} copula, thus, requiring reasonably large samples.

\begin{figure}[t!]
	\centering
		\includegraphics[width=\textwidth]{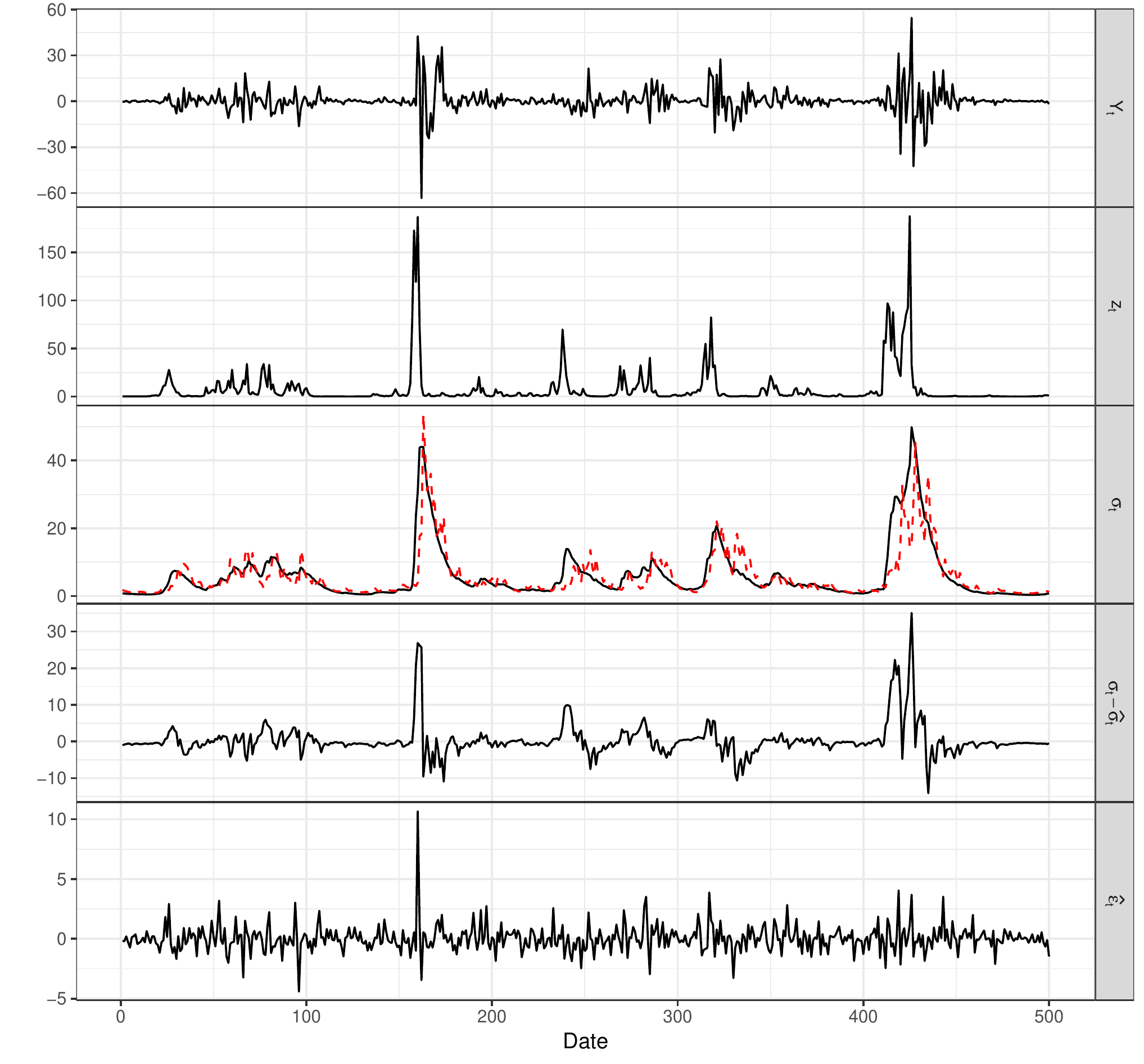}
	\caption{From top to bottom: $Y_t$; $z_t$, $\sigma_t$ (black) and $\widehat{\sigma}_t$ (red, dashed); $\sigma_t-\widehat{\sigma}_t$; $\widehat{\varepsilon}_t$.}
	\label{fig:SimTS}
\end{figure}

Now, we turn to a comparison of power by simulating from \eqref{eq:APARCH-X(1,1)} with $\vtheta^{\circ}=\vtheta^{\circ}_{p}$. Table~\ref{tab:power} displays the results. As expected, power increases for all tests the larger the sample. For almost all sample sizes and choices of $D$, our tests are more likely than the Ljung--Box test to signal a misspecified model. This difference in power increases markedly in $n$, because the power increase in $n$ is rather modest for $\LB_n^{(D)}$. Table~\ref{tab:power} reveals larger differences in empirical rejection frequencies across different $D$ than Table~\ref{tab:size}. Irrespective of the sample size, the power of our tests reaches the maximum for $D\approx7$. Comparing $\mathcal{P}_n^{(D)}$ and $\mathcal{F}_n^{(D)}$, we find that clear differences in power favoring $\mathcal{F}_n^{(D)}$ only emerge in large samples.

\begin{figure}[t!]
	\centering
		\includegraphics[width=\textwidth]{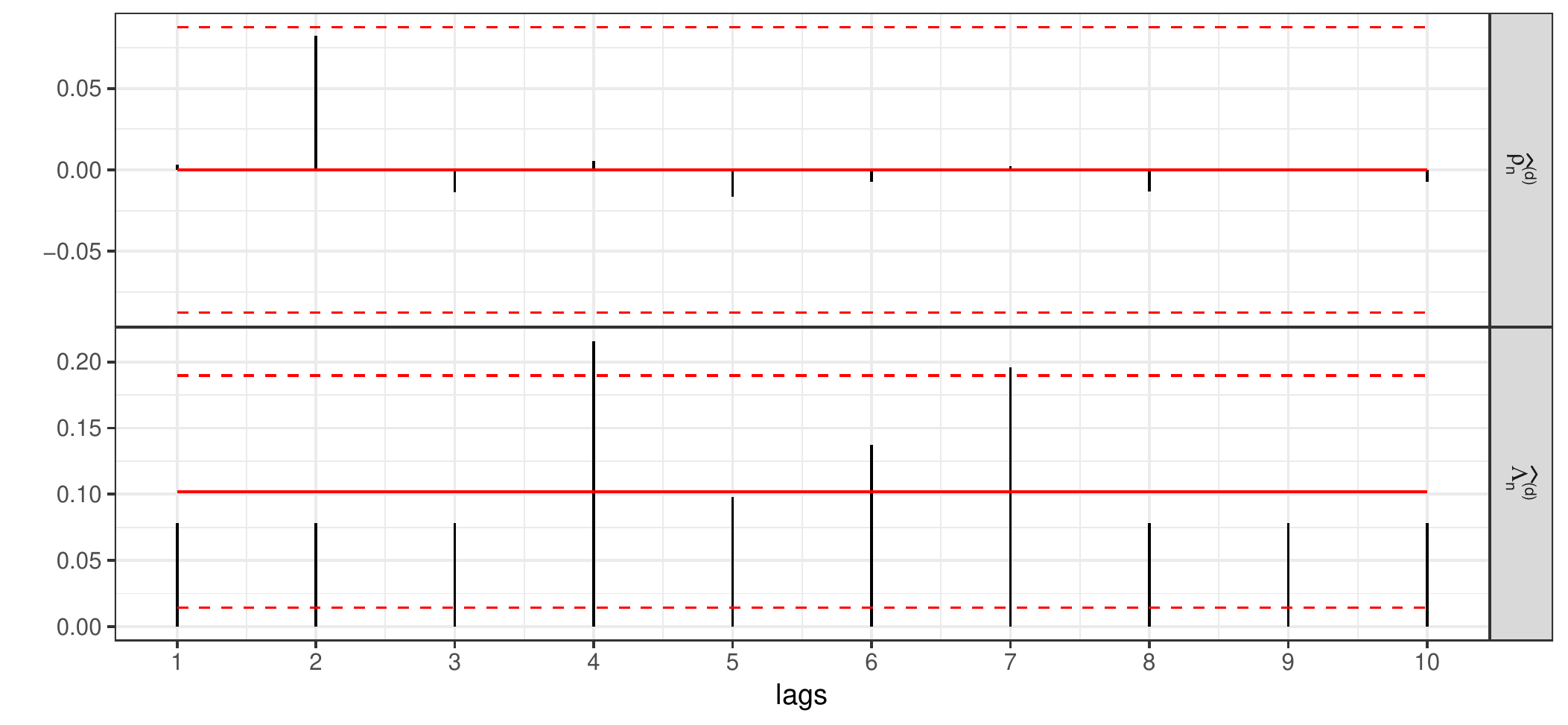}
	\caption{Top: Estimates $\widehat{\rho}_n^{(d)}$ of the autocorrelation at lag $d$; bottom: Estimates $\widehat{\Lambda}_n^{(d)}(1,1)$.}
	\label{fig:ACF}
\end{figure}

Finally, we illustrate the difference between $\LB_n^{(D)}$ and our tests now for a representative sample from \eqref{eq:APARCH-X(1,1)} with $\vtheta^{\circ}=\vtheta^{\circ}_{p}$ and $n=500$. Figure~\ref{fig:SimTS} shows the time series $Y_t$. True volatility $\sigma_t=\sigma_t(\vtheta^{\circ})$ and estimated volatility $\widehat{\sigma}_t$ (estimated based on the misspecified APARCH(1,1) model) are shown in the middle panel, with the differences $\sigma_t-\widehat{\sigma}_t$ plotted below. While volatility seems to be captured well by the misspecified model during calm phases, this is not the case in turbulent periods. This, in turn, induces some serial dependence only in the extreme regions of the distribution of $\widehat{\varepsilon}_t$.

To see this more clearly, consider Figure~\ref{fig:ACF}. It shows lag $d$-estimates of the squared residual autocorrelations, $\widehat{\rho}_n^{(d)}$, and the tail copula, $\widehat{\Lambda}_n^{(d)}(1,1)$.\footnote{We have excluded the estimates $\widehat{\rho}_n^{(d)}=\widehat{\Lambda}_n^{(d)}(1,1)=1$ for $d=0$ to zoom in on the relevant autocorrelations in Figure~\ref{fig:ACF}.} The pointwise 95\%-confidence intervals are indicated by the red dashed lines in both panels. None of the autocorrelations are significant at the 5\%-level, yet the $\widehat{\Lambda}_n^{(d)}$ suggest two significant lags of the tail copula. Thus, while the linear dependence in the body of the distribution appears to be negligible, the misspecified volatility estimates nonetheless induce some serial \textit{extremal} dependence in the filtered residuals.

\subsection{Misspecified Innovations}\label{App: Misspecified Innovations}

\begin{table}[t!]
	\begin{center}
		\begin{tabular}{lllrrrrrrrrrr}
			\toprule
$n$	& Test    &	$\alpha$&\multicolumn{10}{c}{$D$}  \\[0.25ex]
\cline{4-13}\\[-2.25ex] 
		&         & 				& 1 & 2 & 3 & 4 & 5 & 6 & 7 & 8 & 9 & 10 \\
\midrule	
500 & $\mathcal{P}_n^{(D)}$	& 1\%		  &	 1.4  &   2.9  &    3.2  &    3.4 &  3.6  &    3.6  &   3.5 &   3.4  &  3.5  &   3.3	\\
		&          							& 5\%	    &	 5.6  &   5.7  &    6.2  &    6.4 &  6.6  &    6.7  &   6.7 &   6.5  &  6.5  &   6.4  \\
		&          							& 10\%    &	12.0  &   9.3  &    9.7  &    9.8 &  9.8  &   10.0  &  10.0 &   9.5  &  9.3  &   9.1  \\[0.25ex]
			\cline{2-13}\\[-2.25ex]           
		& $\mathcal{F}_n^{(D)}$	& 1\%		  &	 2.4  &   3.2  &    3.9  &    4.1 &  4.2  &    4.2  &   4.1 &   4.1  &  4.0  &   3.8	\\
		&          							& 5\%	    &	 5.8  &   6.2  &    6.3  &    6.4 &  6.6  &    6.6  &   6.7 &   6.5  &  6.4  &   6.0  \\
		&          							& 10\%    &	 9.4  &   9.7  &    9.3  &    9.3 &  9.2  &    8.7  &   8.8 &   8.4  &  8.4  &   8.0  \\[0.25ex]
			\cline{2-13}\\[-2.25ex]           
		& $\LB_n^{(D)}$					& 1\%		  &	 3.2  &   5.2  &    5.5  &    5.8 &  5.9  &    5.6  &   5.5 &   5.7  &  5.7  &   5.3	\\
		&          							& 5\%	    &	 6.0  &   8.3  &    8.6  &    9.1 &  8.8  &    8.7  &   8.4 &   8.3  &  8.5  &   7.9  \\
		&          							& 10\%    &	 8.3  &  10.7  &   11.4  &   11.3 & 11.3  &   10.7  &  10.4 &  10.4  & 10.5  &   9.9  \\
		\midrule                            
1000& $\mathcal{P}_n^{(D)}$	& 1\%		  &	 1.8  &   1.9  &    2.3  &    2.5 &  2.7  &    2.5  &   2.6 &   2.6  &  2.5  &   2.4  \\
		&          							& 5\%	    &	 6.1  &   4.8  &    5.5  &    5.4 &  5.5  &    5.5  &   5.6 &   5.4  &  5.5  &   5.4  \\
		&          							& 10\%    &	10.9  &   9.1  &    8.5  &    8.9 &  8.5  &    8.6  &   8.7 &   8.8  &  8.6  &   8.3  \\[0.25ex]
			\cline{2-13}\\[-2.25ex]           
		& $\mathcal{F}_n^{(D)}$	& 1\%		  &	 1.9  &   2.4  &    2.6  &    2.7 &  2.8  &    2.8  &   2.7 &   2.7  &  2.7  &   2.6  \\
		&          							& 5\%	    &	 5.0  &   5.5  &    5.6  &    5.3 &  5.1  &    4.7  &   4.8 &   4.8  &  4.9  &   4.8  \\
		&          							& 10\%    &	 9.0  &   9.0  &    8.9  &    8.3 &  7.9  &    7.2  &   7.1 &   7.1  &  7.2  &   6.8  \\[0.25ex]
			\cline{2-13}\\[-2.25ex]             
		& $\LB_n^{(D)}$					& 1\%		  &	 2.3  &   4.4  &    5.1  &    5.2 &  4.7  &    4.5  &   4.6 &   4.7  &  5.2  &   5.1  \\
		&          							& 5\%	    &	 4.2  &   7.3  &    7.7  &    7.3 &  6.3  &    6.3  &   6.7 &   6.8  &  7.3  &   7.3  \\
		&          							& 10\%    &	 5.7  &   9.3  &    9.7  &    8.8 &  8.2  &    7.8  &   8.1 &   8.4  &  9.0  &   9.1  \\
		\midrule                              
2000& $\mathcal{P}_n^{(D)}$	& 1\%		  &	 1.1  &   1.6  &    1.7  &    1.7 &  1.7  &    1.5  &   1.6 &   1.7  &  1.7  &   1.5  \\
		&          							& 5\%	    &	 4.6  &   5.2  &    5.0  &    4.7 &  4.5  &    4.4  &   4.2 &   4.1  &  4.0  &   4.0   \\
		&          							& 10\%    &	 9.7  &   9.2  &    9.1  &    8.7 &  8.3  &    7.5  &   7.6 &   7.4  &  7.2  &   7.0   \\[0.25ex]
			\cline{2-13}\\[-2.25ex]             
		& $\mathcal{F}_n^{(D)}$	& 1\%		  &	 1.8  &   2.0  &    2.0  &    2.0 &  2.1  &    2.0  &   1.9 &   1.9  &  1.8  &   1.8  \\
		&          							& 5\%	    &	 5.9  &   5.6  &    5.3  &    4.8 &  5.0  &    4.4  &   4.2 &   4.1  &  4.0  &   3.7   \\
		&          							& 10\%    &	10.2  &   9.4  &    9.2  &    8.2 &  7.8  &    7.3  &   6.8 &   6.3  &  6.3  &   5.9   \\[0.25ex]
			\cline{2-13}\\[-2.25ex]             
		& $\LB_n^{(D)}$					& 1\%		  &	 1.3  &   3.6  &    4.9  &    4.6 &  4.4  &    3.8  &   3.9 &   3.8  &  4.2  &   4.1  \\
		&          							& 5\%	    &	 2.8  &   5.6  &    7.3  &    6.7 &  6.3  &    5.4  &   5.8 &   5.6  &  6.2  &   6.1   \\
		&          							& 10\%    &	 4.1  &   7.2  &    8.9  &    8.5 &  7.7  &    6.7  &   7.0 &   7.0  &  7.6  &   7.5   \\
	\bottomrule
		\end{tabular}
	\end{center}
\caption{\label{tab:size2} Size in \% of tests based on $\mathcal{P}_n^{(D)}$, $\mathcal{F}_n^{(D)}$ and $\LB_n^{(D)}$ for significance levels $\alpha\in\{1\%, 5\%, 10\%\}$ and $D=1,\ldots,10$. Results are for model \eqref{eq:GARCH(1,1)} with $(a_1,b_1,c_1)^\prime=(-3, 0, 0)^\prime$ and $(a_2,b_2,c_2)^\prime=(-1, 0, 0)^\prime$.}
\end{table}

Now, we reconsider the setup of Section~\ref{Misspecified Innovations} in the main paper with fixed $k=\lfloor 0.11\cdot n^{0.99}\rfloor$ and varying $D\in\{1,\ldots,10\}$. We first report size, where $(a_1,b_1,c_1)^\prime=(-3, 0, 0)^\prime$ and $(a_2,b_2,c_2)^\prime=(-1, 0, 0)^\prime$ for model \eqref{eq:GARCH(1,1)}. Table~\ref{tab:size2} displays the tests' size, together with that of the estimation effects-corrected Ljung--Box test \citep[Theorem 8.2]{FZ10}. Except perhaps for small samples and $\alpha=1\%$, the size of our tests is very good, and generally better than that of the Ljung--Box test. Moreover, the size of our tests is also very stable across different $D$.

\begin{table}[t!]
	\begin{center}
		\begin{tabular}{lllrrrrrrrrrr}
			\toprule
$n$	& Test    &	$\alpha$&\multicolumn{10}{c}{$D$}  \\[0.25ex]
\cline{4-13}\\[-2.25ex] 
		&         & 				& 1 & 2 & 3 & 4 & 5 & 6 & 7 & 8 & 9 & 10 \\
\midrule	
500 & $\mathcal{P}_n^{(D)}$	& 1\%		  &	47.2   &  51.8  &   50.8 &   48.6  &   47.0  &   45.0  &   44.0 &   42.6 &   40.8  &   39.7	\\
		&          							& 5\%	    &	63.2   &  61.8  &   61.5 &   59.0  &   56.8  &   55.5  &   53.8 &   52.8 &   51.2  &   50.1  \\
		&          							& 10\%    &	71.0   &  67.6  &   66.6 &   65.3  &   63.6  &   61.7  &   59.9 &   58.4 &   57.1  &   55.8  \\[0.25ex]
			\cline{2-13}\\[-2.25ex]           
		& $\mathcal{F}_n^{(D)}$	& 1\%		  &	55.9   &  56.9  &   55.1 &   53.1  &   50.9  &   49.1  &   47.6 &   46.2 &   44.7  &   43.0	\\
		&          							& 5\%	    &	66.6   &  67.0  &   65.2 &   62.9  &   60.7  &   59.0  &   57.3 &   55.5 &   54.1  &   52.6  \\
		&          							& 10\%    &	72.0   &  72.3  &   70.7 &   68.7  &   66.2  &   64.4  &   62.8 &   61.1 &   59.7  &   58.2  \\[0.25ex]
			\cline{2-13}\\[-2.25ex]           
		& $\LB_n^{(D)}$					& 1\%		  &	13.5   &  15.5  &   15.0 &   14.1  &   13.1  &   12.4  &   11.3 &   10.3 &    9.9  &    9.8	\\
		&          							& 5\%	    &	16.9   &  18.0  &   17.0 &   16.3  &   15.4  &   14.7  &   13.7 &   12.9 &   12.3  &   12.2  \\
		&          							& 10\%    &	19.7   &  19.7  &   18.4 &   17.2  &   16.6  &   15.9  &   14.9 &   14.3 &   14.0  &   13.8  \\
		\midrule                             
1000& $\mathcal{P}_n^{(D)}$	& 1\%		  &	67.5   &  68.0  &   67.9 &   66.2  &   64.6  &   62.5  &   60.7 &   58.8 &   57.4  &   56.2	\\
		&          							& 5\%	    &	77.1   &  78.4  &   77.8 &   76.2  &   74.8  &   73.1  &   71.7 &   70.5 &   69.4  &   67.9  \\
		&          							& 10\%    &	81.6   &  82.5  &   82.5 &   81.4  &   79.7  &   78.4  &   76.9 &   75.8 &   74.7  &   73.7  \\[0.25ex]
			\cline{2-13}\\[-2.25ex]           
		& $\mathcal{F}_n^{(D)}$	& 1\%		  &	72.5   &  74.7  &   73.4 &   71.5  &   70.0  &   68.2  &   66.3 &   64.9 &   63.4  &   62.2	\\
		&          							& 5\%	    &	80.5   &  82.3  &   81.8 &   80.7  &   78.9  &   77.5  &   75.7 &   74.5 &   73.5  &   72.6  \\
		&          							& 10\%    &	83.6   &  85.6  &   85.3 &   84.6  &   83.3  &   81.6  &   80.5 &   79.1 &   78.0  &   76.8  \\[0.25ex]
			\cline{2-13}\\[-2.25ex]           
		& $\LB_n^{(D)}$					& 1\%		  &	 6.2   &   7.5  &    7.7 &    7.7  &    7.5  &    7.2  &    7.1 &    6.9 &    6.7  &    6.6	\\
		&          							& 5\%	    &	 8.1   &   9.0  &    8.9 &    8.9  &    8.7  &    8.5  &    8.3 &    8.1 &    7.9  &    7.7  \\
		&          							& 10\%    &	 9.3   &   9.8  &    9.7 &    9.4  &    9.2  &    9.2  &    9.0 &    8.8 &    8.7  &    8.5  \\
		\midrule                            
2000& $\mathcal{P}_n^{(D)}$	& 1\%		  &	84.9   &  87.8  &   88.0 &   87.1  &   86.2  &   85.5  &   84.8 &   83.8 &   82.7  &   81.9	\\
		&          							& 5\%	    &	90.6   &  92.6  &   92.4 &   92.4  &   91.7  &   90.8  &   90.1 &   89.4 &   88.9  &   88.4  \\
		&          							& 10\%    &	91.9   &  94.4  &   94.6 &   94.3  &   93.8  &   93.5  &   92.5 &   92.2 &   91.6  &   90.8  \\[0.25ex]
			\cline{2-13}\\[-2.25ex]           
		& $\mathcal{F}_n^{(D)}$	& 1\%		  &	88.9   &  90.6  &   90.6 &   90.1  &   89.4  &   88.4  &   87.9 &   87.4 &   86.9  &   86.7	\\
		&          							& 5\%	    &	92.5   &  93.9  &   94.1 &   93.9  &   93.4  &   92.6  &   92.2 &   92.0 &   91.4  &   90.9  \\
		&          							& 10\%    &	93.9   &  95.6  &   95.6 &   95.4  &   95.3  &   94.5  &   93.9 &   93.8 &   93.5  &   93.0  \\[0.25ex]
			\cline{2-13}\\[-2.25ex]           
		& $\LB_n^{(D)}$					& 1\%		  &	 1.8   &   2.4  &    2.5 &    2.3  &    2.4  &    2.5  &    2.7 &    2.7 &    2.8  &    3.0	\\
		&          							& 5\%	    &	 2.5   &   2.9  &    3.2 &    3.1  &    3.1  &    3.0  &    3.1 &    3.1 &    3.2  &    3.6  \\
		&          							& 10\%    &	 3.0   &   3.7  &    3.6 &    3.5  &    3.5  &    3.4  &    3.6 &    3.5 &    3.5  &    3.9  \\
	\bottomrule
		\end{tabular}
	\end{center}
\caption{\label{tab:power2}Power in \% of tests based on $\mathcal{P}_n^{(D)}$, $\mathcal{F}_n^{(D)}$ and $\LB_n^{(D)}$ for significance levels $\alpha\in\{1\%, 5\%, 10\%\}$ and $D=1,\ldots,10$. Results are for model \eqref{eq:GARCH(1,1)} with $(a_1,b_1,c_1)^\prime=(-3, -6, 0.6)^\prime$ and $(a_2,b_2,c_2)^\prime=(-1, -2, 0.6)^\prime$.}
\end{table}

Now, we turn to a comparison of power by simulating from \eqref{eq:GARCH(1,1)} with $(a_1,b_1,c_1)^\prime=(-3, -6, 0.6)^\prime$ and $(a_2,b_2,c_2)^\prime=(-1, -2, 0.6)^\prime$. Table~\ref{tab:power2} displays the results. Once more our tests are more likely than the Ljung--Box test to signal a misspecified model for all sample sizes and choices of $D$. This difference in power increases markedly in $n$, because the $\LB_n^{(D)}$-test even \textit{loses} power for increasing $n$. Table~\ref{tab:power2} reveals larger differences in empirical rejection frequencies across different $D$ than Table~\ref{tab:size}. There is a tendency for power to decrease for larger $D$. Yet, that power decrease is only small up until $D=5$. Comparing $\mathcal{P}_n^{(D)}$ and $\mathcal{F}_n^{(D)}$, we find that there are only minor differences favoring $\mathcal{F}_n^{(D)}$.

Summing up the results of Appendices~\ref{App: Misspecified Volatility} and~\ref{App: Misspecified Innovations}, we recommend the $\mathcal{F}_n^{(D)}$-test with $D=5$, due to its good size and power. By virtue of its simplicity, we can, however, also recommend the $\mathcal{P}_n^{(D)}$-based test as a viable alternative. Although the choice $D=5$ leads to good results in both simulation setups, we recommend---as is common practice for other tests---to also report results for other choices of $D$.


\end{appendices}

\end{document}